\documentclass[aps,prl,twocolumn,nofootinbib,superscriptaddress,10pt]{revtex4-2}

\usepackage{graphicx,color,amsmath,amsfonts,enumerate,amsthm,amssymb,mathtools,enumitem,thmtools,hyperref,subfigure,mathdots,enumitem,centernot,bm,soul,bbm,stmaryrd}
\usepackage[capitalise, noabbrev]{cleveref}
\usepackage{multirow}
\usepackage[figuresright]{rotating}

\usepackage{tikz,pgfplots}
\pgfplotsset{compat=1.17}
\usetikzlibrary{quantikz2}
\hypersetup{colorlinks=true,linkcolor=blue,citecolor=blue,urlcolor=blue, hypertexnames=false}

\usepackage{algorithm,algpseudocode}
\usepackage{algorithmicx}
\usepackage{mathrsfs}
\usepackage{soul}

\usepackage{mathtools,mleftright}
\mleftright
\let\originalleft\left
\let\originalright\right
\renewcommand{\left}{\mathopen{}\mathclose\bgroup\originalleft}
\renewcommand{\right}{\aftergroup\egroup\originalright}

\def\A{ {\cal A} }

\def\E{ {\mathbb E} } 
 
\def\G{ {\cal G} }

\def\R{ \mathbb{R} }

\def\>{\rangle}
\def\<{\langle}
 
\newcommand{\ketbra}[2]{\ensuremath{\left|#1\right\rangle\!\left\langle#2\right|}}

\newcommand{\tr}[1]{\mathrm{tr}\left( #1 \right)}

\newcommand{\diag}[1]{\mathrm{diag}\left( #1 \right)}

\DeclareMathOperator{\Var}{Var}

\DeclareMathOperator{\Tr}{tr}

\newcommand{\bbc}{\mathbb{C}}

\newcommand{\bbz}{\mathbb{Z}}
\newcommand{\hphide}[1]{}

\newcommand{\bigo}[1]{\mathcal{O}\left(#1\right)}

\newcommand{\sym}{\Pi_{\operatorname{sym}}}

\theoremstyle{plain}
\newtheorem{thm}{Theorem}
\newtheorem{theorem}[thm]{Theorem}

\newtheorem{lemma}{Lemma}
\newtheorem{prop}{Proposition}

\theoremstyle{definition} 
\newtheorem{defn}{Definition}

\newcommand{\PT}{\mathrm{PT}}

\newcommand{\rmi}{\mathrm{i}}

\newcommand{\rme}{\operatorname{e}}

\newcommand{\caH}{\mathcal{H}}

\newcommand{\caB}{\mathcal{B}}
\newcommand{\caD}{\mathcal{D}}

\newcommand{\Wg}{\mathrm{Wg}}
\usepackage{graphicx}
\usepackage{bm}
\usepackage{amsmath}
\usepackage{lipsum}
\usepackage{balance}
\usepackage{soul}

\usetikzlibrary{positioning}
\usetikzlibrary{snakes}
\usetikzlibrary{decorations}
\definecolor{evergreen}{rgb}{0.27,0.62,0.20}

\tikzstyle{tensor_blue}  =[rectangle,draw=black,fill=blue!25,       thick,minimum size=0.6cm]
\tikzstyle{tensor_green} =[rectangle,draw=black,fill=green!20,      thick,minimum size=0.6cm]
\tikzstyle{tensor_purple}=[rectangle,draw=black,fill=blue!50!red!50,thick,minimum size=0.6cm]

\tikzstyle{Permute_2} = [rectangle, draw=black, thick, fill=red!25, minimum width=1.8cm,minimum height = 0.6cm]
\tikzstyle{Permute_3} = [rectangle, draw=black, thick, fill=red!25, minimum width=3.0cm,minimum height = 0.6cm]
\tikzstyle{Permute_4} = [rectangle, draw=black, thick, fill=red!25, minimum width=4.2cm,minimum height = 0.6cm]
\tikzstyle{Permute_5} = [rectangle, draw=black, thick, fill=red!25, minimum width=5.4cm,minimum height = 0.6cm]

\tikzstyle{UVU_6} = [rectangle, draw=black, thick, fill=blue!25, minimum width=2.1cm,minimum height = 0.6cm]
\tikzstyle{2party} = [rectangle, draw=black, thick, fill=green!20, minimum width=1.3cm,minimum height = 0.6cm]
\tikzstyle{terminal}  = [=, thick, minimum width=0.3cm, minimum height = 0.2cm]

\newcommand{\RN}[1]{%
	\textup{\uppercase\expandafter{\romannumeral#1}}%
}
\renewcommand{\thetable}{\Alph{section}\arabic{table}}

\definecolor{dullblue}{rgb}{.29,.47,.77}

\def\eqref#1{\textup{(\ref{#1})}}
\newcommand{\eref}[1]{Eq.~\textup{(\ref{#1})}}
\newcommand{\lref}[1]{Lemma~\ref{#1}}
\newcommand{\tref}[1]{Theorem~\ref{#1}}
\newcommand{\pref}[1]{Proposition~\ref{#1}}

\usepackage{orcidlink}

\begin{document}

\title{Quantum Nonlinear Properties from a Single Measurement Setting}

\author{Zihao Li}
\email{zihaoli@hku.hk}
\thanks{equal contribution}
\affiliation{QICI Quantum Information and Computation Initiative, School of Computing and Data Science, The University of Hong Kong, Pokfulam Road, Hong Kong, China}

\author{Datong Chen}
\thanks{equal contribution}
\affiliation{State Key Laboratory of Surface Physics, Department of Physics, and Center for Field Theory and Particle Physics, Fudan University, Shanghai 200433, China}
\affiliation{Institute for Nanoelectronic Devices and Quantum Computing, Fudan University, Shanghai 200433, China}
\affiliation{Shanghai Research Center for Quantum Sciences, Shanghai 201315, China}

\author{Dayue Qin}
\affiliation{Key Laboratory for Information Science of Electromagnetic Waves (Ministry of Education), Fudan University, Shanghai 200433, China}

\author{Yuxiang Yang}
\email{yuxiang@cs.hku.hk}
\affiliation{QICI Quantum Information and Computation Initiative, School of Computing and Data Science, The University of Hong Kong, Pokfulam Road, Hong Kong, China}

\author{You Zhou}
\email{you\_zhou@fudan.edu.cn}
\affiliation{Key Laboratory for Information Science of Electromagnetic Waves (Ministry of Education), Fudan University, Shanghai 200433, China}

\date{\today}

\begin{abstract}
Nonlinear properties of quantum states are essential to quantum information and many-body physics, but assessing them experimentally is challenging, as it typically requires multi-copy operations or a large number of measurement settings. To address this challenge, we develop a universal framework, collision-based nonlinear estimation (CBNE), for efficiently measuring nonlinear quantities of a quantum state $\rho$, such as the higher-order expectation value \(\Tr(O\rho^t)\) for some observable $O$, using single-copy randomized measurements. Strikingly, our protocol requires only a single measurement setting, provided that the system dimension is sufficiently large or a few ancillary qubits are available; this contrasts with the conventional expectation that multiple measurement bases are necessary for nonlinear estimation. In addition, CBNE is observable-independent at the experimental stage, which enables simultaneous estimation of multiple nonlinear functions. It further extends to broader tasks, including the estimation of principal component properties and partial-transpose moments of quantum states. Our results provide a practical and scalable route for measuring nonlinear state properties on near-term quantum devices.
\end{abstract}

\maketitle  

\def \loopwidth {1em}
\def \colspacing {2em}
\def \rowspacing {1.5em}
\setlength\tabcolsep{20pt}
\newsavebox{\composition}
\newsavebox{\tensorproduct}
\newsavebox{\traceytrace}

\let\oldaddcontentsline\addcontentsline
\renewcommand{\addcontentsline}[3]{}

Nonlinear functions of quantum systems, such as the state moments $\Tr(\rho^t)$ and higher-order expectation values $\Tr(O\rho^t)$, play central roles in areas ranging from quantum information processing to many-body physics. 
They underpin numerous applications including quantum error mitigation 
\cite{PhysRevX.11.041036,PhysRevX.11.031057,RevModPhys.95.045005,PRXQuantum.4.010303,zhou2024hybrid,PhysRevLett.133.080601,o2023purification}, quantum virtual cooling (QVC) 
\cite{PhysRevX.9.031013,PhysRevLett.133.080601}, entanglement spectroscopy 
\cite{PhysRevB.96.195136,Yirka2021qubitefficient}, and quantum entropy estimation \cite{islam2015measuring,brydges2019probing,PhysRevLett.120.050406,PhysRevA.97.023604}. 
Beyond these, the \emph{partial-transpose} (PT) moments of quantum states 
provide powerful criteria for entanglement detection \cite{PhysRevLett.121.150503,elben2020mixedstate,singlezhou,Neven2021Entanglement,Yu2021Entanglement} and serve as useful probes for quantum phases of matter and entanglement dynamics in open systems \cite{PhysRevLett.109.130502,PhysRevB.94.035152,PhysRevLett.125.140603}. 
They also enable access to the entanglement negativity \cite{PhysRevLett.95.090503,PhysRevA.65.032314,PhysRevA.109.012422,carteret2016estimating}, a prominent entanglement measure for mixed quantum states.

Conventional approaches for accessing these nonlinear quantities, such as generalized swap tests \cite{PhysRevLett.88.217901,bruni2004measurimg,Quek2024multivariatetrace}, require collective operations on multiple copies of the state prepared in parallel \cite{liu2025estimating,liu2024auxiliary,Shin2025resourceefficient,PhysRevLett.94.040502,PhysRevLett.129.260501,zhang2025measuring,chen2025simultaneous,shi2025near,ye2025exponential}, which are experimentally demanding. 
More recently, \emph{randomized measurement} (RM) protocols and the framework of \emph{classical shadow estimation} (CSE) have emerged as powerful single-copy alternatives \cite{PhysRevLett.108.110503,HKPshadow20,PhysRevX.14.031035,PhysRevLett.124.010504,elben2023randomized,PhysRevA.99.052323,PhysRevLett.108.110503,PhysRevResearch.5.023027,PRXQuantum.5.010352,arienzo2023closed,PhysRevLett.130.230403,Brieger25,gresch2025guaranteed,west2025real,wu2024error,PhysRevLett.134.110802,wang2026superiority}.
By applying random circuits followed by fixed-basis measurements, these methods enable the estimation of nonlinear properties from the classical correlations between independent measurement runs \cite{Kliesch2021Certification,elben2023randomized}.

However, the implementation of these single-copy methods still incurs a substantial overhead in the number of random circuits (i.e., measurement settings) used. Specifically, to achieve an  estimation error $\epsilon$, the required setting number usually scales as $\sim 1/\epsilon^2$ and, for many tasks, can even grow exponentially with the system size \cite{HKPshadow20,PhysRevLett.127.260501,elben2020mixedstate,pelecanos2026beating,elben2023randomized}. 
This poses a significant practical bottleneck for near-term devices, as frequent switching between measurement settings is resource-intensive, and most hardware platforms prefer to operate
with minimal setting variations \cite{PhysRevLett.131.240602,Zhou2023perform,PRXQuantum.2.010307,huggins2022unbiasing,singlezhou}. 
Unfortunately, estimating nonlinear functions with a constant number of single-copy measurement settings is currently achievable only for second-order quantities, such as $\Tr(\rho^2)$ and $\Tr(O\rho^2)$ with specific structured observables $O$ \cite{PhysRevLett.124.010504,anshu2022distributed,du2025optimal}.

\begin{figure}[b]
\centering
\includegraphics[width=0.48\textwidth]{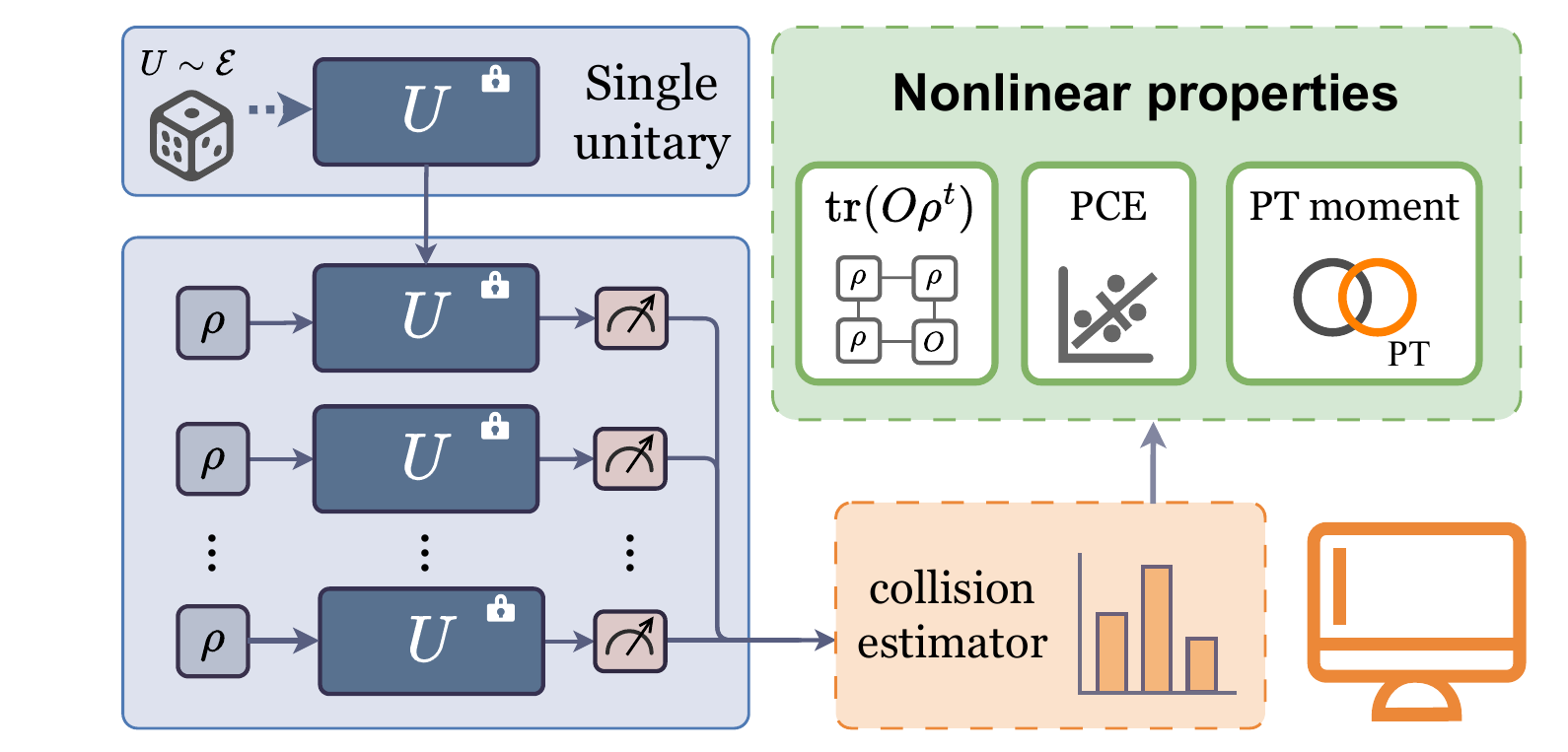}
\caption{Schematic illustration of the CBNE framework.
Left: experimental stage. A fixed random unitary is repeatedly applied to sequentially prepared copies of $\rho$, followed by projective measurements in the computational basis. 
Right: classical postprocessing. Collision estimators constructed from the measurement outcomes enable the estimation of nonlinear functions of $\rho$, including $\Tr(O\rho^t)$, $\Tr(\rho^t)$, principal component properties, PT moments, etc.
}
\label{fig:scheme}
\end{figure}

In this work, we overcome this limitation by introducing a universal and observable-independent framework, termed \emph{collision-based nonlinear estimation} (CBNE), for measuring nonlinear functions such as $\Tr(O\rho^t)$, including $\Tr(\rho^t)$, from single-copy RMs. 
As illustrated in Fig.~\ref{fig:scheme}, CBNE requires only \emph{a single} fixed random unitary followed by projective measurements, provided that either the system dimension is sufficiently large or a small number, $\bigo{\log\epsilon^{-1}}$, of ancillary qubits are accessible. 
This substantially reduces the experimental overhead compared to existing approaches (see Table~\ref{tab:compare} in the End Matter), suggesting that the conventional reliance on many measurement bases is not intrinsic.

Inheriting the mindset of CSE, CBNE enables the estimation of multiple nonlinear functions from the same measurement data through simple and efficient linear postprocessing based on collision statistics.
Moreover, it achieves the best-known sample cost for general orders $t\geq 3$ among existing single-copy measurement protocols.
For state moment estimation, it further achieves the optimal scaling in system dimension among all protocols that use a single measurement setting.  
Beyond these, CBNE can be applied to learn principal component properties of quantum states, and admits a natural generalization to measuring PT moments, thereby enabling entanglement detection with a single measurement setting.

\emph{Collision-based estimation from RMs}---
We first present the basic version of the CBNE protocol, leaving its variants and extensions to later sections.
Consider an unknown $n$-qubit quantum state $\rho$ on a Hilbert space $\caH$ that can be prepared repeatedly. As shown in Fig.~\ref{fig:scheme}, the experimental stage of CBNE builds upon RMs \cite{elben2023randomized} and proceeds as follows.
We apply to $\rho$ a random unitary $U$ drawn from some ensemble $\mathcal E$, then measure the state in the computational basis $\{|b\>\}_{b=0}^{d-1}$,  where $d=2^n$ is the system dimension. 
This process is repeated for $N_U$ distinct random unitaries, with $N_M$ projective measurements performed for each unitary.

To illustrate the key idea of CBNE, let us first consider the estimation of state moments $p_t = \Tr(\rho^t)$.
Let $\mathbf{b}_U=\{\hat{b}_1,\dots,\hat{b}_{N_M}\}\in[d]^{N_M}$ be the list of measurement outcomes collected under a fixed random unitary \(U\).
In postprocessing, we calculate the following estimator by counting the number of \( k \)-wise collisions among the elements in \( \mathbf{b}_U \)  (for each $k=2,\dots,t$): 
\begin{align}\label{eq:collision_moment_prl}
	\hat M_k^U
	:=
	\frac{\kappa_k}{d\binom{N_M}{k}}
	\sum_{i_1<i_2<\cdots<i_k}
	\mathbf{1}\{\hat{b}_{i_1}=\cdots=\hat{b}_{i_k}\}, 
\end{align}
where $\kappa_k=\binom{k+d-1}{k}$, and $\textbf{1}\{Y\}=1$ if the event $Y$ is true and 0 otherwise.
As shown in Supplementary Material (SM) Sec.~S6 \cite{SM}, $\hat M_k^U$ can be efficiently computed classically.
In addition, when the ensemble $\mathcal E$ forms a unitary $t$-design \cite{Mele23,haferkamp2023efficient}, its expectation value (taken over the randomness in unitary selection and in quantum measurements) is given by (see SM Sec.~S2): 
\begin{align}\label{eq:zeta_def_prl}
\E \left[\hat M_k^U\right]\!=
\kappa_k \, \underset{U\sim \mathcal E}{\E} \left[{\Pr}_{\rho}(b|U)^k\right]
= \frac{1}{k!}\sum_{\pi\in \mathrm{S}_k}\!\Tr(\rho^{\otimes k}R_\pi) :=\zeta_k,
\end{align}
where $\Pr_{\rho}(b|U):=\<b|U\rho U^\dag|b\>$, $\underset{U\sim \mathcal E}{\E}[\cdot]$ denotes the average over random unitaries in $\mathcal E$, 
$\mathrm{S}_k$ denotes the permutation group of order $k$, 
and $R_\pi$ is the permutation operator on $\caH^{\otimes k}$ 
associated with $\pi\in \mathrm{S}_k$.

In the summation of \eref{eq:zeta_def_prl}, each $\pi \in \mathrm{S}_k$ contributes according to its cycle structure, making $\zeta_k$ a polynomial in $\{p_j\}_{j \le k}$: 
\begin{align}\label{eq:zetaRelation}
\zeta_2= \frac{1}{2}+\frac{p_2}{2}, \quad  
\zeta_3= \frac{1}{6}+\frac{p_2}{2}+\frac{p_3}{3}, \quad  \dots 
\end{align} 
Let $\hat{\zeta}_k := \sum_{s=1}^{N_U} \hat{M}_k^{U_s}/N_U$ denote the estimator of $\zeta_k$ obtained by averaging over $N_U$ random unitaries.
Finally, by substituting these estimates into \eref{eq:zetaRelation}, we can sequentially solve for $p_2,\dots,p_t$ \cite{compare_PhysRevLett.108.110503}.

\emph{From moments to more general nonlinear functions}---
The CBNE framework extends naturally from moments $p_t$ to $\Tr(O\rho^t)$. 
Let $O_0:=O-\Tr(O)\openone/d$ be the traceless part of $O$. 
Using the data $\mathbf{b}_U$, we introduce the following estimator by generalizing $\hat M_k^U$ in \eref{eq:collision_moment_prl}:  
\begin{align}\label{eq:collision_general_prl}
\hat{\Gamma}_k^U(O_0)
:=
\frac{\kappa_{k+1}}{d\binom{N_M}{k}}
\sum_{i_1<\cdots<i_k}\!
\mathbf{1}\{\hat{b}_{i_1}\!=\cdots=\hat{b}_{i_k}\}
\,\widetilde{\Pr}_{O_0}(\hat{b}_{i_1}|U).
\end{align}
Here $\widetilde{\Pr}_{O_0}(b|U)=\bra{b} U O_0 U^{\dag} \ket{b}$ represents the \emph{quasi}-outcome probability of (fictitious) computational-basis measurements on the \emph{pseudostate} $U O_0 U^{\dag}$, and can be computed classically.

When the ensemble $\mathcal E$ forms a unitary $(t+1)$-design,
$\hat{\Gamma}_k^{U} (O_0)$ is an unbiased estimator for (see SM Sec.~S3): 
\begin{align}\label{eq:xikO0def}
\E \left[ \hat{\Gamma}_k^U (O_0)\right] 
= \frac{\sum_{\pi\in \mathrm{S}_{k+1}}  \Tr\left[ \left(\rho^{\otimes k} \otimes O_0 \right)R_\pi\right] }{(k+1)! } 
:=\xi_k^{O_0}. 
\end{align}
By definition, $\xi_k^{O_0}$ can be expanded as a polynomial in 
$p_j$ and $\Tr(O_0\rho^i)$ with $1\leq i,j\leq k$:  
\begin{align}\label{eq:xiORelation}
\!\xi_1^{O_0}&=  \frac{\Tr(O_0\rho)}{2}, 
\quad \ 
\xi_2^{O_0}=  \frac{\Tr(O_0\rho)+\Tr(O_0\rho^2)}{3},
\nonumber\\ 
\!\xi_3^{O_0}&=  \frac{(1+p_2)\Tr(O_0\rho)}{8} + \frac{\Tr(O_0\rho^2)+\Tr(O_0\rho^3)}{4}, \ \, \dots
\end{align}
Inverting these relations using the estimated values of $\xi_i^{O_0}$ and $p_j$ yields $\Tr(O_0\rho^k)$ sequentially for $k=1,2,\dots,t$. 
Consequently, we obtain $\Tr(O\rho^k)= \Tr(O_0\rho^k)+ \Tr(O)p_k/d$. 

The complete CBNE protocol is summarized in Algorithm~\ref{alg:CBNEmain} in the End Matter. Its performance is guaranteed by the following theorem; see SM Sec.~S3 for a proof.

\begin{theorem}\label{thm:MomOResult}
Suppose $t\geq 1$ is a constant integer and $O$ is an observable with bounded operator norm $\|O\|\le 1$. 
Let $O_0:=O-\Tr(O)\openone/d$ and $\mathfrak{B}= \max \left\lbrace \Tr(O_0^2),1\right\rbrace$. 
For any $0<\epsilon<1$, the CBNE protocol can return $\epsilon$-additive error estimates for all $\Tr(O\rho),\dots,\Tr(O\rho^t)$ and $p_2,\dots,p_t$ with high probability, provided that
\begin{align} 
N_U=\max\left\lbrace 1, c_1 \,\Xi_1 \right\rbrace, \ \  
N_M\geq c_2\,d \min\left\lbrace \Xi_1^{1/t}, 1 \right\rbrace, 
\end{align}
and the ensemble $\mathcal E$ forms a unitary $2(t+1)$-design. 
Here, $c_1,c_2>0$ are some constants independent of $\mathfrak{B},d$ and $\epsilon$, and $\Xi_1$ is a shorthand for $\Xi_1(d,\mathfrak{B},\epsilon):=\mathfrak{B}/(d\epsilon^2)$. 
State moment estimation is recovered by taking $O=\openone$.
\end{theorem}

Our CBNE protocol exhibits notable advantages in several aspects  
(see Table~\ref{tab:compare} in the End Matter and SM Sec.~S7 for comparisons with existing approaches).

\textbf{First}, from an experimental perspective, it requires only a single measurement setting ($N_U=1$) when the system dimension is sufficiently large, i.e., $d \ge c_1 \mathfrak{B}/\epsilon^2$, a condition that is generically satisfied for $O=\openone$ and for observables with bounded rank. 
If instead $d < c_1 \mathfrak{B}/\epsilon^2$, single-setting estimation can still be achieved via a slight modification of the basic framework:  
after preparing each $\rho$, we additionally introduce
\begin{align} 
n_a=\left\lceil \log_2 \left(\frac{c_1 \mathfrak{B}}{d\,\epsilon^2} \right)\right\rceil
\leq 2\log_2 (\epsilon^{-1}) +\bigo{1}
\end{align}
ancillary qubits (initialized in the trivial state $\rho_a=|0^{n_a}\>\<0^{n_a}|$) to increase the system size; 
the CBNE protocol is then executed within the extended system  \(\tilde{\rho} = \rho \otimes \rho_a\) and observable \(\tilde{O} = O\otimes \rho_a\), preserving the target value \(\Tr(\tilde{O} \tilde{\rho}^{\,t}) = \Tr(O \rho^t)\).  
This modification reduces the number of required measurement settings to $N_U=1$, while preserving the overall sample cost \cite{eff_dim}.

\textbf{Second}, the sample complexity of our protocol reads 
\begin{align} \label{eq:SampleCompCBNE}
N_{\text{tot}}=N_U N_M= \bigo{\max\bigg\lbrace \frac{d^{1-1/t}\mathfrak{B}^{1/t}}{\epsilon^{2/t}}, \frac{\mathfrak{B}}{\epsilon^2} \bigg\rbrace}, 
\end{align}
which is the best known result so far for order $t\geq 3$. 
In particular, in the high-precision regime $\epsilon^2 \le \mathfrak{B}/d$,
the cost becomes $N_{\text{tot}}=\bigo{\mathfrak{B}/\epsilon^2}$, which achieves the same scaling as the CSE method for estimating $\Tr(O\rho)$ \cite{HKPshadow20}, demonstrating that nonlinear functions can be estimated as efficiently as linear ones  in this regime.

For state-moment estimation, corresponding to $O=\openone$ and $\mathfrak{B}=1$, the sample complexity scales as $\mathcal O(d^{1-1/t})$ for any fixed accuracy. 
By establishing a matching lower bound, we further show that this scaling is optimal among all protocols that use a single projective measurement setting (see SM Sec.~S10 for a proof).

\begin{theorem}
\label{thm:LB_main}
Let $t\geq 2$ be a constant integer and $0<\epsilon<(2t)^{1-t}$. 
If a protocol measures $N$ copies of  $\rho$ in a single (possibly randomly chosen) orthonormal basis and outputs an $\epsilon$-additive error estimate for $p_t=\Tr(\rho^t)$ with probability at least $2/3$, then necessarily $N= \Omega (d^{1-1/t} )$.
\end{theorem}

\noindent Therefore, CBNE achieves the optimal sample complexity with the system dimension for state-moment estimation under the single-setting measurement model.

\textbf{Third}, our protocol inherits and further generalizes two key advantages of CSE \cite{HKPshadow20}: 
(i) The experimental stage is independent of the target observable, enabling simultaneous estimation of multiple observables;  
(ii) As analyzed in SM Sec.~S6, the classical postprocessing overhead of CBNE is minimal, reducing to simple collision counting with few or no matrix operations.

\textbf{Fourth}, as shown in the End Matter, 
when the target observable $O$ has a low rank or equals the identity $\openone$, 
log-depth random circuits drawn from an approximate unitary design \cite{schuster2025random,cui2025unitary}
can be applied in place of the random unitaries used in CBNE, 
while retaining almost the same performance guarantees. 
This further drastically reduces the experimental overhead,  
making the protocol appealing to near-term quantum devices \cite{Preskill2018quantumcomputingin}.

\emph{Numerical results}---
We validate the performance of CBNE by numerically simulating the estimation of $p_3=\Tr(\rho^3)$. The target state $\rho$ is chosen as the ground state of the 1D \emph{transverse-field Ising model} (TFIM) with open boundary conditions, whose Hamiltonian has the form 
$H = -J\sum_{i=1}^{n-1} X_i X_{i+1} - h \sum_{i=1}^n Z_i$.
For computational feasibility, the unitary ensemble used in our simulation is generated via $2$-local brickwork circuits  composed of nearest-neighbor Haar-random gates, which are known to form an approximate design~\cite{Brandao2016LocalDesign,Haferkamp2022Interleaved}.

\begin{figure}
\centering
\includegraphics[width=0.45\textwidth]{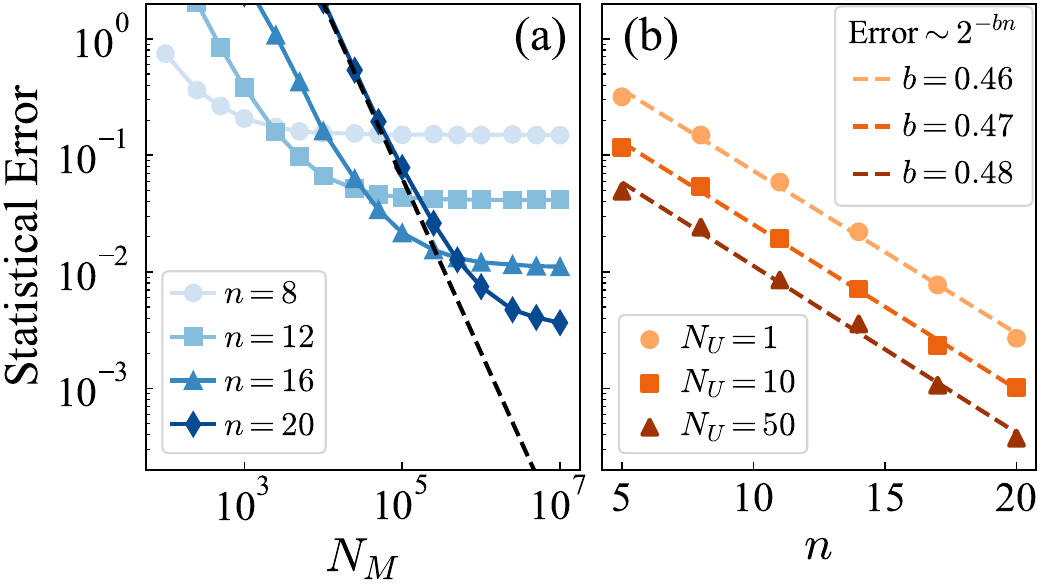}
\caption{Average statistical error in estimating $p_3$ via CBNE. Here $\rho=(1-p)|\mathrm{gs}\>\<\mathrm{gs}|+pI/d$, $p=0.2$, and $|\mathrm{gs}\>$ is the ground state of TFIM with $J=h=1.0$. 
(a) Error scaling versus \( N_M \) for \( N_U=1 \) across different system sizes \( n \); the dashed black line  represents the scaling $\propto N_M^{-3/2}$. 
(b) Error scaling versus \( n \) for fixed \( N_M=10^8 \) and different \( N_U \).}
\label{fig:SE_thirdM}
\end{figure}

We investigate the scaling behavior of the estimation error with respect to the system size $n$ and the number of measurements $N_M$ per unitary. 
Figure~\ref{fig:SE_thirdM}(a) reveals two distinct scaling regimes. 
In the small $N_M$ regime, the error is proportional to $N_M^{-3/2}$, and decreases rapidly with increasing $N_M$. 
This behavior is consistent with our theoretical result in \tref{thm:MomOResult}, which shows that the error in estimating $\Tr(O\rho^t)$ scales as $N_M^{-t/2}$ when $N_M\ll d$. 
In the large $N_M$ regime, the error saturates and instead 
decreases exponentially with $n$, as illustrated in Fig.~\ref{fig:SE_thirdM}(b). 
This exponential suppression agrees with the theoretical prediction in \tref{thm:MomOResult}, 
demonstrating that moment estimation can be achieved with a single measurement setting when the system size is large enough.

\emph{Principal component estimation}---  
To further showcase the versatility of CBNE, we apply it to the task of \emph{principal component estimation} (PCE), 
which is valuable for many applications such as error mitigation 
\cite{PhysRevX.11.041036,PhysRevX.11.031057,RevModPhys.95.045005,PRXQuantum.4.010303,o2023purification,PhysRevLett.133.080601}. 
In this task, we are given copies of an unknown  state $\rho$, 
and the goal is to estimate eigen-property $\bra{\psi}O\ket{\psi}$ of the principal component of $\rho$.  
Here $O$ is an arbitrary observable with $\|O\|\leq 1$, and $\ket{\psi}$ is the principal eigenstate of $\rho$ that corresponds to the largest eigenvalue $\lambda$. 
Following prior works~\cite{huang2022quantum,grier2024principal}, we assume that
$\lambda$ is larger than all other eigenvalues of $\rho$ by a constant that is independent of the system dimension $d$.

To address the PCE problem, recent near-term proposals \cite{PhysRevX.11.041036,PhysRevX.11.031057,PhysRevLett.133.080601} consider the estimation of $\Tr(O\rho^t)/p_t$, which converges to $\bra{\psi}O\ket{\psi}$ as $t$ increases.  
Our CBNE protocol directly estimates both $\Tr(O\rho^t)$ and $p_t$, and thus provides a natural approach to PCE. 
Its performance guarantees for this task are summarized below; see SM Sec.~S4 for a proof.

\begin{theorem}\label{thm:PrinEigen}
The CBNE protocol can estimate $\bra{\psi}O\ket{\psi}$ up to any constant additive error with high probability, 
provided that $t\geq c_1$, $N_U= \max\left\lbrace 1, c_2 \mathfrak{B}/d\right\rbrace$, $N_M\geq c_3 d$.
Here $\mathfrak{B}= \max \left\lbrace \Tr(O_0^2),1\right\rbrace$, and $c_1,c_2,c_3>0$ are some constants independent of $\mathfrak{B}$ and $d$. 
\end{theorem}

Again, our protocol operates with very few measurement settings. 
In addition, the observable $O$'s information is not employed during the experimental stage.
With this constraint, the best previously known protocol for PCE has a sample complexity of $\mathcal{O}(d)$ for general $O$, but relies on collective measurements on multiple copies of $\rho$ \cite{grier2024principal}.  
In contrast, our protocol achieves this efficiency using only single-copy measurements.

\begin{figure}[t]
\centering
\includegraphics[width=0.45\textwidth]{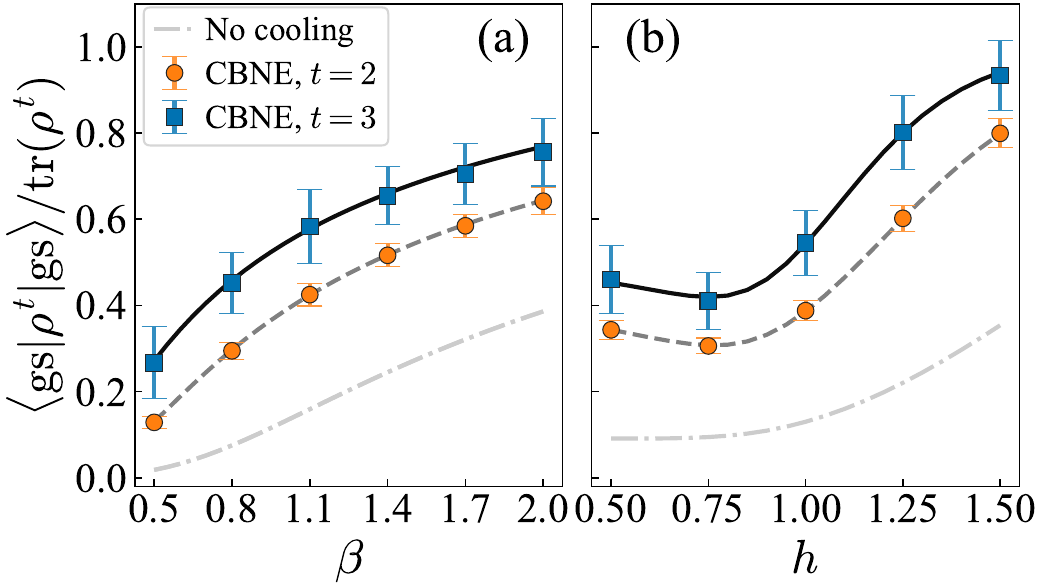}
\caption{Estimated fidelities between $|\mathrm{{gs}}\rangle$ and the virtually cooled state $\rho^t/\tr{\rho^t}$ as functions of the inverse temperature $\beta$ and coupling strength $h$. Here $\rho$ is the Gibbs state of the 14-qubit TFIM, and $|\mathrm{gs}\>$ denotes the ground state. The parameters are chosen as (a) $J=h=1.0$ and (b) $J=1.0$, $\beta=1.0$. Solid and dashed curves represent the exact fidelities. Each point represents the average fidelity estimated by CBNE with $N_U = 1$ and $N_M = 10^6$ over 100 independent trials, and error bars indicate the standard deviation.
}
\label{fig:VCooling}
\end{figure}

Next, as a proof-of-principle demonstration, we apply the CBNE protocol to QVC~\cite{PhysRevX.9.031013,du2025optimal}. This technique aims to characterize properties of an effective low-temperature state using measurements on a noisy thermal ensemble, without requiring physical cooling of the system. For a Gibbs state $\rho = \rme^{-\beta H}\!\!/\!\Tr(\rme^{-\beta H})$ with inverse temperature $\beta$, the virtually cooled state $\rho^{(t)}:=\rho^t/p_t \propto \rme^{-t\beta H}$
corresponds to a system at a lower effective temperature. Since CBNE serves as an effective tool for estimating the properties of $\rho^{(t)}$, it can be used naturally for QVC. We numerically demonstrate this using the TFIM Hamiltonian, focusing on measuring the fidelity between the virtually cooled state and the ground state. Remarkably, as illustrated in Fig.~\ref{fig:VCooling}, CBNE yields accurate estimates with only a single measurement setting, significantly simplifying experimental implementation.

\emph{Estimation of PT moments}---
Beyond estimating nonlinear properties $\Tr(O\rho^t)$, our CBNE framework can be extended to measure PT moments of quantum states.   
Consider a bipartition of the original $n$-qubit state $\rho$ into two subsystems $A$ and $B$, and assume without loss of generality that 
their dimensions $d_A \ge d_B$. 
The $t$-th PT moment is then defined as $p_t^{\PT}:=\Tr\left[(\rho_{AB}^{\top_{\! B}})^t\right]$, where $(\cdot)^{\top_{\! B}}$ denotes partial transpose with respect to $B$ \cite{PhysRevLett.121.150503,elben2020mixedstate,singlezhou}.

Here we outline the key features of our \emph{PT moment estimation} (PTME) protocol, with full details provided in the End Matter.
Unlike CBNE, which uses global unitaries acting on the entire system, 
PTME applies random unitaries only to subsystem $A$, making it more experimentally accessible. 
It then measures one part of \(A\) in the computational basis and performs Bell measurements on qubit pairs between the remaining part of \(A\) and subsystem \(B\), as illustrated in Fig.~\ref{fig:PTME}(a).
As before, $N_U$ and $N_M$ denote the number of distinct random unitaries used and the number of measurements per unitary, respectively. 

\begin{figure}[t]
\centering
\includegraphics[width=0.46\textwidth]{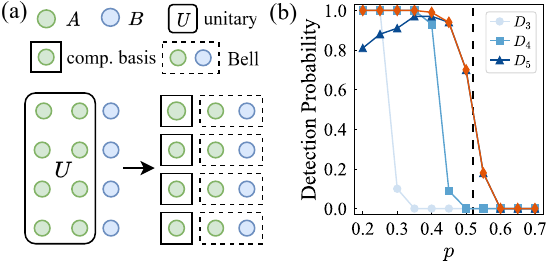}
\caption{(a) Illustration of the PTME protocol. A random unitary \(U\) is applied to subsystem \(A\), followed by computational-basis measurements on part of \(A\) and Bell measurements on paired qubits between \(B\) and the remaining part of \(A\).
(b) Entanglement-detection probability for 
\(\rho=(1-p)|\mathrm{gs}\rangle\langle\mathrm{gs}|+pI/d\), 
where \(|\mathrm{gs}\rangle\) is the ground state of a 20-qubit TFIM with \(J=h=1\). 
We take subsystem $A$ as the first 18 qubits and $B$ as the remaining 2 qubits. 
Each point is the fraction of 100 independent PTME trials in which the estimated witness detects entanglement, with \(N_U=1\) and \(N_M=10^8\).
Red diamonds indicate detection by at least one of \(D_3,D_4,D_5\), 
and the vertical dashed line marks the theoretical threshold at which at least one of \(D_3,D_4,D_5\) becomes positive.
}
\label{fig:PTME}
\end{figure}

\begin{theorem}\label{thm:PTResult}
Suppose $t\geq 2$ is a constant integer.
For any $0<\epsilon<1$, the PTME protocol can return $\epsilon$-additive error estimates for all 
$p_2^{\PT},\dots,p_t^{\PT}\!$ with high probability, provided that
$d_{\!A}\epsilon\geq c_1$ and 
\begin{align} 
N_U\!=\max\left\lbrace 1, c_2 \,\Xi_2 \right\rbrace, \ \  
N_M\!\geq c_3\,d_{\!A} d_B \min\left\lbrace \Xi_2^{1/t}, 1 \right\rbrace\!
\end{align}
for some constants $c_1,c_2,c_3>0$ independent of $d_A,d_B,\epsilon$. Here, $\Xi_2$ is a shorthand for $\Xi_2(d_A,d_B,\epsilon):=d_B/(d_{A}\epsilon^2)$. 
\end{theorem}

Notably, when subsystem $A$ is sufficiently larger than $B$, the PTME protocol again requires only $N_U=1$ measurement setting, greatly simplifying experimental implementation. 
If instead $d_A$ is small, similar to CBNE, single-setting estimation of PT moments can still be achieved by first introducing at most $2\log_2 (\epsilon^{-1}) +\bigo{1}$ ancillary qubits (initialized in a pure state $\rho_a$) to subsystem $A$, and then applying PTME to the extended state $\tilde{\rho} = \rho_{AB}\otimes \rho_a$. 
This modification preserves the target value $\Tr\left[(\tilde\rho^{\top_{\! B}})^t\right]=\Tr\left[(\rho_{AB}^{\top_{\! B}})^t\right]$ and does not increase the sample cost. 
Moreover, PTME achieves the best-known sample complexity for orders $t\geq 4$. 
A detailed comparison with existing approaches is provided in SM Sec.~S7.

We demonstrate the utility of PTME for entanglement detection using the \(D_m\) witness hierarchy~\cite{Neven2021Entanglement}; see SM Sec.~S9 for details. 
Each \(D_m\) is constructed from PT moments up to order \(m\), and higher orders provide stronger witnesses but are typically harder to estimate experimentally. 
PTME addresses this challenge by estimating several PT moments simultaneously from a single measurement setting. 
As shown in Fig.~\ref{fig:PTME}(b), combining the witnesses \(D_3,D_4,D_5\) leverages the statistical robustness of lower orders and the detection power of higher orders, thereby improving the overall detection probability and extending the detectable range of entangled states.

\emph{Conclusion}---
We developed the CBNE framework for the efficient and practical estimation of nonlinear properties of quantum states with a single measurement setting. 
Beyond its versatility for a wide range of quantum information processing tasks, CBNE suggests that switching among multiple measurement bases is not fundamentally necessary for accessing a broad class of nonlinear state properties. 
This observation opens intriguing directions for quantum learning and foundational studies.

Several directions merit further investigation. 
First, extending our framework to more general nonlinear functions like $\Tr(O\rho^{\otimes t})$ and $\Tr(O_1\rho^{t_1}\cdots O_m\rho^{t_m})$ can broaden its applicability to quantum information and many-body physics \cite{PRXQuantum.5.030338,PhysRevLett.127.260501,zhou2024hybrid}; see SM Sec.~S8 for a preliminary exploration.  
Second, it would be valuable to explore hybrid schemes that combine CBNE with limited collective operations \cite{PRXQuantum.6.010336,chen2024optimal}, potentially further improving sample efficiency. 
Finally, investigating fundamental lower bounds on the sample cost of nonlinear estimation \cite{anshu2022distributed,chen2022exponential,gong2024sample} would help assess the optimality of our framework.

\bigskip
\paragraph*{Note added.}
In this work, we establish an $\Omega(d^{1-1/t})$ sample lower bound for estimating the $t$-th moment $\Tr(\rho^t)$ using a single projective measurement basis. 
We recently became aware of an independent work by Zhenhuan Liu~\cite{zhenhuanliu2026}, which establishes the same lower bound for the broader class of non-adaptive single-copy measurement protocols using a different approach. Together with our upper bound $\mathcal O(d^{1-1/t})$, this result shows that CBNE is sample-optimal for state-moment estimation among all protocols using non-adaptive single-copy measurements, for fixed $t$ and fixed accuracy.

\section*{Acknowledgments}
We thank Huangjun Zhu, Zhenhuan Liu, Changhao Yi, Qisheng Wang, and Yanglin Hu
for inspiring  discussions. 
This work is supported by
the National Natural Science Foundation of China via the Excellent Young Scientists Fund (Hong Kong and Macau) Project 12322516, the Hong Kong Grant Council (RGC) through the National Natural Science Foundation of China (NSFC)/Research Grants Council (RGC) Joint Research Scheme project N\_HKU7107/24 and General Research Fund (GRF) project 17305625. Y.Z. acknowledges the support from the National Natural Science Foundation of China (NSFC) Grant No.~12575012, the Quantum Science and Technology-National Science and Technology Major Project Grant Nos.~2024ZD0301900 and 2021ZD0302000, the Shanghai QiYuan Innovation Foundation, the Shanghai Municipal Commission of Science and Technology with Grant No.~25511103200, the Shanghai Science and Technology Innovation Action Plan Grant No.~24LZ1400200, the Shanghai Pilot Program for Basic Research-Fudan University 21TQ1400100 (25TQ003), the CCF-Quantum CTek Superconducting Quantum Computing CCF-QC2025006.

\onecolumngrid
\section*{End Matter}

\twocolumngrid

\begin{figure}[b]
\begin{algorithm}[H]
{\small
\hspace{1pt}\textbf{Input:}  \!$N_U N_M$ prepared state $\rho$, observable $O$, integer $t\geq 1$.   \\

\begin{algorithmic}[1]
\caption{{\small CBNE for estimating $\left\{p_k, \Tr(O\rho^k)\right\}_{k=1}^t$}}
\label{alg:CBNEmain}

\For{$s=1,2,\dots,N_U$,}
\State{Randomly choose a unitary $U_s$ from an ensemble $\mathcal E$ and record it. }
\For{$j=1,2,\dots,N_M$,}
\State{Apply $U_s$ to $\rho$, obtaining the rotated state $U_s\rho U_s^\dagger$.} 
\State{Measure the rotated state in the computational basis, obtaining outcome $\hat{b}_j\in\{0,\dots,d-1\}$.}
\EndFor 
\State{Compute the collision estimators $\hat M_k^{U_{\!s}}$ and $\hat{\Gamma}_k^{U_{\!s}} (O_0)$ [defined in Eqs.~\eqref{eq:collision_moment_prl} and \eqref{eq:collision_general_prl}] for all $k=1,\dots,t$ using the measurement data $\mathbf{b}_{U_{\!s}}=\{\hat{b}_1,\dots,\hat{b}_{N_M}\}$.}
\EndFor 
\State{For $k=1,\dots,t$, compute the averages
\begin{align}
\quad  
\hat\zeta_k := \frac{1}{N_U} \sum_{s=1}^{N_U} 
\hat{M}_k^{U_s}
\ \  \text{and}\ \ 
\hat\xi_k^{O_0}:= \frac{1}{N_U} \sum_{s=1}^{N_U} \hat{\Gamma}_k^{U_{\!s}}(O_0)
\end{align}
as estimates for $\zeta_k$ and $\xi_k^{O_0}$, respectively.}

\State{Substitute $\hat\zeta_2,\dots,\hat\zeta_t$ into \eref{eq:zetaRelation}, then solve the equations to obtain estimates $\hat{p}_k$ for $p_k$ ($k=2,\dots,t$).}

\State{Substitute $\hat{p}_2,\dots,\hat{p}_t$ and $\hat{\xi}_1^{O_0},\dots,\hat{\xi}_t^{O_0}$ into  \eref{eq:xiORelation}, then solve the equations to obtain estimates $\hat{o}_k^{(0)}$ for $\Tr(O_0\rho^k)$ ($k=1,\dots,t$).}

\State{Let $\hat{p}_1\!=1$; compute $\hat{o}_k \!:= \hat{o}_k^{(0)}+\Tr(O)\hat{p}_k/d$ for $k=\!1,\dots,t$.} 

\end{algorithmic}
\hspace{-105pt} \textbf{Output:} $\hat{p}_2, \dots, \hat{p}_t$ and $\hat{o}_1, \dots, \hat{o}_t$. 
}
\end{algorithm}
\vspace{-7pt}
\end{figure}

\emph{Comparison with previous works}--- 
In Table~\ref{tab:compare}, we compare our CBNE and PTME protocols with existing approaches for estimating nonlinear properties of quantum states using single-copy operations. 
In particular, we focus on the number of required measurement settings \(N_U\), which is the key experimental resource considered in this work. 
Further details and comparisons on other aspects are provided in SM Sec.~S7.

\emph{Complete algorithm of the CBNE protocol}--- 
We summarize the full CBNE protocol in Algorithm~\ref{alg:CBNEmain} below. 

\begin{table*}
\centering
\caption{Comparison of the number of measurement settings \(N_U\) required by various protocols for estimating nonlinear properties with single-copy operations
(see SM Sec.~S7 for details). 
Here, \(n\) is the qubit number of the target state \(\rho\), \(d=2^n\) is the system dimension, 
\(d_A\) and \(d_B\) are subsystem dimensions, 
\(\epsilon\) is the target additive error, 
\(\mathfrak{B}=\max\{\Tr(O_0^2),1\}\), and ``dichotomic'' means Hermitian with eigenvalues \(\pm1\). 
Previous purity-estimation protocols and the ORM protocol for dichotomic observables can also achieve single-setting estimation in the large-\(d\) regime or using a few ancillary qubits, but only for order \(t=2\).
\label{tab:compare}
}
\begin{math}
\begin{array}{c|cc}
\hline\hline
\text{Protocol} & N_U & \text{Conditions/Constraints} \\[0.3ex]
\hline
\text{CBNE for $\Tr(O\rho^t)$ (this work)} 
& 1 
& d \gtrsim \mathfrak{B} \epsilon^{-2} \ \text{or}\ \sim 2\log_2(\epsilon^{-1}) \ \text{anc.\! qubits} \\[0.1ex]

\text{PTME for $p_t^{\PT}$ (this work)} 
& 1 
& d_A \gtrsim d_B \epsilon^{-2} \ \text{or}\ \sim 2\log_2(\epsilon^{-1}) \ \text{anc.\! qubits} \\[0.5ex]

\hline

\text{Purity estimation
\cite{PhysRevA.99.052323,brydges2019probing,PhysRevLett.124.010504,anshu2022distributed}
} 
& \bigo{\max\lbrace 1, 1/(\epsilon^2 d) \rbrace}
& \text{restricted to order} \ t=2 \\[0.3ex]

\text{ORM protocol for $\Tr(O\rho^2)$
(dichotomic $O$) \cite{du2025optimal} 
} 
& \bigo{\max\lbrace 1, 1/(\epsilon^2 d) \rbrace }
& \text{restricted to $t=2$; observable-dependent} \\[0.3ex]

\text{ORM protocol for $\Tr(O\rho^2)$
(general $O$) \cite{du2025optimal} 
} 
& \bigo{\max\lbrace 1, 1/(\epsilon^2 d) \rbrace
\log(\epsilon^{-1}) }
& \text{restricted to $t=2$; observable-dependent} \\[0.3ex]

\text{BRM protocol for $\Tr(O\rho^2)$ 
\cite{du2025optimal}
} 
& \mathcal{O}(\epsilon^{-2})
& \text{restricted to $t=2$} \\

\text{ZZL protocol for $p^{\mathrm{PT}}_3$ \cite{singlezhou}} 
& \mathcal{O}(\epsilon^{-2}) 
& \text{only $t=3$; global $U$ on $A\cup B$ or $d_A=d_B$} \\

\text{Standard CSE 
\cite{elben2020mixedstate,PhysRevLett.127.260501,HKPshadow20,elben2023randomized,PRXQuantum.4.010303}
} 
& \ \text{large; exponential in $n$}\ 
& \text{---} \\
\text{Variant CSE \cite{PhysRevX.13.011049,PhysRevLett.131.160601}} 
& 1 
& \bigo{n+\log(\epsilon^{-1})} \ \text{anc.\! qubits} \\[0.5ex]

\hline\hline
\end{array}
\end{math}
\vspace{-2pt}
\end{table*}

\emph{CBNE with an approximate unitary design ensemble}---
The CBNE protocol introduced in the main text employs random unitaries drawn from an exact $2(t+1)$-design ensemble $\mathcal{E}$, which generally requires circuit depths that scale polynomially with the number of qubits. 
This raises the natural question of whether \emph{approximate unitary designs} can be used instead, as such ensembles are much more experimentally feasible: Recent works~\cite{schuster2025random,Laracuente2024approximate} have shown that for constant $k$, an $\mu$-approximate unitary $k$-design on $n$ qubits can be implemented with circuit depth $\mathcal{O}(\log(n/\mu))$ without using ancillary qubits. The depth can be further reduced to $\mathcal{O}(\log\log(n/\mu))$ when ancillary qubits are available~\cite{cui2025unitary}. A more detailed introduction
to approximate unitary designs is given in SM Sec.~S1.

Here, we affirmatively answer this question. First, for the estimation of quantum state moments, 
the CBNE protocol using approximate designs maintains the performance  stated in Theorem~\ref{thm:MomOResult}
(see SM Sec.~S2 for proof):

\begin{prop}\label{prop:appCBNEmoment}
Suppose $t>1$ is a constant integer. For any $0<\epsilon<1$, 
the CBNE protocol can return $\epsilon$-additive error estimates for all 
$p_2,p_3,\dots,p_t$ with high probability, provided that: $(i)$ the random unitaries are sampled from a $(c_3\epsilon^2)$-approximate $2t$-design ensemble, and 
\begin{align}\label{eq:MomResultApp}
(ii) \, N_U= \max\left\lbrace 1, \dfrac{c_1}{d\epsilon^2} \right\rbrace ,  
\ 
(iii) \, N_M\geq  c_2  \min\left\lbrace \frac{d^{1-1/t}}{\epsilon^{2/t}}, d \right\rbrace.     
\end{align}
Here $c_1,c_2,c_3>0$ are constants  independent of $d$ and $\epsilon$. 
\end{prop}

Second, for estimating $\Tr(O\rho^t)$ with approximate unitary design, we introduce a slight modification to CBNE protocol's postprocessing. 
Instead of $\hat{\Gamma}_k^U(O_0)$ in \eref{eq:collision_general_prl}, we now calculate $\hat{\Gamma}_k^U(O)$ and use $\hat\xi_k^{O}=\sum_{s=1}^{N_U} \hat{\Gamma}_k^{U_{\!s}}(O)/N_U$ as the estimate for $\xi_k^{O}$ [defined in \eref{eq:xikO0def} with $O_0 \mapsto O$]. This variant is technically convenient for controlling the bias induced by approximate designs. 
Since \(\xi^O_k\) can be expanded as a polynomial in \(p_j\) and \(\Tr(O\rho^i)\) for \(1\le i,j\le k\), the quantities \(\Tr(O\rho^k)\) can again be recovered sequentially from the estimated \(p_k\) and \(\xi^O_k\).

As proved in SM Sec.~S3, the performance of this variant CBNE protocol is guaranteed by the following theorem. Note that we assume without loss of generality that the observable $O$ is positive semidefinite ($O\geq0$); if not, we decompose $O$ as $O = O_+ - O_-$ with $O_\pm \geq0$ and estimate $\Tr(O_+\rho^t)$ and $\Tr(O_-\rho^t)$ separately. 

\begin{prop}\label{prop:appCBNE}
Suppose $t\geq1$ is a constant integer and $O$ is an observable with $O\geq 0$ and $\Tr(O)\geq1$. 
For any $0<\epsilon<1$, the variant CBNE protocol can return $\epsilon$-error estimates for all 
$\big\{p_k,\Tr(O\rho^k)\big\}_{k=1}^t$ with high probability, provided that: 
$(i)$ the random unitaries are sampled from a $[c_3\epsilon^2\!/\!\Tr(O)^2]$-approximate $2(t+1)$-design ensemble, and 
\begin{align}
{(ii)} \ N_U &= \max\left\lbrace 1, \dfrac{c_1 \Tr(O)^2}{d\epsilon^2} \right\rbrace, 
\\
{(iii)} \ N_M &\geq  c_2  \min\left\{ \frac{d^{1-1/t}\Tr(O)^{2/t}}{\epsilon^{2/t}}, d \right\}.  
\end{align}
Here $c_1,c_2,c_3>0$ are constants  independent of $O,d,\epsilon$. 
\end{prop}

Compared to \tref{thm:MomOResult}, the numbers $N_U$ and $N_M$ differ only in the replacement of $\mathfrak{B} = \max \lbrace \Tr(O_0^2), 1\rbrace$ with $\Tr(O)^2$. Therefore, for low-rank observables $O$, 
nonlinear estimation with approximate designs achieves performance close to exact designs while requiring much shallower circuits.
This relaxation significantly reduces the implementation overhead of CBNE and strengthens its near-term applicability.

\begin{figure}[b]
\begin{algorithm}[H]
{\small
\hspace{-38pt}\textbf{Input:}  $N_U N_M$ prepared state $\rho$ and integer $t\geq 2$.   \\

\begin{algorithmic}[1]
\caption{{\small PTME protocol for estimating $\left\{p_k^{\PT}\right\}_{k=2}^t$} \ \ }
\label{alg:PTME}

\For{$s=1,2,\dots,N_U$,}
\State{Randomly choose a unitary $U_s$ on subsystem $A$ from a $2t$-design ensemble. }
\For{$j=1,2,\dots,N_M$,}
\State{Apply $U_s$ to subsystem $A$ of $\rho$, obtaining the rotated state $\rho_{U_s}=(U_s\otimes I_B)\rho(U_s\otimes I_B)^{\dag}$.}

\State{Measure the subsystem $A_1$ of $\rho_{U_s}$ in the computational basis, obtaining outcome $\hat{b}_j\in\{0,\dots,d_{A_1}-1\}$.}  
\State{Measure the subsystems $A_2$ and $B$ with the POVM $\mathcal{G}=\{G_+, G_-\}$,
	obtaining outcome $\hat{r}_{j}\in\{+1,-1\}$. 
	Here the POVM elements $G_{\pm}=(I\pm \mathbb{S})/2$, with $I$ and $\mathbb{S}$ being the identity and swap operators on $A_2\cup B$, respectively. }
\EndFor 
\State{Compute the following estimator for all $k=2,3,\dots,t$ using the measurement data 
$\mathbf{b}_{U_s} = \big\{\hat{b}_1, \hat{b}_2, \dots, \hat{b}_{N_M}\big\}$ and $\mathbf{r}_{U_s} = \{\hat{r}_1, \dots, \hat{r}_{N_M}\}$:  
\begin{align}\label{eq:hatLambdaPT}
	\ \hat{\Lambda}_k^{U_{\!s}}
	:= \frac{d_{\!A}^{\,k}}{k!\, d_{\!A_1} \binom{N_M}{k}}
	\sum_{i_1 <\dots < i_k }  \! \left( \hat{r}_{i_1}\cdots \hat{r}_{i_k} \right)  \textbf{1}\big\{\hat{b}_{i_1}=\dots=\hat{b}_{i_k}\big\}. \!
\end{align}}
\EndFor 
\State{For $k=2,3,\dots,t$, compute 
	$\hat\zeta_k^{\PT} := \frac{1}{N_U} \sum_{s=1}^{N_U} \hat{\Lambda}_k^{U_{\!s}}$ as our estimate for $\zeta_k^{\PT}$.}

\vspace{0.2em}

\State{Substitute $\hat\zeta_2^{\PT},\dots,\hat\zeta_t^{\PT}$ into  \eref{eq:zetaPTRelation}, then solve the equations to obtain estimates $\hat{p}_k^{\PT}$ for $p_k^{\PT}$ ($k=2,\dots,t$).}

\end{algorithmic}
\hspace{-149pt} \textbf{Output:} $\hat{p}_2^{\PT}, \dots, \hat{p}_t^{\PT}$. 
}
\end{algorithm}
\vspace{-7pt}
\end{figure}

\emph{PTME protocol}--- 
Here we provide more details on our PTME protocol. 
The protocol estimates PT moments through symmetrized quantities that play the same role as \(\zeta_k\) in the CBNE protocol. 
For integer \(k\ge2\), define
\begin{equation}\label{eq:zetakPT}
	\zeta_k^{\PT}
    := \frac{1}{k!}\sum_{\pi\in \mathrm{S}_k}
	\Tr\left\{\rho_{AB}^{\otimes k} \left[(R_\pi)_A \otimes (R_\pi^{\top})_B\right]\right\},
\end{equation}
where $(R_\pi)_A$ denotes the permutation operator acting on subsystem $A$ of $k$ copies, and similarly for $(R_\pi^{\top})_B$. 
As in ordinary moment estimation, \(\zeta_k^{\PT}\) depends only on the cycle structure of \(\pi\), and hence is a polynomial in the PT moments \(\{p_j^{\PT}\}_{j\le k}\):
\begin{align}\label{eq:zetaPTRelation}
\zeta_2^{\PT}= \frac{1}{2}+\frac{p_2^{\PT}}{2}, \quad  
\zeta_3^{\PT}= \frac{1}{6}+\frac{p_2^{\PT}}{3}+\frac{p_3^{\PT}}{2}, \quad  \dots 
\end{align}
Hence, estimating $\{\zeta_k^{\PT}\}_{k=2}^t$ allows one to sequentially recover $p_2^{\PT},\dots,p_t^{\PT}$ via \eref{eq:zetaPTRelation}.

By leveraging tools from representation theory, 
we construct a signed collision estimator $\hat{\Lambda}_k^{U}$ for $\zeta_k^{\PT}$, as defined in \eref{eq:hatLambdaPT}. 
This estimator processes measurement data obtained under a fixed unitary $U$, and can be computed efficiently classically (see SM Sec.~S6). 
As proved in SM Sec.~S5, $\hat{\Lambda}_k^{U}$ satisfies
\begin{align}\label{eq:ELambdakU}
	\mathbb{E}\left[ \hat{\Lambda}_k^{U} \right]
	= \zeta_k^{\PT} + \mathcal{O}(d_A^{-1}),   
\end{align}
so its bias is of order \(\mathcal{O}(d_A^{-1})\), which is exponentially small in the number of qubits of subsystem \(A\).
This guarantees accurate estimation of $\zeta_k^{\PT}$ and hence the PT moments.

The complete PTME protocol is summarized in  Algorithm~\ref{alg:PTME}; see Fig.~\ref{fig:PTME}(a) in the main text for an illustration. 
Recall that the subsystem $A$ is not smaller than $B$. 
For protocol implementation, we further partition $A$ into two disjoint subsystems $A_1$ and $A_2$, such that $d_A=d_{A_1}d_{A_2}$ and $d_{A_2}=d_B$.
This dimensional matching ensures the POVM $\mathcal{G}=\{G_+, G_-\}$  used in Step 6 of Algorithm~\ref{alg:PTME} is properly defined.

Note that the POVM elements $G_{+}$ and $G_{-}$ are projectors onto the $+1$ and $-1$  eigenspaces of the swap operator $\mathbb{S}$, respectively. 
Since $\mathbb{S}=\bigotimes_{i=1}^{n_B}\mathbb{S}_i$ admits a qubit-wise tensor-product form, 
the POVM $\mathcal{G}$ can be realized through $n_B$ local 2-qubit projective measurements: For each pair of qubits across $A_2$ and $B$, one performs the Bell measurement 
$\{I-|\Psi_-\>\<\Psi_-|,|\Psi_-\>\<\Psi_-|\}$, where
\begin{align}
|\Psi_-\>= \frac{1}{\sqrt{2}} (|01\>-|10\>), 
\ \ 
|\Psi_-\>\<\Psi_-|
= \frac{1}{2}(I_i-\mathbb{S}_i) . 
\end{align}
The overall  measurement result of $\mathcal{G}$ is determined by the parity of the outcome $|\Psi_-\rangle\langle\Psi_-|$ counts: An even (odd) number of such outcomes corresponds to $+1$ ($-1$). 


%

\let\addcontentsline\oldaddcontentsline

\newpage
\onecolumngrid
\newpage
\vspace{3em}
\begin{center}
\textbf{\large Quantum Nonlinear Properties from a Single Measurement Setting: \\ \smallskip Supplementary Material}
\end{center}

\setcounter{page}{1}

\setcounter{secnumdepth}{3}

\renewcommand{\figurename}{Fig.}
\renewcommand{\theequation}{S\arabic{equation}}
\renewcommand{\thetable}{S\arabic{table}}
\renewcommand{\thetheorem}{S\arabic{theorem}}
\renewcommand{\thelemma}{S\arabic{lemma}} 
\renewcommand{\theprop}{S\arabic{prop}} 
\renewcommand{\thealgorithm}{S\arabic{algorithm}}
\renewcommand{\thefigure}{S\arabic{figure}}

\renewcommand{\thesection}{S\arabic{section}}
\renewcommand{\thesubsection}{\Alph{subsection}}
\renewcommand{\thesubsubsection}{\alph{subsubsection}}

\setcounter{section}{0}
\setcounter{equation}{0}
\setcounter{figure}{0}
\setcounter{lemma}{0}
\setcounter{theorem}{0}
\setcounter{prop}{0}	
\setcounter{algorithm}{0}

\def\eqref#1{\textup{(\ref{#1})}}

\date{\today}

\tableofcontents

\section{Preliminaries}\label{sec:Prelimi}
For our proofs in the following context, in this section, we review some key concepts and notation about permutations, Haar measure, and (approximate) unitary designs. 
A more detailed review can be found in Ref.~\cite{Mele23}. 

Throughout the paper, we consider an \( n \)-qubit Hilbert space \( \mathcal{H} \) with dimension \( d = 2^n \). Denote by \( \mathcal{D}(\mathcal{H}) \) the set of all density operators on \( \mathcal{H} \), and by \( \mathcal{L}(\mathcal{H}) \) the set of linear operators on \( \mathcal{H} \). Let \( \openone \) be the identity operator on \( \mathcal{H} \), and $\{|b\>\}_{b=0}^{d-1}$ be the computational basis of \( \mathcal{H} \). For any \( A \in \mathcal{L}(\mathcal{H}) \), we write \( \|A\| \) for its operator norm, \( \|A\|_1\) for its Schatten 1-norm, and \( \|A\|_2 \) for its Schatten 2-norm (Frobenius norm).

Let \( \mathrm{S}_k \) denote the permutation group of order \( k \). For any \( \pi \in \mathrm{S}_k \), we define $\#(\pi)$ as the number of disjoint cycles in $\pi$, 
and define the associated permutation operator \( R_\pi \) acting on \( \mathcal{H}^{\otimes k} \) as the unitary operator satisfying  
\begin{align}
R_\pi \left( |\varphi_1\rangle \otimes \cdots \otimes |\varphi_k\rangle \right) 
= |\varphi_{\pi^{-1}(1)}\rangle \otimes \cdots \otimes |\varphi_{\pi^{-1}(k)}\rangle\qquad 
\forall \, |\varphi_1\rangle, \ldots, |\varphi_k\rangle \in \mathcal{H}. 
\end{align}
Here, the action of \( R_\pi \) rearranges the tensor factors according to the inverse permutation \( \pi^{-1} \). For instance, if \( \pi = (1,2) \in \mathrm{S}_2 \) (the transposition swapping elements 1 and 2), then \( R_{\pi} \) is the swap operator, which we denote specifically by $\mathbb{S}$. 
To describe permutations acting independently on sub-blocks, we write \((\tau_1, \tau_2) \in \mathrm{S}_{2k}\) for the permutation where \(\tau_1 \in \mathrm{S}_k\) acts on the first \(k\) elements and \(\tau_2 \in \mathrm{S}_k\) acts on the remaining \(k\) elements.

A Haar random unitary on a Hilbert space \( \mathcal{H} \) of dimension \( d \) is a unitary operator \( U \) drawn uniformly from the \emph{Haar measure} $\mu_{\rm H}$ on the unitary group \( {\rm U}(d) \). This measure is the unique unitarily invariant probability distribution satisfying 
\begin{align}
\underset{U\sim \mu_{\rm H}}{\E} [f(U)]
:=
\int_{{\rm U}(d)} f(U) {\rm d}\mu_{\rm H}(U) 
= \int_{{\rm U}(d)} f(VU) {\rm d}\mu_{\rm H}(U) 
= \int_{{\rm U}(d)} f(UV) {\rm d}\mu_{\rm H}(U) 
\end{align}
for any integrable function \( f \) and fixed \( V \in {\rm U}(d) \). Here we use ${\E}_{U\sim \mu_{\rm H}} [f(U)]$ to denote the expectation of $f (U )$  with respect to the Haar measure.  
For any $A \in \mathcal{L}(\mathcal{H}^{\otimes k})$, the \emph{$k$-fold Haar-random twirling channel} is defined as  
\begin{align}\label{eq:twirlingDef}
	\Phi_{\rm H}(A) := \underset{U\sim \mu_{\rm H}}{\E} \left[ U^{\otimes k} A U^{\dagger\otimes k} \right]. 
\end{align}  
When $A=(\ketbra{\psi}{\psi})^{\otimes k}$ with $\ket{\psi}\in\caH$ and $\langle\psi|\psi\rangle=1$, we have
\begin{align}\label{eq:PhiHaar}
\Phi_{\rm H}((\ketbra{\psi}{\psi})^{\otimes k})
=\frac{1}{\kappa_k }\Pi_{\mathrm{sym}}^{(k)}
=\frac{1}{ k!\, \kappa_k} \sum_{\pi \in \mathrm{S}_k} R_\pi, 
\qquad
\kappa_k=\binom{k+d-1}{k}, 
\end{align}
where $\sym^{(k)}$ is the projector onto the symmetric subspace of $\caH^{\otimes k}$, 
and $\kappa_k$ is the dimension of this symmetric subspace. 
For general $A \in \mathcal{L}(\mathcal{H}^{\otimes k})$, the $k$-fold twirling result reads 
\begin{align}\label{eq:tirling}
	\Phi_{\rm H}(A)
	= \sum_{\tau,\pi\in \mathrm{S}_k} \Wg_{\pi,\tau}^{(k,d)}\Tr(A R_{\pi}^{\top})R_{\tau}, 
\end{align}
where $(\cdot)^{\top}$ denotes the transpose operation (with respect to the computational basis), and 
$\Wg_{\pi,\tau}^{(k,d)}\in\mathbb{R}$ is known as the \emph{Weingarten function}. 
When the space dimension $d$ is large, the Weingarten function scales as
\begin{align}\label{eq:Weingarten}
	\Wg_{\pi,\tau}^{(k,d)}=d^{-k}\left(\delta_{\pi, \tau}+\mathcal{O}(d^{-1})\right).
\end{align}
This equation holds because the permutation operators are asymptotically orthogonal in the limit of $d\rightarrow \infty$ \cite{kostenberger2021weingarten,aharonov2022quantum}.

Generating Haar-random unitaries on quantum computers is inefficient, since most unitaries require an exponential number of elementary gates with respect to the system size $n$. This motivates the definition of (approximate) unitary $k$-design, which are unitary ensembles $\mathcal E$ that (approximately) replicates the statistical properties of the Haar measure up to the $k$-th moment. 
Formally, $\mathcal E$ is called a \emph{unitary $k$-design} if $\Phi_{\mathcal E}= \Phi_{\rm H}$, where
\begin{align}\label{eq:designDef}
\Phi_{\mathcal E}(A):= \underset{U\sim \mathcal E}{\E} \left[ U^{\otimes k} A U^{\dagger\otimes k} \right] 
\end{align}  
is the $k$-fold twirling channel with respect to $\mathcal E$.
Furthermore, $\mathcal E$ is called an \emph{$\epsilon$-approximate unitary $k$-design} if 
\begin{align}\label{eq:appdesignDef}
(1-\epsilon)\Phi_{\rm H} 
\preccurlyeq \Phi_{\mathcal E} 
\preccurlyeq (1+\epsilon) \Phi_{\rm H},  
\end{align}
where $\Phi\preccurlyeq\Phi'$ denotes that $\Phi'-\Phi$ is a completely-positive map.
By definition, any ($\epsilon$-approximate) unitary $k$-design is automatically a ($\epsilon$-approximate) unitary $(k-1)$-design.
Recent works \cite{schuster2025random, Laracuente2024approximate} have proved that for constant $k$, an $\epsilon$-approximate unitary $k$-design on $n$ qubits can be efficiently constructed in circuit depth $\mathcal O(\log(n/\epsilon))$, without using ancilla qubits. When ancilla qubits are permitted, the required depth can be further reduced to $\mathcal O(\log\log(n/\epsilon))$ \cite{cui2025unitary}.

\section{Estimation of quantum state moments with  the CBNE protocol}\label{sec:momentCBNE}

In this section, we provide further details on the CBNE protocol for quantum state moment estimation. While Algorithm~\ref{alg:CBNEmain} in the End Matter summarizes the full CBNE protocol, for the analysis in this section we isolate its moment-estimation part and present it as Algorithm~\ref{alg:CBNEmoments}, which should be viewed as a proof-oriented subroutine rather than a separate protocol.

Recall that the quantity $\zeta_k$ [defined in \eref{eq:zeta_def_prl} of the main text] can be expressed as a polynomial of $p_j$ with $j\leq k$: 
\begin{align}\label{eq:zetaExpan}
\zeta_2= \frac{1}{2}+\frac{p_2}{2}, 
\qquad  
\zeta_3= \frac{1}{6}+\frac{p_2}{2}+\frac{p_3}{3},
\qquad 
\zeta_4= \frac{1}{24}+\frac{p_2^2}{8}+\frac{p_2}{4}+\frac{p_3}{3}+\frac{p_4}{4}, \qquad \dots  
\end{align}
Consequently, in Step~10 of Algorithm~\ref{alg:CBNEmoments}, one can sequentially invert these relations with the estimated $\zeta_k$ to obtain the desired moment values: 
\begin{align}
p_2= 2\zeta_2-1, 
\qquad
p_3= 3\zeta_3-\frac{1}{2}-\frac{3p_2}{2}, 
\qquad
p_4= 4\zeta_4- \frac{1}{6}-\frac{p_2^2}{2}-p_2-\frac{4p_3}{3},
\qquad \dots   
\end{align}

The rest of this section is organized as follows: 
In Sec.~\ref{sec:ExpVarCBNEmoment}, we derive the expectation and variance of our estimator $\hat{M}_k^{U}$ in \eref{eq:SMcollision_moment_prl}. In Sec.~\ref{sec:performCBNEmoment}, we prove the performance of CBNE for state moment estimation as stated in \tref{thm:MomOResult} of the main text. In Sec.~\ref{sec:appdesignCBNEmoment}, we prove \pref{prop:appCBNEmoment} in the End Matter, which demonstrates the performance of CBNE for state moment estimation when the unitary ensemble $\mathcal E$ is an approximate unitary design.

\begin{figure}[b]
\begin{algorithm}[H]
{\small
\hspace{-213pt}\textbf{Input:}  $N_U N_M$ sequentially prepared quantum state $\rho$ and integer $t\geq 2$.   \\

\begin{algorithmic}[1]
\caption{{\small CBNE protocol for estimating $\left\{p_t\right\}_{k=2}^t$} \ \ }
\label{alg:CBNEmoments}

\For{$s=1,2,\dots,N_U$,}
\State{Randomly choose a unitary $U_s$ from an ensemble $\mathcal E$. }
\For{$j=1,2,\dots,N_M$,}
\State{Evolve the state $\rho$ using $U_s$ to get $U_s\rho U_s^{\dag}$.} 
\State{Measure the rotated state in the computational basis, obtaining outcome $\hat{b}_j\in\{0,1,\dots,d-1\}$.}
\EndFor 
\State{Compute the following collision estimator for all $k=2,\dots,t$ using the measurement data $\mathbf{b}_{U_s} = \{\hat{b}_1, \dots, \hat{b}_{N_M}\}$:  
\begin{align}\label{eq:SMcollision_moment_prl}
\hat M_k^{U_s}
:=
\frac{\kappa_k}{d\binom{N_M}{k}}
\sum_{i_1<i_2<\cdots<i_k}
\mathbf{1}\{\hat{b}_{i_1}=\cdots=\hat{b}_{i_k}\}, 
\end{align}}

\EndFor 
\State{For $k=2,\dots,t$, compute the average 
    $\hat\zeta_k := \frac{1}{N_U} \sum_{s=1}^{N_U} \hat{M}_k^{U_s}$ as our estimate for $\zeta_k$.} 

\State{Substitute $\hat\zeta_2,\dots,\hat\zeta_t$ into  \eref{eq:zetaExpan}, then solve the equations to obtain estimates $\hat{p}_k$ for $p_k$ ($k=2,\dots,t$).}

\end{algorithmic}
\hspace{-427pt} \textbf{Output:} $\hat{p}_2, \dots, \hat{p}_t$. 
}
\end{algorithm}
\end{figure}

\subsection{Expectation and variance of \texorpdfstring{$\hat{M}_k^{U}$}{}}\label{sec:ExpVarCBNEmoment}
The statistical fluctuation of our estimator $\hat{M}_k^{U}$ (see \eref{eq:SMcollision_moment_prl}) comes from: 
(i) the randomness in the selection of $U$, and 
(ii) the randomness in quantum measurement outcomes $\mathbf{b}_{U} = \{\hat{b}_1, \dots, \hat{b}_{N_M}\}$ according to Born's rule. 
Taking these two sources of randomness into account, the expectation and variance of $\hat{M}_k^{U}$ are 
clarified by the following two lemmas.

\begin{lemma}\label{lam:UnbiasMkU}
Suppose $U$ is a random unitary sampled from a $k$-design ensemble $\mathcal E_{k}$. 
Then $\E \left[ \hat{M}_k^{U}\right]= \zeta_k$. 
\end{lemma}

\begin{lemma}\label{lem:MomVariance}
Suppose $k>1$ is a constant integer independent of the system dimension $d$, and $U$ is a random unitary  sampled from a $2k$-design ensemble $\mathcal E_{2k}$. 
Then 
\begin{align}
	\Var\left[\hat{M}_k^{U}\right] 
	= \bigo{\max\left\lbrace \frac{1}{d} ,\frac{d^{k-1}}{N_{\!M}^{k}}\right\rbrace  } .
\end{align}
\end{lemma}

\begin{proof}[Proof of \lref{lam:UnbiasMkU}]
We have
\begin{align}
\E \left[ \hat{M}_k^{U}\right] 
&=\underset{U }{\E} \, \underset{\,\mathbf{b}_{U}}{\E} \left( \hat{M}_k^{U} \Big| U \right)
=\underset{U }{\E} \left[ \frac{\kappa_k}{d} \sum_{b=0}^{d-1} \Pr\left(\hat{b}_{i_1}=\dots=\hat{b}_{i_k}=b \,|U \right) \right] 
=\frac{\kappa_k}{d} \sum_{b=0}^{d-1} \E_{U} \left[ {\Pr}_{\rho}(b|U)^k \right] 
\nonumber \\
&=\frac{\kappa_k}{d} \sum_{b=0}^{d-1} \E_{U} \left[ \<b|U\rho U^{\dag}|b\>^k \right] 
= \frac{\kappa_k}{d} \sum_{b=0}^{d-1} \Tr \left[ \rho^{\otimes k}  \Phi_{\mathcal E_{k}} \big( \ketbra{b}{b}^{\otimes k} \big) \right]
\stackrel{(a)}{=}  \frac{1}{k!} \sum_{\pi\in \mathrm{S}_k}\Tr\left(\rho^{\otimes k}R_\pi\right)
= \zeta_k , 
\end{align}
where the equality $(a)$ follows from \eref{eq:PhiHaar} and the relation $\Phi_{\mathcal E_{k}}=\Phi_{\rm H}$. 
This confirms \lref{lam:UnbiasMkU}, implying that $\hat{M}_k^{U}$ is an unbiased estimator of $\zeta_k$. 
\end{proof}

\begin{proof}[Proof of \lref{lem:MomVariance}]
The variance of $\hat{M}_k^{U}$ reads
\begin{align}\label{eq:VarhatM} 
	\Var\left[ \hat{M}_k^{U}\right] 
	&= \E\left[ \left( \hat{M}_k^{U}\right) ^2 \right] - \E\left[ \hat{M}_k^{U} \right]^2 
	= \E_{U} \E_{\mathbf{b}_{U}} \left[ \left( \hat{M}_k^{U}\right) ^2 \Big| U \right]
	- \zeta_k^2
	\nonumber\\
	&= \binom{N_M}{k}^{-2} \left( \frac{\kappa_k}{d}\right)^2
	\sum_{\substack{1\leq i_1<i_2<\dots<i_k\leq N_M \\ 1\leq j_1<j_2<\dots<j_k\leq N_M}} 
	\E_{U} \E_{\mathbf{b}_{U}} \Big(
	\textbf{1}\left[\hat{b}_{i_1}=\dots=\hat{b}_{i_k}\right]  
	\textbf{1}\left[\hat{b}_{j_1}=\dots=\hat{b}_{j_k}\right]  \Big| U \Big) - \zeta_k^2. 
\end{align}
where the second equality follows from \lref{lam:UnbiasMkU}. 

Note that $(i_1,\dots,i_k)$ and $(j_1,\dots,j_k)$ are two $k$-tuples drawn from the index set $\{1,2,\dots,N_M\}$. Let ${\rm Co}[(i_1,\dots,i_k);(j_1,\dots,j_k)]$ be the collision number between them, that is, the number of common elements shared by $(i_1,\dots,i_k)$ and $(j_1,\dots,j_k)$. 
For $l\in\{0,1,\dots,k-1\}$, if 
${\rm Co}[(i_1,\dots,i_k);(j_1,\dots,j_k)]=k-l$, we have 
\begin{align}\label{eq:EE1casea}
&\E_{U} \E_{\mathbf{b}_{U}} \Big(
\textbf{1}\left[\hat{b}_{i_1}=\dots=\hat{b}_{i_k}\right]  
\textbf{1}\left[\hat{b}_{j_1}=\dots=\hat{b}_{j_k}\right]  \Big| U \Big)
= \E_{U} \E_{\mathbf{b}_{U}} \Big(
\textbf{1}\left[\hat{b}_{i_1}=\dots=\hat{b}_{i_k}=\hat{b}_{j_1}=\dots=\hat{b}_{j_k}\right] \Big| U \Big) \nonumber\\
&=
\E_{U} \sum_{b=0}^{d-1} \<b|U\rho U^{\dag}|b\>^{k+l}=\sum_{b=0}^{d-1} \Tr \left[ \rho^{\otimes (k+l)}  \Phi_{\mathcal E_{2k}} \big( \ketbra{b}{b}^{\otimes (k+l)} \big) \right] 
=\frac{d}{\kappa_{k+l}} \Tr \left[ \rho^{\otimes {(k+l)}} \Pi_{\mathrm{sym}}^{(k+l)} \right] 
= \bigo{d^{1-k-l}}, 
\end{align}
where the last equality follows because $\Tr \left[ \rho^{\otimes {(k+l)}} \Pi_{\mathrm{sym}}^{(k+l)} \right] \leq 1$ and that 
$\kappa_{k+l}=\binom{k+l+d-1}{k+l}=\Theta(d^{k+l})$ when $k$ and $l$ are both constant independent of $d$. 
If ${\rm Co}[(i_1,\dots,i_k);(j_1,\dots,j_k)]=0$, we have 
\begin{align}\label{eq:EE1caseb}
&\E_{U} \E_{\mathbf{b}_{U}} \Big(
\textbf{1}\left[\hat{b}_{i_1}=\dots=\hat{b}_{i_k}\right]  
\textbf{1}\left[\hat{b}_{j_1}=\dots=\hat{b}_{j_k}\right]  \Big| U \Big)
= 
\E_{U} \sum_{b,b'=0}^{d-1} \<b|U\rho U^{\dag}|b\>^{k} \<b'|U\rho U^{\dag}|b'\>^{k}
\nonumber\\
&=
\E_{U} \sum_{b=0}^{d-1} \<b|U\rho U^{\dag}|b\>^{2k}
+\E_{U} \sum_{b\ne b'} \<b|U\rho U^{\dag}|b\>^{k} \<b'|U\rho U^{\dag}|b'\>^{k}
=d^{2-2k} k!^2\zeta_k^2  +\bigo{d^{1-2k}},
\end{align}
where the last equality holds because 
\begin{align}
	\E_{U} \sum_{b=0}^{d-1} \<b|U\rho U^{\dag}|b\>^{2k}
	=\sum_{b=0}^{d-1} \Tr \left[ \rho^{\otimes 2k}  \Phi_{\mathcal E_{2k}} \big( \ketbra{b}{b}^{\otimes 2k} \big) \right]
	= \frac{d}{\kappa_{2k}} \tr{\rho^{\otimes 2k} \Pi_{\mathrm{sym}}^{2k} }
	= \mathcal{O}(d^{1-2k})  
\end{align}
and  
\begin{align}
&\E_{U} \sum_{b\ne b'} \<b|U\rho U^{\dag}|b\>^{k} \<b'|U\rho U^{\dag}|b'\>^{k}
\nonumber\\
&=  \sum_{b\ne b'} \Tr\left[ \left( \ketbra{b}{b}^{\otimes k}\otimes \ketbra{b'}{b'}^{\otimes k}\right)  
\Phi_{\mathcal E_{2k}}\left( \rho^{\otimes 2k} \right) \right] 
\nonumber\\
&\stackrel{(a)}{=}\sum_{b\ne b'} \sum_{\tau,\pi\in \mathrm{S}_{2k}} \Wg_{\pi,\tau}^{(2k,d)} 
 \Tr\left[ \left( \ketbra{b}{b}^{\otimes k}\otimes \ketbra{b'}{b'}^{\otimes k}\right) R_{\tau} \right] 
 \Tr\left( \rho^{\otimes 2k}  R_{\pi}^{\top} \right) 
\nonumber\\
&\stackrel{(b)}{=}\sum_{b\ne b'} \sum_{\pi\in \mathrm{S}_{2k}} \sum_{\tau_1,\tau_2\in \mathrm{S}_k}
\Wg_{\pi,(\tau_1,\tau_2)}^{(2k,d)}
\Tr\left( \rho^{\otimes 2k}  R_{\pi} \right) 
\nonumber\\
&=d(d-1)\sum_{\tau_1,\tau_2\in \mathrm{S}_k}
\Wg_{(\tau_1,\tau_2),(\tau_1,\tau_2)}^{(2k,d)}  
\Tr\left( \rho^{\otimes 2k}  R_{(\tau_1,\tau_2)} \right) 
+d(d-1)\sum_{\tau_1,\tau_2\in \mathrm{S}_k}\sum_{\substack{\pi\in\mathrm{S}_{2k}\\ \pi\ne(\tau_1,\tau_2)}} 
\Wg_{\pi,(\tau_1,\tau_2)}^{(2k,d)}  
\Tr\left( \rho^{\otimes 2k}  R_{\pi} \right)
\nonumber\\
&\stackrel{(c)}{=}d(d-1)\sum_{\tau_1,\tau_2\in \mathrm{S}_k}
\left[ d^{-2k}+ \bigo{d^{-2k-1}}\right]  
\Tr\left( \rho^{\otimes k}  R_{\tau_1} \right) \Tr\left( \rho^{\otimes k}  R_{\tau_2} \right) 
+d(d-1)\sum_{\tau_1,\tau_2\in \mathrm{S}_k}\sum_{\substack{\pi\in\mathrm{S}_{2k}\\ \pi\ne(\tau_1,\tau_2)}}  \bigo{d^{-2k-1}}
\nonumber\\
&= 
d^{2-2k} \left[ \sum_{\tau \in \mathrm{S}_{k}} \Tr\left(\rho^{\otimes t}R_\tau\right) \right]^2 
+\bigo{d^{1-2k}}
\nonumber\\
&=d^{2-2k} k!^2\zeta_k^2  +\bigo{d^{1-2k}}.
\end{align}
Here,
$(a)$ follows from \eref{eq:tirling}; 
$(b)$ holds because $\Tr\big[ \big( \ketbra{b}{b}^{\otimes k}\!\otimes \ketbra{b'}{b'}^{\otimes k}\big) R_{\tau} \big]=1$ if $\tau=(\tau_1,\tau_2)$ for some $\tau_1,\tau_2\in\mathrm{S}_{k}$, and 0 otherwise; 
$(c)$ follows from \eref{eq:Weingarten}.   

The variance in \eref{eq:VarhatM}  is then upper bounded by
\begin{align}\label{eq:EE1cased}
&\Var\left[ \hat{M}_k^{U}\right] +\zeta_k^2
\nonumber\\
&=  
\binom{N_M}{k}^{-2} \left( \frac{\kappa_k}{d}\right)^2
\sum_{l=0}^{k} \,
\sum_{\substack{i_1<i_2<\dots<i_k, j_1<j_2<\dots<j_k, \\{\rm Co}[(i_1,\dots,i_k);(j_1,\dots,j_k)]=k-l }} 
\E_{U} \E_{\mathbf{b}_{U}} \Big(
\textbf{1}\left[\hat{b}_{i_1}=\dots=\hat{b}_{i_k}\right]  
\textbf{1}\left[\hat{b}_{j_1}=\dots=\hat{b}_{j_k}\right]  \Big| U \Big)
\nonumber\\
&= 
\binom{N_M}{k}^{-2} \left( \frac{\kappa_k}{d}\right)^2\,  
\sum_{l=0}^{k} 
\binom{N_M}{k+l} \binom{k+l}{k} \binom{k}{l} \, 
\E_{U} \E_{\mathbf{b}_{U}} \Big(
\textbf{1}\left[\hat{b}_{i_1}=\dots=\hat{b}_{i_k}\right]  
\textbf{1}\left[\hat{b}_{j_1}=\dots=\hat{b}_{j_k}\right]  \Big| U \Big)\Big|_{{\rm Co}[(i_1,\dots,i_k);(j_1,\dots,j_k)]=k-l}
\nonumber\\
&= 
\binom{N_M}{k}^{-2} \left( \frac{\kappa_k}{d}\right)^2
\left\lbrace \,  
\sum_{l=0}^{k-1} 
\binom{N_M}{k+l} \binom{k+l}{k} \binom{k}{l} \, 
\bigo{d^{1-k-l}}
+   
\binom{N_M}{2k} \binom{2k}{k}
\left[d^{2-2k} k!^2\zeta_k^2  +\bigo{d^{1-2k}}
\right]\right\rbrace 
\nonumber\\
&= 
\sum_{l=0}^{k-1}\binom{N_M}{k}^{-2}  
\binom{N_M}{k+l} \mathcal{O}(d^{k-1-l})
+ 
\binom{N_M}{k}^{-2}  \binom{N_M}{2k}\binom{2k}{k} \left( \frac{k!\kappa_k}{d^{k}} \right)^2 
\left[\zeta_k^2 
+ \bigo{d^{-1}}
\right], 
\end{align}
where the third equality follows from Eqs.~\eqref{eq:EE1casea} and \eqref{eq:EE1caseb}. 
Note that 
\begin{align}
\frac{k!\kappa_k}{d^{k}}
= \frac{k!}{d^{k}}\binom{k+d-1}{k}
= \left( 1+ \frac{k-1}{d} \right) \left( 1+ \frac{k-2}{d} \right) \cdots \left( 1+ \frac{1}{d} \right) 
= 1+\bigo{d^{-1}}
\end{align}
and 
\begin{align} 
	\binom{N_M}{k}^{-2}  \binom{N_M}{2k}\binom{2k}{k}
	&= \left[ \frac{N_M!}{k!(N_M-k)!}\right]^{-2}  \frac{N_M!}{(2k)!(N_M-2k)!} \cdot \frac{(2k)!}{(k)!(k)!}
	= \frac{(N_M-k)!^2}{N_M!(N_M-2k)!}  
	\nonumber\\
	&
	= \left( \frac{N_M-k}{N_M}\right)  \left( \frac{N_M-k-1}{N_M-1}\right)  \cdots \left(\frac{N_M-2k+1}{N_M-k+1}\right) 
	\leq 1.
\end{align}
By plugging these into \eref{eq:EE1cased}, we obtain 
\begin{align}
\Var\left[ \hat{M}_k^{U}\right] +\zeta_k^2
\leq 
\sum_{l=0}^{k-1} \bigo{\frac{d^{k-l-1}}{N_M^{k-l}}} 
+ \left[1+\bigo{d^{-1}}\right]^2 \left[ \zeta_k^2 + \mathcal{O}(d^{-1})\right] 
= 
\bigo{\max\left\lbrace \frac{1}{d} ,\frac{d^{k-1}}{N_{\!M}^{k}}\right\rbrace  } 
+\zeta_k^2 ,
\end{align}
which confirms \lref{lem:MomVariance}. 
\end{proof}

\subsection{Proof of CBNE's performance in state moment estimation}\label{sec:performCBNEmoment}

When applied to state moment estimation, 
the performance of the CBNE protocol is summarized in the following theorem, which is a restatement of \tref{thm:MomOResult} from the main text in the special case of $O=\openone$. 

\begin{theorem}\label{thm:MomResult}
Suppose $\rho\in\caD(\caH)$ and $t>1$ is a constant integer independent of the system dimension $d$.
For any $0<\epsilon<1$, the CBNE protocol (in Algorithm~\ref{alg:CBNEmoments}) can return $\epsilon$-additive error estimates for all 
$p_2,p_3,\dots,p_t$ with high probability, provided that:
$(i)$ the random unitaries are sampled from a $2t$-design ensemble, 
\begin{align}
(ii)\ N_U= \max\left\lbrace 1, \dfrac{c_1}{d\epsilon^2} \right\rbrace  
,\quad\text{and}\quad 
(iii)\ N_M\geq  c_2  \min\left\lbrace \frac{d^{1-1/t}}{\epsilon^{2/t}}, d \right\rbrace. 
\end{align}
Here $c_1,c_2>0$ are some constants independent of $d$ and $\epsilon$. 
\end{theorem}

\begin{proof}[Proof of \tref{thm:MomResult}]
Let $0<\delta<1$ be a constant. 
According to the union bound and \lref{lem:EstimatorTrans} below, in order to estimate all 
$p_2,\dots,p_t$ within $\epsilon$ additive error and failure probability at most $\delta$, it 
suffices to ensure that the following condition holds for all $k=2,\dots,t$:  
\begin{equation}\label{eq:EachUBceps}
\Pr\left\{ \left|\hat\zeta_k-\zeta_k\right| > c\epsilon \right\} \leq \frac{\delta}{t},
\end{equation}
where $c>0$ is some constant factor independent of $\epsilon$ and $d$.  

By \lref{lam:UnbiasMkU}, $\hat\zeta_k = \frac{1}{N_U} \sum_{s=1}^{N_U} \hat{M}_k^{U_s}$ is an unbiased estimator for $\zeta_k$. 
Then Chebyshev's inequality implies that 
\begin{align}
\Pr\left\{ \left|\hat\zeta_k-\zeta_k\right| > c\epsilon \right\}
\leq \frac{1}{c^2\epsilon^2} \Var\left[ \hat\zeta_k\right] 
= \frac{1}{c^2\epsilon^2 N_U} \Var\!\Big[ \hat{M}_k^{U} \Big] 
= \bigo{\frac{1}{\epsilon^2 N_U} \cdot \max\left\lbrace \frac{1}{d} ,\frac{d^{k-1}}{N_{\!M}^{k}}\right\rbrace},  
\end{align}
where the last equality follows from \lref{lem:MomVariance}.
Hence, to make \eref{eq:EachUBceps} hold, it is sufficient to take 
\begin{align} 
\frac{1}{\epsilon^2 N_U d} \leq 1
\quad\text{and}\quad  
\frac{d^{k-1}}{\epsilon^2 N_U N_{\!M}^{k}} \leq 1 ,
\end{align}
where we omit constant coefficients $c,t,\delta$. 
It is straightforward to verify that these two conditions are satisfied by choosing $N_U$ and $N_M$ as in \tref{thm:MomResult}. 
This completes the proof. 
\end{proof}

\begin{lemma}\label{lem:EstimatorTrans}
Suppose $t\geq 2$ is a constant integer independent of the system dimension $d$. For any $0<\eta<1$, if the estimates $\{\hat{\zeta}_k\}_{k=2}^t$ satisfy
	\begin{align}\label{eq:Cond1ptTrans}
		\left| \hat{\zeta}_k - \zeta_k \right| \leq \eta \quad \forall k = 2, \dots, t,
	\end{align} 
	then $\{\hat{p}_k\}_{k=2}^t$ computed in Step~10 of Algorithm~\ref{alg:CBNEmoments} satisfy
	$|\hat{p}_k - p_k| = \mathcal{O}(\eta)$ for all $k=2,\dots,t$.
\end{lemma}

\begin{proof}[Proof of \lref{lem:EstimatorTrans}]
We proceed by induction on $t$. 

\textbf{Base case ($t=2$):} 
From the relation $\zeta_2= 1/2+p_2/2$, we see that estimating $\zeta_2$ within $\eta$-additive error directly yields an $\mathcal{O}(\eta)$-additive error estimate for $p_2$. This establishes the lemma for $t=2$.

\textbf{Inductive step ($t \geq 3$):} 
Now we assume that $t\geq3$ is a constant integer and the desired result holds for all integers smaller than $t$. 
For any permutation $\pi \in \mathrm{S}_t$, we have 
\begin{align}\label{eq:TRrhotRpiexp}
\Tr\left(\rho^{\otimes t}R_\pi\right)
=
\begin{cases}
	p_t &  \text{when } \#(\pi)=1, \\
	\vspace{-1.2em} \\
	\prod_{i=2}^{t-1} p_i^{\nu_i(\pi)}   &  \text{when } \#(\pi)\geq 2.
\end{cases}
\end{align}
Here $\#(\pi)$ is the number of disjoint cycles in $\pi$, and the exponent $\nu_i(\pi)$ denotes the number of disjoint cycles of length $i$ in $\pi$, which satisfies $0\leq \nu_i(\pi)<t$ (for $i=2,\dots,t-1$) when $\#(\pi)\geq2$. 
Combining \eref{eq:TRrhotRpiexp} with \eref{eq:SMcollision_moment_prl}, we derive:
\begin{align}
t!\,\zeta_t
=\sum_{\pi\in \mathrm{S}_t}\Tr\left(\rho^{\otimes t}R_\pi\right)
=
(t-1)!\, p_t 
+ \sum_{\pi\in \mathrm{S}_t, \#(\pi)\geq 2} \,\prod_{i=2}^{t-1} p_i^{\nu_i(\pi)}, 
\end{align}
where the first term uses the fact that the number of permutations in $\mathrm{S}_t$ with $\#(\pi)=1$ is $(t-1)!$.
Rearranging this relation yields:
\begin{align}\label{eq:ptExp1}
p_t = t\,\zeta_t - \frac{1}{(t-1)!} \sum_{\pi\in \mathrm{S}_t, \#(\pi)\geq 2} \,\prod_{i=2}^{t-1} p_i^{\nu_i(\pi)}.   
\end{align}

By the induction hypothesis, if \eref{eq:Cond1ptTrans} holds, then the estimates $\{\hat{p}_k\}_{k=2}^{t-1}$ obtained from $\{\hat{\zeta}_k\}_{k=2}^{t-1}$ satisfy $|\hat{p}_k - p_k| = \mathcal{O}(\eta)$. 
Since the RHS of \eref{eq:ptExp1} is a summation of a constant number of terms, 
substituting these $\{\hat{p}_k\}_{k=2}^{t-1}$ and $\hat{\zeta}_t$ into \eref{eq:ptExp1} produces an estimate $\hat{p}_t$ for $p_t$ within error $\mathcal{O}(\eta)$. This completes the inductive step and the proof.
\end{proof}

\subsection{Moment estimation with an approximate unitary design ensemble}\label{sec:appdesignCBNEmoment}

This subsection aims to prove \pref{prop:appCBNEmoment} in the End Matter.
To this end, we first analyze the expectation and variance of our estimator $\hat{M}_k^{U}$ (see \eref{eq:SMcollision_moment_prl}) in the case where $U$ is drawn from an approximate unitary design ensemble.

\begin{lemma}\label{lam:biasMkUapp}
Suppose $U$ is a random unitary sampled from a $\mu$-approximate $k$-design ensemble $\mathcal E_{k}^\mu$, then 
\begin{align}
 \left| \underset{U\sim \mathcal E_{k}^\mu,  \mathbf{b}_{U}}{\E} \left[ \hat{M}_k^{U}\right] -\zeta_k \right|  \leq \mu\, \zeta_k. 
\end{align}
\end{lemma}

\begin{proof}[Proof of \lref{lam:biasMkUapp}]
We have
\begin{align}
\left| \underset{U\sim \mathcal E_{k}^\mu,  \mathbf{b}_{U}}{\E}\left[ \hat{M}_k^{U}\right] -\zeta_k  \right| 
&\stackrel{(a)}{=} 
\left| 
 \underset{U\sim \mathcal E_{k}}{\E} \, \underset{\,\mathbf{b}_{U}}{\E} \left( \hat{M}_k^{U} \Big| U \right)
-\underset{U\sim \mathcal E_{k}^\mu}{\E} \, \underset{\,\mathbf{b}_{U}}{\E} \left( \hat{M}_k^{U} \Big| U \right)
\right|
\nonumber \\
&=
\left| 
 \frac{\kappa_k}{d}\, \underset{U\sim \mathcal E_{k}}{\E}  \left( \, \sum_{b=0}^{d-1} \<b|U\rho U^{\dag}|b\>^k \right) 
-\frac{\kappa_k}{d}\, \underset{U\sim \mathcal E_{k}^\mu}{\E} \left(  \, \sum_{b=0}^{d-1} \<b|U\rho U^{\dag}|b\>^k \right) 
\right|
\nonumber \\
&\leq \frac{\kappa_k}{d} \sum_{b=0}^{d-1}  \,
\left| 
 \Tr \left[ \rho^{\otimes k}  \Phi_{\mathcal E_{k}} \big( \ketbra{b}{b}^{\otimes k} \big) \right] 
-\Tr \left[ \rho^{\otimes k}  \Phi_{\mathcal E_{k}^\mu} \big( \ketbra{b}{b}^{\otimes k} \big) \right]  
\right| 
\nonumber \\
&\stackrel{(b)}{\leq}  \frac{\kappa_k}{d} \sum_{b=0}^{d-1}  \mu\,\Tr \left[ \rho^{\otimes k}  \Phi_{\mathcal E_{k}} \big( \ketbra{b}{b}^{\otimes k} \big) \right]
\nonumber \\
&\stackrel{(c)}{=} \frac{\mu}{k!} \sum_{\pi\in \mathrm{S}_k}\Tr\left(\rho^{\otimes k}R_\pi\right)
= \mu\, \zeta_k , 
\end{align}
which confirms \lref{lam:biasMkUapp}. 
Here, $(a)$ follows from \lref{lam:UnbiasMkU};  
$(b)$ follows from the definition of the $\mu$-approximate unitary $k$-design in \eref{eq:appdesignDef}; $(c)$ follows from \eref{eq:PhiHaar} and the relation $\Phi_{\mathcal E_{k}}=\Phi_{\rm H}$. 
\end{proof}

\begin{lemma}\label{lem:MomVarianceApp}
Suppose $k>1$ is a constant integer and $U$ is a random unitary sampled from a $\mu$-approximate $2k$-design ensemble $\mathcal E_{2k}^\mu$ with $\mu\le1$. Then 
\begin{align}
\underset{U\sim \mathcal E_{2k}^\mu,  \mathbf{b}_{U}}{\Var} \left[\hat{M}_k^{U}\right] 
=  \bigo{\max\left\lbrace \frac{1}{d} ,\frac{d^{k-1}}{N_{\!M}^{k}}\right\rbrace } +3\mu.
\end{align}
\end{lemma}

\begin{proof}[Proof of \lref{lem:MomVarianceApp}]
Similar to \eref{eq:VarhatM}, 
the variance of $\hat{M}_k^{U}$ can be expanded as 
\begin{align}\label{eq:VarhatMapp} 
&\underset{U\sim \mathcal E_{2k}^\mu,  \mathbf{b}_{U}}{\Var} \left[ \hat{M}_k^{U}\right] 
= \underset{U\sim \mathcal E_{2k}^\mu}{\E} \,
\underset{\mathbf{b}_{U}}{\E}
\left[ \left( \hat{M}_k^{U}\right) ^2 \Big| U \right] 
- \left( \underset{U\sim \mathcal E_{2k}^\mu,  \mathbf{b}_{U}}{\E} \left[\hat{M}_k^{U}\right] \right)^2 
\nonumber\\
&= 
\binom{N_M}{k}^{-2} \left( \frac{\kappa_k}{d}\right)^2
\sum_{\substack{1\leq i_1<\dots<i_k\leq N_M \\ 1\leq j_1<\dots<j_k\leq N_M}} 
\underset{U\sim \mathcal E_{2k}^\mu}{\E} \, \underset{\,\mathbf{b}_{U}}{\E} \Big(
\textbf{1}\left[\hat{b}_{i_1}=\dots=\hat{b}_{i_k}\right]  
\textbf{1}\left[\hat{b}_{j_1}=\dots=\hat{b}_{j_k}\right]  \Big| U \Big) 
- \left( \underset{U\sim \mathcal E_{2k}^\mu,  \mathbf{b}_{U}}{\E} \left[\hat{M}_k^{U}\right] \right)^2 
\nonumber\\
&\stackrel{(a)}{\leq}  
\binom{N_M}{k}^{-2} \left( \frac{\kappa_k}{d}\right)^2
\sum_{\substack{1\leq i_1<\dots<i_k\leq N_M \\ 1\leq j_1<\dots<j_k\leq N_M}} 
(1+\mu) \underset{U\sim \mathcal E_{2k}}{\E} \, \underset{\,\mathbf{b}_{U}}{\E} \Big(
\textbf{1}\left[\hat{b}_{i_1}=\dots=\hat{b}_{i_k}\right]  
\textbf{1}\left[\hat{b}_{j_1}=\dots=\hat{b}_{j_k}\right]  \Big| U \Big) 
- \big[ (1-\mu)\zeta_k\big]^2
\nonumber\\
&\stackrel{(b)}{=} 
(1+\mu) \left(\underset{U\sim \mathcal E_{2k},\mathbf{b}_{U}}{\Var} \left[ \hat{M}_k^{U}\right] + \zeta_k^2\right)
- (1-\mu)^2\zeta_k^2
\nonumber\\
&= 
(1+\mu)  \underset{U\sim \mathcal E_{2k},\mathbf{b}_{U}}{\Var} \left[ \hat{M}_k^{U}\right] + \mu(3-\mu)\zeta_k^2
\nonumber\\
&\stackrel{(c)}{=}
\bigo{\max\left\lbrace \frac{1}{d} ,\frac{d^{k-1}}{N_{\!M}^{k}} \right\rbrace } +3\mu. 
\end{align}
Here, $(a)$ holds because $\underset{U\sim \mathcal E_{2k}^\mu,\mathbf{b}_{U}}{\E} \left[\hat{M}_k^{U}\right]\geq (1-\mu)\zeta_k\geq 0$ (see \lref{lam:biasMkUapp}) and 
\begin{align}\label{eq:EE11UB}
\underset{U\sim \mathcal E_{2k}^{\mu}}{\E} \, \underset{\,\mathbf{b}_{U}}{\E} \Big(
\textbf{1}\left[\hat{b}_{i_1}=\dots=\hat{b}_{i_k}\right]  
\textbf{1}\left[\hat{b}_{j_1}=\dots=\hat{b}_{j_k}\right]  \Big| U \Big) 
\leq 
(1+\mu) \underset{U\sim \mathcal E_{2k}}{\E} \, \underset{\,\mathbf{b}_{U}}{\E} \Big(
\textbf{1}\left[\hat{b}_{i_1}=\dots=\hat{b}_{i_k}\right]  
\textbf{1}\left[\hat{b}_{j_1}=\dots=\hat{b}_{j_k}\right]  \Big| U \Big),
\end{align}	
which is proved below; $(b)$ follows from \eref{eq:VarhatM}, and $(c)$ follows from \lref{lem:MomVariance}. 

To prove \eref{eq:EE11UB}, we consider two cases. First,  if 
${\rm Co}[(i_1,\dots,i_k);(j_1,\dots,j_k)]=k-l$ with $l=0,\dots,k-1$, we have 
\begin{align}
&\left| 
\underset{U\sim \mathcal E_{2k}^\mu}{\E} \, \underset{\,\mathbf{b}_{U}}{\E} \Big(
\textbf{1}\left[\hat{b}_{i_1}=\dots=\hat{b}_{i_k}\right]  
\textbf{1}\left[\hat{b}_{j_1}=\dots=\hat{b}_{j_k}\right]  \Big| U \Big) 
-
\underset{U\sim \mathcal E_{2k}}{\E} \, \underset{\,\mathbf{b}_{U}}{\E} \Big(
\textbf{1}\left[\hat{b}_{i_1}=\dots=\hat{b}_{i_k}\right]  
\textbf{1}\left[\hat{b}_{j_1}=\dots=\hat{b}_{j_k}\right]  \Big| U \Big) 
\right| 
\nonumber\\&= 
\left| 
 \underset{U\sim \mathcal E_{2k}^\mu}{\E}  \sum_{b=0}^{d-1} \<b|U\rho U^{\dag}|b\>^{k+l}
-\underset{U\sim \mathcal E_{2k}}{\E}  \sum_{b=0}^{d-1} \<b|U\rho U^{\dag}|b\>^{k+l}
\right| 
\nonumber\\&\leq 
\sum_{b=0}^{d-1}\, \left| 
  \Tr \left[ \rho^{\otimes (k+l)}  \Phi_{\mathcal E_{2k}^\mu} \big( \ketbra{b}{b}^{\otimes (k+l)} \big) \right] 
- \Tr \left[ \rho^{\otimes (k+l)}  \Phi_{\mathcal E_{2k}} \big( \ketbra{b}{b}^{\otimes (k+l)} \big) \right] 
\right| 
\nonumber\\
&\leq  
\sum_{b=0}^{d-1}\, \mu\,
\Tr \left[ \rho^{\otimes (k+l)}  \Phi_{\mathcal E_{2k}} \big( \ketbra{b}{b}^{\otimes (k+l)} \big) \right] 
\nonumber\\
&= 
\mu \underset{U\sim \mathcal E_{2k}}{\E} \, \underset{\,\mathbf{b}_{U}}{\E} \Big(
\textbf{1}\left[\hat{b}_{i_1}=\dots=\hat{b}_{i_k}\right]  
\textbf{1}\left[\hat{b}_{j_1}=\dots=\hat{b}_{j_k}\right]  \Big| U \Big), 
\end{align}
which confirms \eref{eq:EE11UB}.
Here the second inequality follows from \eref{eq:appdesignDef}. 
Second, if ${\rm Co}[(i_1,\dots,i_k);(j_1,\dots,j_k)]=0$, we have 
\begin{align}
&\left| 
\underset{U\sim \mathcal E_{2k}^\mu}{\E} \, \underset{\,\mathbf{b}_{U}}{\E} \Big(
\textbf{1}\left[\hat{b}_{i_1}=\dots=\hat{b}_{i_k}\right]  
\textbf{1}\left[\hat{b}_{j_1}=\dots=\hat{b}_{j_k}\right]  \Big| U \Big) 
-
\underset{U\sim \mathcal E_{2k}}{\E} \, \underset{\,\mathbf{b}_{U}}{\E} \Big(
\textbf{1}\left[\hat{b}_{i_1}=\dots=\hat{b}_{i_k}\right]  
\textbf{1}\left[\hat{b}_{j_1}=\dots=\hat{b}_{j_k}\right]  \Big| U \Big) 
\right| 
\nonumber\\&= 
\left| 
 \underset{U\sim \mathcal E_{2k}^\mu}{\E}  \sum_{b,b'=0}^{d-1} \<b|U\rho U^{\dag}|b\>^{k} \<b'|U\rho U^{\dag}|b'\>^{k}
-\underset{U\sim \mathcal E_{2k}}{\E}  \sum_{b,b'=0}^{d-1} \<b|U\rho U^{\dag}|b\>^{k} \<b'|U\rho U^{\dag}|b'\>^{k} \right| 
\nonumber\\
&\leq 
\sum_{b,b'=0}^{d-1}\, \left| 
\Tr \left[ \rho^{\otimes 2k}  \Phi_{\mathcal E_{2k}^\mu} \big( \ketbra{b}{b}^{\otimes k}\!\otimes \ketbra{b'}{b'}^{\otimes k} \big) \right] 
- \Tr \left[ \rho^{\otimes 2k}  \Phi_{\mathcal E_{2k}} \big( \ketbra{b}{b}^{\otimes k}\!\otimes \ketbra{b'}{b'}^{\otimes k} \big) \right] 
\right| 
\nonumber\\
&\leq  
\sum_{b,b'=0}^{d-1}\, \mu\,
\Tr \left[ \rho^{\otimes 2k}  \Phi_{\mathcal E_{2k}} \big( \ketbra{b}{b}^{\otimes k}\!\otimes \ketbra{b'}{b'}^{\otimes k} \big) \right] 
\nonumber\\
&= 
\mu \underset{U\sim \mathcal E_{2k}}{\E} \, \underset{\,\mathbf{b}_{U}}{\E} \Big(
\textbf{1}\left[\hat{b}_{i_1}=\dots=\hat{b}_{i_k}\right]  
\textbf{1}\left[\hat{b}_{j_1}=\dots=\hat{b}_{j_k}\right]  \Big| U \Big) ,
\end{align}
which confirms \eref{eq:EE11UB} again.
This completes the proof of \lref{lem:MomVarianceApp}.
\end{proof}

We now prove Proposition~\ref{prop:appCBNEmoment} via Lemmas~\ref{lam:biasMkUapp} and \ref{lem:MomVarianceApp}.

\begin{prop}[Restatement of \pref{prop:appCBNEmoment}]
\label{prop:appCBNEmomentSM}
Suppose $\rho\in\caD(\caH)$ and $t>1$ is a constant integer independent of the system dimension $d$.
For any $0<\epsilon<1$, the CBNE protocol (in Algorithm~\ref{alg:CBNEmoments}) can return $\epsilon$-additive error estimates for all 
$p_2,p_3,\dots,p_t$ with high probability, provided that: $(i)$ the random unitaries are sampled from a $(c_3\epsilon^2)$-approximate $2t$-design ensemble, 
\begin{align}
(ii) \ N_U= \max\left\lbrace 1, \dfrac{c_1}{d\epsilon^2} \right\rbrace ,  
\quad \text{and} \quad 
(iii) \ N_M\geq  c_2  \min\left\lbrace \frac{d^{1-1/t}}{\epsilon^{2/t}}, d \right\rbrace,     
\end{align}
where $c_1,c_2,c_3>0$ are constants  independent of $d$ and $\epsilon$. 
\end{prop}

\begin{proof}[Proof of \pref{prop:appCBNEmomentSM}]
Let $0<\delta<1$ be a constant. 
According to the union bound and \lref{lem:EstimatorTrans}, in order to estimate all 
$p_2,\dots,p_t$ within $\epsilon$ additive error and failure probability at most $\delta$, it 
suffices to ensure that the following condition holds for all $k=2,3,\dots,t$: 
\begin{equation}\label{eq:EachUBcepsAPP}
	\Pr\left\{ \left|\hat\zeta_k-\zeta_k\right| > c\epsilon \right\} \leq \frac{\delta}{t},
\end{equation}
where $0<c<1$ is some constant factor independent of $\epsilon$ and $d$.  
Note that \eref{eq:EachUBcepsAPP} can be achieved if 
\begin{align}\label{eq:CondiAppMom}
\left|\E[\hat\zeta_k]-\zeta_k\right| < \frac{c\epsilon}{2} 
\quad\text{and}\quad
\Pr\left\{ \left|\hat\zeta_k-\E[\hat\zeta_k]\right| > \frac{c\epsilon}{2}  \right\} \leq \frac{\delta}{t} \quad \forall k=2,3,\dots,t. 
\end{align} 

Assume that the random unitaries used in the CBNE protocol are sampled from a 
$(c^2\delta\epsilon^2/20t)$-approximate $2t$-design ensemble. 
Then \lref{lam:biasMkUapp} implies that
\begin{align}
 \left|\E[\hat\zeta_k]-\zeta_k\right| 
=\Big|\E\left[ \hat M^U_k\right] -\zeta_k\Big| 
\leq \frac{c^2\delta \epsilon^2}{20 t}
\leq \frac{c\epsilon}{2}. 
\end{align} 
In addition, Chebyshev's inequality implies that 
\begin{align}\label{eq:ChebyY}
\Pr\left\{ \left|\hat\zeta_k-\E[\hat\zeta_k]\right| > \frac{c\epsilon}{2}  \right\} 
&\leq \frac{4}{c^2\epsilon^2}\Var[\hat\zeta_k]  
= \frac{4}{c^2\epsilon^2} \cdot \frac{1}{N_U} 
\left[ \bigo{\max\left\lbrace \frac{1}{d} ,\frac{d^{k-1}}{N_{\!M}^{k}} \right\rbrace } 
       + \frac{3c^2\delta \epsilon^2}{20 t} \right] 
\nonumber\\
&\leq\frac{1}{\epsilon^2 N_U} \bigo{\max\left\lbrace \frac{1}{d} ,\frac{d^{k-1}}{N_{\!M}^{k}} \right\rbrace } 
     + \frac{3\delta}{ 5t }, 
\end{align} 
where the second inequality follows from \lref{lem:MomVarianceApp}. 
Therefore, to make sure that \eref{eq:CondiAppMom} holds, it is sufficient to take 
\begin{align} 
\frac{1}{\epsilon^2 dN_U} \leq 1
\quad\text{and}\quad  
\frac{d^{k-1}}{\epsilon^2  N_U N_{\!M}^{k}} \leq 1 , 
\end{align}
where we omit constant coefficients, including $\delta$ and $t$. 
It is straightforward to verify that these two conditions are satisfied by choosing $N_U$ and $N_M$ as in \eref{eq:MomResultApp}. 
This completes the proof of \pref{prop:appCBNEmomentSM}. 
\end{proof}

\section{Estimation of \texorpdfstring{$\Tr(O\rho^t)$}{} with  the CBNE protocol}\label{sec:CBNEmain}

In this section, we provide further details on the CBNE protocol for estimating nonlinear quantities of the form \(\Tr(O\rho^t)\).
The complete protocol is summarized in Algorithm~\ref{alg:CBNEmain} in the End Matter.
We also present in Algorithm~\ref{alg:variantCBNE} a closely related variant of the CBNE protocol, which will be used in the analysis with approximate unitary designs.

The following quantity [introduced in \eref{eq:xikO0def} of the main text] will be frequently used: 
\begin{equation}
\xi_k^{O} := \frac{1}{(k+1)! } \sum_{\pi\in \mathrm{S}_{k+1}}  \Tr\left[ \left(\rho^{\otimes k} \otimes O \right)R_\pi\right] ,
\end{equation}
which 
can be expressed as a polynomial of $p_i$ and $\Tr(O\rho^j)$ with $1\leq i,j\leq k$: 
\begin{align}\label{eq:xiORelation+}
\xi_1^{O}&=  \frac{\Tr(O\rho)+\Tr(O)}{2}, 
\nonumber\\ 
\xi_2^{O}&=  \frac{\Tr(O\rho)+\Tr(O\rho^2)}{3}+\frac{(1+p_2)\Tr(O)}{6},
\nonumber\\ 
\xi_3^{O}&=  \frac{(1+p_2)\Tr(O\rho)}{8} + \frac{\Tr(O\rho^2)+\Tr(O\rho^3)}{4}
+\frac{(1+3p_2+2p_3) \Tr(O)}{24}, 
\nonumber\\
&\dots \dots
\end{align}

The rest of this section is organized as follows: 
In Sec.~\ref{sec:gO}, we introduce an  auxiliary function and an lemma.
In Sec.~\ref{sec:ExpVarCBNE}, we derive the expectation and variance of our estimator $\hat{\Gamma}_k^{U}(O)$ [see  Eqs.~\eqref{eq:collision_general_prl} and \eqref{eq:GammakUsO}]. 
In Sec.~\ref{sec:CBNEmoment}, we prove the performance of CBNE as stated in \tref{thm:MomOResult} of the main text. In Sec.~\ref{sec:appdesignCBNE}, we prove \pref{prop:appCBNE} in the End Matter, which clarifies the performance of the variant CBNE protocol when the ensemble $\mathcal E$ is an approximate unitary design.

\begin{figure}
\begin{algorithm}[H]
{\small
\hspace{-152pt}\textbf{Input:}  $N_U N_M$ sequentially prepared quantum state $\rho$, observable $O$, and integer $t\geq 1$.   \\

\begin{algorithmic}[1]
\caption{{\small Variant CBNE protocol for estimating $\left\{ p_k, \Tr(O\rho^k)\right\}_{k=1}^t$} \ \ }
\label{alg:variantCBNE}

\For{$s=1,2,\dots,N_U$,}
\State{Randomly choose a unitary $U_s$ from an ensemble $\mathcal E$ and record it. }
\For{$j=1,2,\dots,N_M$,}
\State{Evolve the state $\rho$ using $U_s$ to get $U_s\rho U_s^{\dag}$.} 
\State{Measure the rotated state in the computational basis, obtaining outcome $\hat{b}_j\in\{0,\dots,d-1\}$.}
\EndFor 
\State{Compute the following collision estimator for all $k=1,\dots,t$ using the measurement data $\mathbf{b}_{U_s} = \{\hat{b}_1, \dots, \hat{b}_{N_M}\}$:  
\begin{align}\label{eq:GammakUsO}
\hat{\Gamma}_k^{U_{\!s}} (O) =    
	\frac{ \kappa_{k+1}}{d\binom{N_M}{k}}  
	\sum_{ i_1 <i_2 < \dots < i_k }  \, 
	\textbf{1}\left\{\hat{b}_{i_1}=\dots=\hat{b}_{i_k}\right\} \, \widetilde{\Pr}_{O}(\hat{b}_{i_1} |U_s), 
\end{align}
where $\widetilde{\Pr}_{O}(b|U)=\bra{b} U O U^{\dag} \ket{b}$.
}

\EndFor 
\State{For $k=1,2,\dots,t$, compute $\hat\xi_k^{O}= \frac{1}{N_U} \sum_{s=1}^{N_U} \hat{\Gamma}_k^{U_{\!s}}(O)$ as our estimate for $\xi_k^{O}$. }
\State{Run the post-processing of Algorithm~\ref{alg:CBNEmoments} to obtain estimates $\hat{p}_k$ for $p_k$ ($k=2,\dots,t$).}
\State{Substitute the estimated values of $\xi_1^{O},\dots,\xi_t^{O}$ and $p_2,\dots,p_t$ into  \eref{eq:xiORelation+}, then solve the equations to obtain estimates $\hat{o}_k$ for $\Tr(O\rho^k)$ ($k=1,\dots,t$).}

\end{algorithmic}
\hspace{-368pt} \textbf{Output:} 
$\hat{p}_2, \dots, \hat{p}_t$ and 
$\hat{o}_1, \dots, \hat{o}_t$. 
}
\end{algorithm}
\end{figure}

\subsection{An auxiliary function}\label{sec:gO}

For later use, we define the following function for an Hermitian observable $O$:
\begin{align}\label{eq:definegO}
g(O):=
\begin{cases}
	\Tr(O)^2 &  \text{when } O \geq 0, \\
	\Tr(O^2)   &  \text{when } \Tr(O)=0.   
\end{cases}
\end{align}

\begin{lemma}\label{lem:xikOUB}
Suppose integer $k\geq 1$, $\rho\in\caD(\caH)$, 
and $O$ is an observable on $\caH$ with  $O\geq0$ or $\Tr(O)=0$. 
Then 
\begin{align}\label{eq:xikOUB}
\max_{\pi\in\mathrm{S}_{k+2}}
\Big| \!\Tr\!\big[ \left( \rho^{\otimes k} \otimes O^{\otimes 2} \right) R_{\pi} \big] \Big| 
\leq g(O)
\quad \text{and} \quad
\left| \xi_k^{O}\right| \leq
\max_{\pi\in\mathrm{S}_{k+1}} \left| \Tr\left[ \left(\rho^{\otimes k} \otimes O \right)R_\pi\right]  \right|
\leq \sqrt{g(O)}. 
\end{align}
\end{lemma}

\begin{proof}[Proof of \lref{lem:xikOUB}]
Note that 
\begin{align}
&\max_{\pi\in\mathrm{S}_{k+2}}
\Big| \!\Tr\!\big[ \left( \rho^{\otimes k} \otimes O^{\otimes 2} \right) R_{\pi} \big] \Big| 
\leq
\max_{i,j|0\leq i+j \leq k} \left\lbrace \left| \Tr(O\rho^iO\rho^j) \right|,
\left| \Tr(O\rho^i) \Tr(O\rho^j) \right| \right\rbrace 
\stackrel{(a)}{\leq} \max \left\lbrace \Tr(O)^2, \Tr(O^2) \right\rbrace 
\stackrel{(b)}{=}g(O), 
\\
&\qquad\qquad\quad
\left| \xi_k^{O}\right| \leq
\max_{\pi\in\mathrm{S}_{k+1}} \left| \Tr\left[ \left(\rho^{\otimes k} \otimes O \right)R_\pi\right]  \right| 
\leq \max_{i|0\leq i\leq k} \left| \Tr(O\rho^i) \right|
\stackrel{(c)}{\leq} \max \left\lbrace \Tr(O), \sqrt{\Tr(O^2)} \right\rbrace 
\stackrel{(d)}{=}\sqrt{g(O)}, 
\end{align}
where $(a)$ and $(c)$ follow from \lref{lem:UsefulIneqs} below; 
$(b)$ and $(d)$ hold because $\Tr(O^2)=\|O\|_2^2 \leq\|O\|_1^2= \Tr(O)^2$ when $O\geq0$. 
This completes the proof. 
\end{proof}

\begin{lemma}\label{lem:UsefulIneqs}
Suppose integers $k,l\geq 1$, $\rho\in\caD(\caH)$, 
and $O$ is an Hermitian observable on $\caH$. 
Then 
\begin{align}\label{eq:UsefulIneqs}
\left| \Tr(O\rho^k) \right| \leq \sqrt{\Tr(O^2)},
\qquad 
\left| \Tr(O^2\rho^k) \right| \leq \Tr(O^2),
\qquad 
\left| \Tr(O\rho^{k}O\rho^l) \right| \leq \Tr(O^2). 
\end{align}
\end{lemma}

\begin{proof}[Proof of \lref{lem:UsefulIneqs}]
Equation~\eqref{eq:UsefulIneqs} can be derived as follows: 
\begin{align}
\left| \Tr(O\rho^k) \right| 
&\leq \|O\| \leq \|O\|_2 =\sqrt{\Tr(O^2)},
\\
\left| \Tr(O^2\rho^k) \right| 
&\leq \|O^2\| \leq \|O^2\|_1 =\Tr(O^2), 
\label{eq:TrO2rhokUB}
\\
\left| \Tr(O\rho^{k}O\rho^l) \right| 
&\stackrel{(a)}{\leq} \|O\rho^{k}\|_2 \|O\rho^l\|_2 
= \sqrt{\Tr(O^2\rho^{2k})} \sqrt{\Tr(O^2\rho^{2l})} \stackrel{(b)}{\leq} \Tr(O^2),
\end{align}
where $(a)$ follows from H\"{o}lder's inequality and $(b)$ follows from \eref{eq:TrO2rhokUB}. 
\end{proof}

\subsection{Expectation and variance of our estimator \texorpdfstring{$\hat{\Gamma}_k^{U}(O)$}{}}\label{sec:ExpVarCBNE}
The statistical fluctuation of our estimator $\hat{\Gamma}_k^{U}(O)$ [see Eqs.~\eqref{eq:collision_general_prl} and \eqref{eq:GammakUsO}] comes from: 
(i) the randomness in the selection of the random unitary $U$, and 
(ii) the randomness in quantum measurement outcomes $\mathbf{b}_{U} = \{\hat{b}_1, \dots, \hat{b}_{N_M}\}$. 
Taking these two sources of randomness into account, the expectation and variance of $\hat{\Gamma}_k^{U}(O)$ are 
clarified by the following two lemmas.

\begin{lemma}\label{lam:UnbiasGammakU}
Suppose $U$ is a random unitary  sampled from a $(k+1)$-design ensemble $\mathcal E_{k+1}$, then 
$\E \left[ \hat{\Gamma}_k^U (O)\right] =\xi_k^{O}$ for any observable $O$. 
\end{lemma}

\begin{proof}[Proof of \lref{lam:UnbiasGammakU}]
We have
\begin{align}
\E \left[ \hat{\Gamma}_k^U (O)\right] 
&=\underset{U }{\E} \, \underset{\,\mathbf{b}_{U}}{\E} \left( \hat{\Gamma}_k^U (O) \Big| U \right)
 =\underset{U }{\E} \left[ \frac{\kappa_{k+1}}{d} \sum_{b=0}^{d-1} 
            \Pr\left(\hat{b}_{i_1}=\dots=\hat{b}_{i_k}=b \,|U \right) \<b|UO U^{\dag}|b\>\right] 
\nonumber \\
&=\frac{\kappa_{k+1}}{d}\, \E_{U} \left[\, \sum_{b=0}^{d-1} \<b|U\rho U^{\dag}|b\>^k \<b|UO U^{\dag}|b\> \right] 
= \frac{\kappa_{k+1}}{d} \sum_{b=0}^{d-1} \Tr \left[ \left( \rho^{\otimes k}\otimes O \right)  
\Phi_{\mathcal E_{k+1}} \big( \ketbra{b}{b}^{\otimes {(k+1)}} \big) \right]
\nonumber \\
&= \Tr\!\Big[ \left( \rho^{\otimes k}\otimes O \right)  \Pi_{\mathrm{sym}}^{(k+1)} \Big]
=\frac{1 }{(k+1)! } \sum_{\pi\in \mathrm{S}_{k+1}}  \Tr\left[ \left(\rho^{\otimes k} \otimes O \right)R_\pi\right]
=\xi_k^{O} . 
\end{align}
This confirms \lref{lam:UnbiasGammakU}, implying that $\hat{\Gamma}_k^U (O)$ is an unbiased estimator of $\xi_k^{O}$. 
\end{proof}

\begin{lemma}\label{lem:MomOVariance}
Suppose $k>1$ is a constant integer independent of the system dimension $d$, 
$O$ is an Hermitian observable satisfying $O\geq0$ or $\Tr(O)=0$,
and $U$ is a random unitary  sampled from a $2(k+1)$-design ensemble $\mathcal E_{2k+2}$. Then 
\begin{align}
\Var\left[\hat{\Gamma}_k^U (O)\right] 
= \bigo{g(O) \cdot \max\left\lbrace \frac{1}{d} ,\frac{d^{k-1}}{N_{\!M}^{k}}\right\rbrace},  
\end{align}
where $g(O)$ is defined in \eref{eq:definegO}. 
\end{lemma}

\begin{proof}[Proof of \lref{lem:MomOVariance}]
The variance of $\hat{\Gamma}_k^U (O)$ reads
\begin{align}\label{eq:GammakUO} 
\Var\left[ \hat{\Gamma}_k^U (O)\right] 
&= \E\left[ \left( \hat{\Gamma}_k^U (O)\right) ^2 \right] - \E\left[ \hat{\Gamma}_k^U (O) \right]^2 
= \underset{U }{\E} \, \underset{\,\mathbf{b}_{U}}{\E} \left[ \left( \hat{\Gamma}_k^U (O)\right) ^2 \Big| U \right]
- \left( \xi_k^{O}\right) ^2
\nonumber\\
&= \binom{N_M}{k}^{-2} \left( \frac{\kappa_{k+1}}{d}\right)^2
\sum_{\substack{1\leq i_1<\dots<i_k\leq N_M \\ 1\leq j_1<\dots<j_k\leq N_M}} 
\underset{U }{\E} \, \underset{\,\mathbf{b}_{U}}{\E} \Big(
f_1(\hat{b}_{i_1},\dots,\hat{b}_{i_k},\hat{b}_{j_1},\dots,\hat{b}_{j_k}, U,O)\Big| U \Big) 
-\left( \xi_k^{O}\right) ^2. 
\end{align}
where the second equality follows from \lref{lam:UnbiasGammakU}, and the function $f_1$ is defined as 
\begin{align}
f_1(\hat{b}_{i_1},\dots,\hat{b}_{i_k},\hat{b}_{j_1},\dots,\hat{b}_{j_k}, U,O)
:=
\textbf{1}\left\{\hat{b}_{i_1}=\dots=\hat{b}_{i_k}\right\}  
\textbf{1}\left\{\hat{b}_{j_1}=\dots=\hat{b}_{j_k}\right\} 
\<\hat{b}_{i_1}|UO U^{\dag}|\hat{b}_{i_1}\> \<\hat{b}_{j_1}|UO U^{\dag}|\hat{b}_{j_1}\>. 
\end{align} 

To calculate $\underset{U }{\E} \, \underset{\,\mathbf{b}_{U}}{\E} \Big(
f_1(\hat{b}_{i_1},\dots,\hat{b}_{i_k},\hat{b}_{j_1},\dots,\hat{b}_{j_k}, U,O)\Big| U \Big)$, we consider two cases below. 
Let ${\rm Co}[(i_1,\dots,i_k);(j_1,\dots,j_k)]$ be the number of common elements shared by $(i_1,\dots,i_k)$ and $(j_1,\dots,j_k)$. 
First, if ${\rm Co}[(i_1,\dots,i_k);(j_1,\dots,j_k)]=k-l$ with $l\in\{0,1,\dots,k-1\}$, we have 
\begin{align}\label{eq:EEfcasea}
&\underset{U }{\E} \, \underset{\,\mathbf{b}_{U}}{\E} \Big(
f_1(\hat{b}_{i_1},\dots,\hat{b}_{i_k},\hat{b}_{j_1},\dots,\hat{b}_{j_k}, U,O)\Big| U \Big)
= 
\E_{U} \sum_{b=0}^{d-1} \<b|U\rho U^{\dag}|b\>^{k+l} \<b|UO U^{\dag}|b\>^2
\nonumber\\
&=\sum_{b=0}^{d-1} \Tr \left[ \left( \rho^{\otimes (k+l)} \otimes O^{\otimes 2}\right)  \Phi_{\mathcal E_{2k+2}} \big( \ketbra{b}{b}^{\otimes (k+l+2)} \big) \right] 
=
\frac{d}{\kappa_{k+l+2}} \Tr \left[ \left( \rho^{\otimes (k+l)} \otimes O^{\otimes 2}\right)  \Pi_{\mathrm{sym}}^{(k+l+2)} \right] 
\nonumber\\
&=
\frac{d}{(k+l+2)!\, \kappa_{k+l+2}} \sum_{\pi\in \mathrm{S}_{k+l+2}} \Tr\left[ \left( \rho^{\otimes {(k+l)}}\otimes O^{\otimes 2}\right)  R_\pi \right]
 = \bigo{\frac{g(O)}{d^{k+l+1}}} ,
\end{align}
where the last equality follows from \lref{lem:xikOUB}. 

Second, if ${\rm Co}[(i_1,\dots,i_k);(j_1,\dots,j_k)]=0$, we have 
\begin{align}\label{eq:EEfcaseb}
\underset{U }{\E} \, \underset{\,\mathbf{b}_{U}}{\E} \Big(
&f_1(\hat{b}_{i_1},\dots,\hat{b}_{i_k},\hat{b}_{j_1},\dots,\hat{b}_{j_k}, U,O)\Big| U \Big)
= 
\E_{U} \sum_{b,b'=0}^{d-1} \<b|U\rho U^{\dag}|b\>^{k} \<b|UO U^{\dag}|b\> \<b'|U\rho U^{\dag}|b'\>^{k}\<b'|UO U^{\dag}|b'\>
\nonumber\\
&=
\underbrace{\E_{U} \sum_{b=0}^{d-1} \<b|U\rho U^{\dag}|b\>^{2k}\<b|UO U^{\dag}|b\>^2}_{(*1)}
+\underbrace{\E_{U} \sum_{b\ne b'} \<b|U\rho U^{\dag}|b\>^{k}\<b|UO U^{\dag}|b\> \<b'|U\rho U^{\dag}|b'\>^{k}\<b'|UO U^{\dag}|b'\>}_{(*2)}. 
\end{align}
By using an argument similar to that in \eref{eq:EEfcasea}, we can derive that $(*1)=\mathcal O \big[d^{-2k-1}g(O)\big]$. 
For the term $(*2)$, we have  
\begin{align}\label{eq:EEfcased}
(*2)&=  \sum_{b\ne b'} \Tr\left[ \left( \ketbra{b}{b}^{\otimes (k+1)}\otimes \ketbra{b'}{b'}^{\otimes (k+1)}\right)  
\Phi_{\mathcal E_{2k+2}}\left( \rho^{\otimes k} \otimes O\otimes\rho^{\otimes k} \otimes O \right) \right]
\nonumber\\
&\stackrel{(a)}{=}\sum_{b\ne b'} \sum_{\tau,\pi\in \mathrm{S}_{2k+2}} \Wg_{\pi,\tau}^{(2k+2,d)} 
\Tr\left[ \left( \ketbra{b}{b}^{\otimes (k+1)}\otimes \ketbra{b'}{b'}^{\otimes (k+1)}\right) R_{\tau} \right] 
\Tr\!\Big[ \left( \rho^{\otimes k} \otimes O\otimes\rho^{\otimes k} \otimes O \right) R_{\pi}^{\top} \Big]
\nonumber\\
&\stackrel{(b)}{=}\sum_{b\ne b'} \sum_{\pi\in \mathrm{S}_{2k+2}} \sum_{\,\tau_1,\tau_2\in \mathrm{S}_{k+1}}
\Wg_{\pi,(\tau_1,\tau_2)}^{(2k+2,d)}
\Tr\!\Big[ \left( \rho^{\otimes k} \otimes O\otimes\rho^{\otimes k} \otimes O \right) R_{\pi} \Big]
\nonumber\\
&=d(d-1)\times \underbrace{\sum_{\tau_1,\tau_2\in \mathrm{S}_{k+1}}
\Wg_{(\tau_1,\tau_2),(\tau_1,\tau_2)}^{(2k+2,d)}  
\Tr\!\Big[ \left( \rho^{\otimes k} \otimes O\otimes\rho^{\otimes k} \otimes O \right) R_{(\tau_1,\tau_2)} \Big]}_{(*3)}
\nonumber\\
&\quad\ 
+d(d-1)\times \underbrace{\sum_{\tau_1,\tau_2\in \mathrm{S}_{k+1}}
\sum_{\substack{\pi\in\mathrm{S}_{2k+2}\\ \pi\ne(\tau_1,\tau_2)}} 
\Wg_{\pi,(\tau_1,\tau_2)}^{(2k+2,d)}  
\Tr\!\Big[ \left( \rho^{\otimes k} \otimes O\otimes\rho^{\otimes k} \otimes O \right) R_{\pi} \Big]}_{(*4)}
\nonumber\\ 
&\stackrel{(c)}{\leq}
\frac{\left[ (k+1)!\,\xi_k^{O}\right]^2}{d^{2k}} +\bigo{\frac{g(O)}{d^{2k+1}} }.
\end{align}
Here,
$(a)$ follows from \eref{eq:tirling}; 
$(b)$ holds because $\Tr \left[ \left( \rho^{\otimes k} \otimes O\otimes\rho^{\otimes k} \otimes O \right) R_{\pi}^{\top} \right] =\Tr\left[  \left( \rho^{\otimes k} \otimes O\otimes\rho^{\otimes k} \otimes O \right) R_{\pi}\right]$ and  
$\Tr\!\big[ \big( \ketbra{b}{b}^{\otimes (k+1)}\!\otimes \ketbra{b'}{b'}^{\otimes (k+1)}\big) R_{\tau} \big]=1$ if $\tau=(\tau_1,\tau_2)$ for some $\tau_1,\tau_2\in\mathrm{S}_{k+1}$, and 0 otherwise; 
$(c)$ follows from the following relations: 
\begin{align}\label{eq:EEfcasee}
(*3)&\stackrel{(d)}{=}
\sum_{\tau_1,\tau_2\in \mathrm{S}_{k+1}}
\left[ d^{-2k-2}+ \bigo{d^{-2k-3}}\right]  
\Tr\!\Big[ \left( \rho^{\otimes k} \otimes O \right) R_{\tau_1} \Big]
\Tr\!\Big[ \left( \rho^{\otimes k} \otimes O \right) R_{\tau_2} \Big]
\nonumber\\
&=\frac{1}{d^{2k+2}}  \left\lbrace \sum_{\tau_1\in \mathrm{S}_{k+1}}
\Tr\!\Big[ \left( \rho^{\otimes k} \otimes O \right) R_{\tau_1} \Big] \right\rbrace^2
+\bigo{\frac{1}{d^{2k+3}}} 
\max_{\tau_1\in \mathrm{S}_{k+1}}
\Big| \Tr\!\big[ \left( \rho^{\otimes k} \otimes O \right) R_{\tau_1} \big] \Big|^2 
\nonumber\\ 
&\stackrel{(e)}{=}\frac{1}{d^{2k+2}} \left[ (k+1)!\,\xi_k^{O}\right]^2 +\bigo{\frac{g(O)}{d^{2k+3}} },
\\ 
(*4)&\stackrel{(f)}{\leq}
\bigo{\frac{1}{d^{2k+3}}} \cdot 
\max_{\pi\in\mathrm{S}_{2k+2}}\left| \Tr\!\Big[ \left( \rho^{\otimes 2k} \otimes O^{\otimes 2} \right) R_{\pi} \Big]\right|
\stackrel{(g)}{=}
\bigo{\frac{g(O)}{d^{2k+3}} },
\end{align}
where $(d)$ and $(f)$ follows from \eref{eq:Weingarten}; 
$(e)$ and $(g)$ follows from \lref{lem:xikOUB}. 
In summary, we have 
\begin{align}\label{eq:EEfcasef}
\underset{U }{\E} \, \underset{\,\mathbf{b}_{U}}{\E} \Big(
f_1(\hat{b}_{i_1},\dots,\hat{b}_{i_k},\hat{b}_{j_1},\dots,\hat{b}_{j_k}, U,O)\Big| U \Big)
=(*1)+(*2)\leq 
\frac{\left[ (k+1)!\,\xi_k^{O}\right]^2}{d^{2k}} +\bigo{\frac{g(O)}{d^{2k+1}}}
\end{align}
when ${\rm Co}[(i_1,\dots,i_k);(j_1,\dots,j_k)]=0$. 

The variance in \eref{eq:GammakUO} is then upper bounded by
\begin{align}\label{eq:EEfcaseg}
&\Var\left[ \hat{\Gamma}_k^U (O)\right] +\left( \xi_k^{O}\right)^2
\nonumber\\
&= \binom{N_M}{k}^{-2} \left( \frac{\kappa_{k+1}}{d}\right)^2
\sum_{l=0}^{k} \,
\sum_{\substack{i_1<i_2<\dots<i_k, j_1<j_2<\dots<j_k, \\{\rm Co}[(i_1,\dots,i_k);(j_1,\dots,j_k)]=k-l }} 
\underset{U }{\E} \, \underset{\,\mathbf{b}_{U}}{\E} \Big(
f_1(\hat{b}_{i_1},\dots,\hat{b}_{i_k},\hat{b}_{j_1},\dots,\hat{b}_{j_k}, U,O)\Big| U \Big) 
\nonumber\\
&= 
\binom{N_M}{k}^{-2} \left( \frac{\kappa_{k+1}}{d}\right)^2
\sum_{l=0}^{k} 
\binom{N_M}{k+l} \binom{k+l}{k} \binom{k}{l} \, 
\underset{U }{\E} \, \underset{\,\mathbf{b}_{U}}{\E} \Big(
f_1(\hat{b}_{i_1},\dots,\hat{b}_{i_k},\hat{b}_{j_1},\dots,\hat{b}_{j_k}, U,O)\Big| U \Big)\Big|_{{\rm Co}[(i_1,\dots,i_k);(j_1,\dots,j_k)]=k-l}
\nonumber\\
&\leq 
\binom{N_M}{k}^{-2} \left( \frac{\kappa_{k+1}}{d}\right)^2
\left\lbrace \,  
\sum_{l=0}^{k-1} 
\binom{N_M}{k+l} \binom{k+l}{k} \binom{k}{l} \, 
\bigo{\frac{g(O)}{d^{k+l+1}}}
+   
\binom{N_M}{2k} \binom{2k}{k}
\bigg[\frac{\left( (k+1)!\,\xi_k^{O}\right)^2}{d^{2k}} +\bigo{\frac{g(O)}{d^{2k+1}}}
\bigg]\right\rbrace 
\nonumber\\
&= 
\sum_{l=0}^{k-1}\binom{N_M}{k}^{-2}  
\binom{N_M}{k+l} \mathcal{O}\Big[d^{k-1-l}g(O)\Big]
+ 
\binom{N_M}{k}^{-2}  \binom{N_M}{2k}\binom{2k}{k} \left[ \frac{(k+1)!\kappa_{k+1}}{d^{k+1}} \right]^2 
\left[\left( \xi_k^{O}\right) ^2 + \bigo{\frac{g(O)}{d}} \right], 
\end{align}
where the inequality follows from Eqs.~\eqref{eq:EEfcasea} and \eqref{eq:EEfcasef}. 

Note that 
\begin{align}
	\frac{(k+1)!\kappa_{k+1}}{d^{k+1}}
	= \left( 1+ \frac{k}{d} \right) \left( 1+ \frac{k-1}{d} \right) \cdots \left( 1+ \frac{1}{d} \right) 
	= 1+\bigo{d^{-1}}
\end{align}
and 
\begin{align} 
	\binom{N_M}{k}^{-2}  \binom{N_M}{2k}\binom{2k}{k}
	&= \left( \frac{N_M-k}{N_M}\right)  \left( \frac{N_M-k-1}{N_M-1}\right)  \cdots \left(\frac{N_M-2k+1}{N_M-k+1}\right) 
	\leq 1.
\end{align}
By plugging these into \eref{eq:EEfcaseg}, we obtain 
\begin{align} 
\Var\left[ \hat{\Gamma}_k^U (O)\right] +\left( \xi_k^{O}\right)^2
&\leq
\sum_{l=0}^{k-1} \bigo{\frac{d^{k-l-1}g(O)}{N_M^{k-l}}} 
+ \left[1+\bigo{d^{-1}}\right]^2 
\left[\left( \xi_k^{O}\right) ^2 + \bigo{\frac{g(O)}{d}} \right]
\nonumber\\ 
&= 
\bigo{g(O)\cdot \max\left\lbrace \frac{1}{d} ,\frac{d^{k-1}}{N_{\!M}^{k}}\right\rbrace  }  +\left( \xi_k^{O}\right) ^2,
\end{align}
where we used the relation $\left( \xi_k^{O}\right)^2\leq g(O)$ (from \lref{lem:xikOUB}) in the last line. 
This confirms \lref{lem:MomOVariance}. 
\end{proof}

\subsection{Proof of \tref{thm:MomOResult} in the main text}\label{sec:CBNEmoment}

\begin{theorem}[Formal version of \tref{thm:MomOResult}]
\label{thm:MomOResultRestate}
Suppose $\rho\in\caD(\caH)$, $O$ is an observable on $\caH$ with $\|O\|\le 1$, and $t>1$ is a constant integer independent of $O$ and the system dimension $d$.
Let $\mathfrak{B}= \max \left\lbrace \Tr(O_0^2),1\right\rbrace$. 
For any $0<\epsilon<1$, the CBNE protocol (given in  Algorithm~\ref{alg:CBNEmain} of the End Matter) can return $\epsilon$-additive error estimates for all $p_2,\dots,p_t$  and $\Tr(O\rho),\dots,\Tr(O\rho^t)$ with high probability, provided that: $(i)$ the random unitaries are sampled from a $2(t+1)$-design ensemble, 
\begin{align} 
(ii)\ N_U=\max\left\lbrace 1, c_1 \,\Xi_1 \right\rbrace, 
\quad\text{and}\quad
(iii)\ N_M\geq c_2\,d \min\left\lbrace \Xi_1^{1/t}, 1 \right\rbrace. 
\end{align}
Here, $c_1,c_2>0$ are some constants independent of $\mathfrak{B},d$ and $\epsilon$, and $\Xi_1$ is a shorthand for $\Xi_1(d,\mathfrak{B},\epsilon)=\mathfrak{B}/(d\epsilon^2)$.
\end{theorem}

\begin{proof}[Proof of \tref{thm:MomOResultRestate}]
Let $0<\delta<1$ be a constant. 
According to the union bound and \lref{lem:MoOlessTrans} below, in order to estimate all $p_2,\dots,p_t$ and
$\Tr(O\rho),\dots,\Tr(O\rho^t)$ within $\epsilon$ additive error and failure probability at most $\delta$, it 
suffices to ensure that the following relations hold for all $k=1,2,\dots,t$: 
\begin{align}\label{eq:EachOUBceps}
(a) \ 
\Pr\left\{ \left|\hat{p}_k-p_k\right| > c\epsilon \right\} \leq \frac{\delta}{2t}
\quad \text{and} \quad  
(b) \ 
\Pr\left\{ \left|\hat{\xi}_k^{O_0} - \xi_k^{O_0}\right| > c\epsilon \right\} \leq \frac{\delta}{2t}, 
\end{align}
where $O_0$ is the traceless part of $O$, and $0<c<1$ is some constant factor independent of $O$, $\epsilon$ and $d$.  

According to \tref{thm:MomResult}, the relation $(a)$ in \eref{eq:EachOUBceps}
holds provided that the conditions in \tref{thm:MomOResultRestate} are satisfied. 
We therefore focus on analyzing the relation $(b)$. 
By \lref{lam:UnbiasGammakU}, $\hat\xi_k^{O_0}= \frac{1}{N_U} \sum_{s=1}^{N_U} \hat{\Gamma}_k^{U_{\!s}}(O_0)$  is an unbiased estimator for $\xi_k^{O_0}$. 
Then Chebyshev's inequality implies that 
\begin{align}
\Pr\left\{ \left|\hat{\xi}_k^{O_0} - \xi_k^{O_0}\right| > c\epsilon \right\} 
\leq \frac{1}{c^2\epsilon^2} \Var\left[ \hat{\xi}_k^{O_0}\right] 
= \frac{1}{c^2\epsilon^2 N_U} \Var\!\Big[ \hat{\Gamma}_k^{U}(O_0)\Big] 
= \bigo{\frac{\Tr(O_0^2) }{\epsilon^2 N_U} \cdot \max\left\lbrace \frac{1}{d} ,\frac{d^{k-1}}{N_{\!M}^{k}}\right\rbrace},  
\end{align}
where the last equality follows from \lref{lem:MomOVariance}. 
Hence, to make the relation $(b)$ in \eref{eq:EachOUBceps} hold, it is sufficient to take 
\begin{align} 
\frac{\Tr(O_0^2)}{\epsilon^2 N_U d} \leq 1
\quad\text{and}\quad  
\frac{\Tr(O_0^2) d^{k-1}}{\epsilon^2 N_U N_{\!M}^{k}} \leq 1 , 
\end{align}
where we omit constant coefficients, including $\delta$ and $t$. 
It is straightforward to verify that these two conditions are satisfied by choosing $N_U$ and $N_M$ as in \tref{thm:MomOResultRestate}. 
This completes the proof. 
\end{proof}

\begin{lemma}\label{lem:MoOlessTrans}
Suppose $t\geq1$ is a constant integer and $O$ is an Hermitian observable with $\|O\|\leq1$. 
For any $0<\eta<1$, if the estimates $\{\hat{p}_k\}_{k=2}^t$ and $\{\hat{\xi}_k^{O_0}\}_{k=1}^t$ satisfy
\begin{align}\label{eq:Cond1ptOTrans}
|\hat{p}_k - p_k| \leq \eta 
\quad \text{and} \quad 
\left| \hat{\xi}_k^{O_0} - \xi_k^{O_0}\right|  \leq \eta 
\quad \forall k = 1, \dots, t, 
\end{align} 
then the quantities $\{\hat{o}_k^{(0)}\}_{k=1}^t$ and $\{\hat{o}_k\}_{k=1}^t$ computed in Steps~11--12 of Algorithm~\ref{alg:CBNEmain} satisfy
\begin{align}
\left| \hat{o}_k^{(0)}-\Tr(O_0\rho^k)\right|  = \mathcal{O}(\eta)
\quad \text{and} \quad 
\left|\hat{o}_k-\Tr(O\rho^k)\right|  = \mathcal{O}(\eta)  
 \quad \forall k = 1, \dots, t. 
\end{align} 
\end{lemma}

\begin{proof}[Proof of \lref{lem:MoOlessTrans}]
We proceed by induction on $t$. 

\textbf{Base case ($t=1$):} 
From the relations $\xi_1^{O_0}=\Tr(O_0\rho)/2$ and $\Tr(O\rho)=\Tr(O_0\rho)+\Tr(O)/d$, 
we see that estimating $\xi_k^{O_0}$ within $\eta$-additive error directly yields $\mathcal{O}(\eta)$-additive error estimates for both $\Tr(O_0\rho)$ and $\Tr(O\rho)$. This establishes the lemma for $t=1$.

\textbf{Inductive step ($t \geq 2$):} 
Now we assume that $t\geq2$ is a constant integer (independent of $\eta$, $O$, and $d$) and the desired result holds for all positive integers smaller than $t$. 
For any permutation $\pi \in \mathrm{S}_{t+1}$, we have 
\begin{align}
\Tr\left[ \left(\rho^{\otimes t} \otimes O_0 \right)R_\pi\right]   
=
\begin{cases}
\Tr(O_0\rho^t) &  \text{when } \#(\pi)=1, \\
\vspace{-1.2em} \\
\Tr\left( O_0\rho^{\tilde{\nu}(\pi)} \right) \prod_{i=2}^{t} p_i^{\nu_i(\pi)}   &  \text{when } \#(\pi)\geq 2,  
\end{cases}
\end{align}
where $\#(\pi)$ is the number of disjoint cycles in $\pi$, and the exponents $0\leq \tilde{\nu}(\pi),\nu_i(\pi)<t$ (for $i=2,\dots,t-1$) are some integers depending on $\pi$. 
Combining this with \eref{eq:xikO0def}, we derive:
\begin{align}
(t+1)!\,\xi_t^{O_0}= 
\sum_{\pi\in \mathrm{S}_{t+1}}  \Tr\left[ \left(\rho^{\otimes t} \otimes O_0 \right)R_\pi\right]
=t!\Tr(O_0\rho^t) + \sum_{\pi\in \mathrm{S}_{t+1},\#(\pi)\geq 2} 
\left[ \Tr\left( O_0\rho^{\tilde{\nu}(\pi)} \right) \prod_{i=2}^{t} p_i^{\nu_i(\pi)} \right] ,
\end{align}
where the first term uses the fact that the number of permutations in $\mathrm{S}_t$ with $\#(\pi)=1$ is $t!$. 
Rearranging this relation yields:
\begin{align}\label{eq:ptOExp1}
\Tr(O_0\rho^t) 
= (t+1)\,\xi_t^{O_0} 
- \frac{1}{t!} \sum_{\pi\in \mathrm{S}_{t+1},\#(\pi)\geq 2} 
\left[ \Tr\left( O_0\rho^{\tilde{\nu}(\pi)} \right) \prod_{i=2}^{t} p_i^{\nu_i(\pi)} \right]. 
\end{align}

By the induction hypothesis, if \eref{eq:Cond1ptOTrans} holds, then the estimates $\{\hat{o}_k^{(0)}\}_{k=1}^{t-1}$ obtained from $\{\hat{p}_k\}_{k=2}^{t-1}$ and $\{\hat{\xi}_k^{O_0}\}_{k=1}^{t-1}$ satisfy 
$\big| \hat{o}_k^{(0)}-\Tr(O_0\rho^k)\big|= \mathcal{O}(\eta)$. 
Since the RHS of \eref{eq:ptOExp1} is a summation of a constant number of terms, 
substituting these $\{\hat{o}_k^{(0)}\}_{k=1}^{t-1}$ along with $\{\hat{p}_k\}_{k=2}^{t-1}$ and $\hat{\xi}_t^{O_0}$ into the RHS of \eref{eq:ptOExp1} produces an estimate $\hat{o}_t^{(0)}$ for $\Tr(O_0\rho^t)$ within error $\mathcal{O}(\eta)$. The estimate $\hat{o}_t = \hat{o}_t^{(0)} + \Tr(O) \hat{p}_t / d$ similarly preserves the $\mathcal{O}(\eta)$ additive error for $\Tr(O \rho^t)$. 
This completes the inductive step and the proof.
\end{proof}

\subsection{Nonlinear estimation with an approximate unitary design ensemble}\label{sec:appdesignCBNE}

In this subsection, we prove \pref{prop:appCBNE} in the End Matter.  

\begin{prop}[Restatement of \pref{prop:appCBNE}]
\label{prop:MomOResultApp}
Suppose $\rho\in\caD(\caH)$, $O$ is an observable on $\caH$ satisfying $O\geq 0$ and $\Tr(O)\geq1$, and $t>1$ is a constant integer independent of $O$ and the system dimension $d$.
For any $0<\epsilon<1$, the variant CBNE protocol (given in  Algorithm~\ref{alg:variantCBNE}) can return $\epsilon$-additive error estimates for all $p_2,\dots,p_t$  and $\Tr(O\rho),\dots,\Tr(O\rho^t)$ with high probability, provided that: $(i)$ the random unitaries used in the protocol are sampled from a $[c_3\epsilon^2\!/\!\Tr(O)^2]$-approximate $2(t+1)$-design ensemble, 
\begin{align} \label{eq:MomOResultApp}
{(ii)} \ N_U = \max\left\lbrace 1, \dfrac{c_1 \Tr(O)^2}{d\epsilon^2} \right\rbrace, 
\quad \text{and} \quad
{(iii)} \ N_M \geq  c_2  \min\!\bigg\lbrace \frac{d^{1-1/t}\Tr(O)^{2/t}}{\epsilon^{2/t}}, d \bigg\rbrace,  
\end{align}
where $c_1,c_2,c_3>0$ are some constants  independent of $O$, $d$, and $\epsilon$. 
\end{prop}

\begin{proof}[Proof of \pref{prop:MomOResultApp}]
Let $0<\delta<1$ be a constant. 
According to the union bound and \lref{lem:MoOestTrans} below, in order to estimate all $p_2,\dots,p_t$ and
$\Tr(O\rho),\dots,\Tr(O\rho^t)$ within $\epsilon$ additive error and failure probability at most $\delta$, it 
suffices to ensure that the following relations hold for all $k=1,2,\dots,t$: 
\begin{align}\label{eq:EachOUBcepsAPP}
(a) \ 
\Pr\left\{ \left|\hat{p}_k-p_k\right| > \frac{c\epsilon}{\Tr(O)}  \right\} \leq \frac{\delta}{2t}
\quad \text{and} \quad  
(b) \ 
\Pr\left\{ \left|\hat{\xi}_k^O - \xi_k^O\right| > c\epsilon \right\} \leq \frac{\delta}{2t}, 
\end{align}
where $0<c<1$ is some constant factor independent of $O$, $\epsilon$ and $d$.

According to \pref{prop:appCBNEmoment} in the End Matter, the relation $(a)$ in \eref{eq:EachOUBcepsAPP} 
holds provided that the three conditions in \pref{prop:MomOResultApp} are satisfied. 
We therefore focus on analyzing the relation $(b)$. It can be achieved if
\begin{align}\label{eq:CondiAppMomO}
\left|\E\big[\hat{\xi}_k^O\big]-\xi_k^O 
\right| < \frac{c\epsilon}{2} 
\quad\text{and}\quad
\Pr\left\{ \left|\hat{\xi}_k^O-\E\big[\hat{\xi}_k^O\big]\right| > \frac{c\epsilon}{2}  \right\} \leq \frac{\delta}{2t}.
\end{align} 

Assume that the random unitaries used in the CBNE protocol are sampled from a 
$\mu$-approximate $2(t+1)$-design ensemble, where 
$\mu=c^2\delta\epsilon^2/[30t\Tr(O)^2]$. 
Then \lref{lam:biasGammakUapp} below implies that
\begin{align}
\left|\E\big[\hat{\xi}_k^O\big]-\xi_k^O \right| 
=\Big| \E \left[ \hat{\Gamma}_k^U (O)\right] -\xi_k^{O}\Big| 
\leq \mu\, \xi_k^{O} 
\leq \dfrac{c^2\delta\epsilon^2}{30t\Tr(O)^2} \Tr(O) 
\leq \frac{c\epsilon}{2}, 
\end{align} 
where the second inequality holds because $\xi_k^{O} \leq \Tr(O)$ (see \lref{lem:xikOUB}). 
In addition, Chebyshev's inequality implies that 
\begin{align}\label{eq:ChebyYO}
\Pr\left\{ \left|\hat{\xi}_k^O-\E\big[\hat{\xi}_k^O\big]\right| > \frac{c\epsilon}{2}\right\}
&\leq \frac{4}{c^2\epsilon^2}\Var\!\big[\hat{\xi}_k^O\big] 
\leq \frac{4}{c^2\epsilon^2} \cdot \frac{1}{N_U} 
\left[ \bigo{\Tr(O)^2 \cdot \max\left\lbrace \frac{1}{d} ,\frac{d^{k-1}}{N_{\!M}^{k}}\right\rbrace} 
+ 3\mu \left( \xi_k^{O}\right) ^2 \right] 
\nonumber\\
&\leq \bigo{\frac{\Tr(O)^2 }{\epsilon^2 N_U} \cdot \max\left\lbrace \frac{1}{d} ,\frac{d^{k-1}}{N_{\!M}^{k}} \right\rbrace } 
+ \frac{2\delta}{5 t},
\end{align} 
where the second inequality follows from \lref{lem:GammakVarianceApp} below.   
Therefore, to make sure that \eref{eq:CondiAppMomO} holds, it is sufficient to take 
\begin{align} 
	\frac{\Tr(O)^2}{\epsilon^2 dN_U} \leq 1
	\quad\text{and}\quad  
	\frac{\Tr(O)^2d^{k-1}}{\epsilon^2  N_U N_{\!M}^{k}} \leq 1,
\end{align}
where we omit constant coefficients, including $\delta$ and $t$. 
It is straightforward to verify that these two conditions are satisfied by choosing $N_U$ and $N_M$ as in \eref{eq:MomOResultApp}. 
This completes the proof of \pref{prop:MomOResultApp}. 
\end{proof}

The following two lemmas can be derived by
using arguments similar to those in the proofs of Lemmas~\ref{lam:biasMkUapp} and \ref{lem:MomVarianceApp}, respectively.  

\begin{lemma}\label{lam:biasGammakUapp} 
Suppose $U$ is a random unitary sampled from a $\mu$-approximate $(k+1)$-design 
ensemble $\mathcal E_{k+1}^\mu$, then 
\begin{align}
\left| \underset{U\sim \mathcal E_{k+1}^\mu,\mathbf{b}_{U}}{\E} \left[ \hat{\Gamma}_k^U (O)\right] -\xi_k^{O}\right|  
\leq \mu\, \xi_k^{O} 
\end{align}
for any observable $O\geq0$. 
\end{lemma}

\begin{lemma}\label{lem:GammakVarianceApp}
Suppose $k>1$ is a constant integer independent of the system dimension $d$, 
observable $O\geq 0$, and $U$ is a random unitary  sampled from a 
$\mu$-approximate $2(k+1)$-design ensemble $\mathcal E_{2k+2}^\mu$. Then 
\begin{align}
\underset{U\sim \mathcal E_{2k+2}^\mu,\mathbf{b}_{U}}{\Var}
\left[\hat{\Gamma}_k^U (O)\right] 
= \bigo{\Tr(O)^2 \cdot \max\left\lbrace \frac{1}{d} ,\frac{d^{k-1}}{N_{\!M}^{k}}\right\rbrace} 
+ 3\mu \left( \xi_k^{O}\right) ^2. 
\end{align}
\end{lemma}

The following lemma can be derived by a similar induction argument used in the proof of \lref{lem:MoOlessTrans}.

\begin{lemma}\label{lem:MoOestTrans}
Suppose $k>1$ is a constant integer and observable $O\geq 0$. 
For any $0<\eta<1$, if the estimates $\{\hat{p}_k\}_{k=2}^t$ and $\{\hat{\xi}_k^O\}_{k=1}^t$ satisfy 
\begin{align}
|\hat{p}_k - p_k| \leq \frac{\eta}{\Tr(O)} 
\quad \text{and} \quad 
\left| \hat{\xi}_k^O - \xi_k^O\right|  \leq \eta
\quad \forall k = 1, \dots, t, 
\end{align} 
then the quantities $\{\hat{o}_k\}_{k=1}^t$ computed in Step~11 of Algorithm~\ref{alg:variantCBNE} satisfy
\begin{align}
	\left|\hat{o}_k-\Tr(O\rho^k)\right|  = \mathcal{O}(\eta)  
	\quad \forall k = 1, \dots, t.  
\end{align} 
\end{lemma}

\section{Principal component estimation}\label{sec:PCE}

In this section, we prove \tref{thm:PrinEigen} in the main text, 
which characterizes the efficiency of the CBNE protocol for solving the PCE task. 
First, we formally define the PCE task as follows:
\begin{defn}[PCE task]\label{defn:PCEtask}
Let $\rho$ be an unknown quantum state on a $d$-dimensional Hilbert space $\caH$ with largest eigenvalue $\lambda$
and corresponding eigenstate $\ket{\psi}$.
Assume that $\rho$ has a constant spectral gap, that is, $\lambda-\lambda_2\ge \mu$ for some constant $\mu>0$
independent of $d$, where $\lambda_2$ denotes the second largest eigenvalue of $\rho$.
Given an observable $O$ on $\caH$ with $\|O\|\leq 1$,
the PCE task is to estimate $\bra{\psi}O\ket{\psi}$ up to additive error $0<\epsilon<1$
with high success probability.
\end{defn}

Throughout this section, we assume that the unknown state $\rho$ has the following spectral decomposition:
\begin{align}\label{eq:rhospec1}
	\rho=\lambda \ketbra{\psi}{\psi} + \sum_{i=2}^d \lambda_i \ketbra{\psi_i}{\psi_i},
\end{align}
where the eigenvalues are arranged in decreasing order $\lambda> \lambda_2\geq \dots \geq \lambda_d$, 
and $\lambda-\lambda_2\geq \mu$ with constant $0<\mu\leq1$ independent of the dimension $d$.

\subsection{Proof of \tref{thm:PrinEigen} in the main text}

To prove \tref{thm:PrinEigen}, it suffices to prove the following theorem, which is a stronger and formal version of \tref{thm:PrinEigen}.  

\begin{theorem}\label{thm:PrinEigenFor}
	Suppose the target error $0<\epsilon<1$ is a constant independent of the system dimension $d$ and the target observable $O$. 
	Then the CBNE protocol (given in  Algorithm~\ref{alg:CBNEmain} of the End Matter) can solve the PCE task in Definition~\ref{defn:PCEtask},
	provided that: $(i)$ the random unitaries are sampled from a $2(t+1)$-design ensemble, $(ii)$ integer $t\geq c_1$, 
	\begin{align}\label{eq:PCEtaskNuNm}
		(iii)\ N_U = \max\left\lbrace 1, \dfrac{c_2 \mathfrak{B}}{d} \right\rbrace, 
		\quad \text{and} \quad
		(iv)\ N_M \geq  c_3 \,d^{1-1/t}\mathfrak{B}^{1/t} .  
	\end{align}
    Here $\mathfrak{B}= \max \left\lbrace \Tr(O_0^2),1\right\rbrace$, $O_0$ is the traceless part of $O$, and $c_1,c_2,c_3$ are some constants independent of $\mathfrak{B}$ and $d$.
\end{theorem}

Since $\mathfrak{B}\leq d$,
\tref{thm:PrinEigen} in the main text follows immediately from \tref{thm:PrinEigenFor}.

\begin{proof}[Proof of \tref{thm:PrinEigenFor}] 
	Let $\eta = \epsilon/2$ and integer $t \geq \lceil \log_{1-\mu/3} \frac{\mu \eta}{6} \rceil$, 
	which, by assumption, are both constants independent of $\mathfrak{B}$ and $d$. 
	According to \tref{thm:MomOResult} in the main text, the CBNE protocol can (with high probability) return estimates $\hat{o}_t, \hat{p}_t \in \R$ satisfying 
	\begin{align}\label{eq:PrinEigenProof1}
		\left| \hat{o}_t - \tr{O\rho^t}\right| \leq \frac{\mu^t \eta}{3}
		\quad \text{and} \quad
		\left| \hat{p}_t - \tr{\rho^t}\right| \leq \frac{\mu^t \eta}{3},
	\end{align}
	provided $N_U$ and $N_M$ are chosen as in \eref{eq:PCEtaskNuNm}. 
	Furthermore, if \eqref{eq:PrinEigenProof1} holds, then
	\begin{align}
		\left| \frac{\hat{o}_t}{\hat{p}_t} - \bra{\psi}O\ket{\psi} \right|
		\stackrel{(a)}{\leq}  
		\left| \frac{\hat{o}_t}{\hat{p}_t} - \frac{\tr{O\rho^t}}{\tr{\rho^t}} \right|
		+ \left| \frac{\tr{O\rho^t}}{\tr{\rho^t}} - \bra{\psi}O\ket{\psi} \right| 
		\stackrel{(b)}{\leq} \eta + \eta = \epsilon, 
	\end{align}
	that is, $\hat{o}_t/\hat{p}_t$ is an $\epsilon$-additive error estimate for the desired property $\bra{\psi}O\ket{\psi}$. 
	Here, $(a)$ uses the triangle inequality $|x+y|\leq |x|+|y|$, and $(b)$ follows from Lemmas~\ref{lem:PrincErrorA} and \ref{lem:PrincErrorB} below. 
	
	Therefore, the CBNE protocol solves the PCE task under the three conditions of \tref{thm:PrinEigenFor}.
\end{proof}

\subsection{Auxiliary lemmas}

\begin{lemma}\label{lem:PrincErrorALem1}
The state $\rho$ in \eref{eq:rhospec1} admits the following decomposition: 
\begin{align}\label{eq:PrincErrorArho}
\rho = \lambda\ketbra{\psi}{\psi} + \sum_{j=1}^m w_j \sigma_j, 
\end{align}
where  
\begin{itemize}
\item[$(i)$] $\{\sigma_j\}_{j=1}^m$ are quantum states on $\caH$ that have mutually orthogonal supports, 
i.e., $\Tr(\sigma_i\sigma_j) = 0\  \forall i \neq j$. 

\item[$(ii)$] The coefficients $\{w_j\}_{j=1}^m$ satisfy 
$w_j \in [\mu/3, \lambda-\mu/3]$ for $j=1,2,\dots,m-1$ and $w_m \in (0, \lambda-\mu/3]$.

\item[$(iii)$] $m$ is an integer satisfying $0\leq m\leq 3/\mu$. 
\end{itemize}
\end{lemma}

\begin{figure}[b]
\begin{algorithm}[H]
{\small
\hspace{-148pt}\textbf{Input:}  Eigenvalues $1>\lambda> \lambda_2\geq \lambda_3\geq \dots \geq \lambda_d\geq0$,  spectral gap $0<\mu\leq \lambda$, and $\ell_0=1$.    \\

\begin{algorithmic}[1]
\caption{{\small Algorithm for choosing integers $\ell_1, \ell_2,\dots,\ell_{m-1} \qquad $}}
\label{alg:choose_ell}

\For{$j=1,2,3,\dots$,}

\State{$\ell_j \leftarrow \ell_{j-1}+1$} \Comment{Start new group at $\ell_{j-1}+1$}

\While{$\sum_{i=\ell_{j-1}+1}^{\ell_j} \lambda_i <\mu/3$ and $\ell_j<d$}
\State{$\ell_j \leftarrow \ell_j+1$} \Comment{Expand group until weight $\geq \mu/3$}
\EndWhile

\If{$\ell_j=d$} \Comment{All eigenvalues grouped}
\State{$m \leftarrow j$; \textbf{break}} 
\EndIf

\If{$\sum_{i=\ell_{j}+1}^{d} \lambda_i <\mu/3$} \Comment{Remaining weight too small}
\State{$m \leftarrow j+1$; \textbf{break}} \Comment{Assign the rest to the last group}
\EndIf

\EndFor

\end{algorithmic}
\hspace{-403pt} \textbf{Output:} $\ell_1, \ell_2,\dots,\ell_{m-1}$. 
}
\end{algorithm}
\end{figure}

\begin{proof}[Proof of \lref{lem:PrincErrorALem1}] 
	In the trivial case of $\lambda=1$, we have $\rho = \lambda\ketbra{\psi}{\psi}$ and $m=0$. 
	So the three conditions in \lref{lem:PrincErrorALem1} are automatically satisfied. 
	
	In the following, we consider the nontrivial case with $\lambda<1$.  
	We shall construct the decomposition of $\rho$ explicitly.
	Let $m\geq 1$ and $\ell_0, \ell_1, \dots,\ell_m$ be integers such that
	$1=\ell_0< \ell_1 <\dots<\ell_m=d$.
	For $j=1,2,\dots,m$, define
	\begin{align}\label{eq:wjsigmaj}
		w_j =\sum_{i=\ell_{j-1}+1}^{\ell_j} \lambda_i,  
		\qquad 
		\sigma_j =\frac{1}{w_j} \sum_{i=\ell_{j-1}+1}^{\ell_j} \lambda_i \ketbra{\psi_i}{\psi_i}. 
	\end{align}
	Since \(\{\ket{\psi_i}\}_{i=1}^d\) is an orthonormal basis, the states \(\{\sigma_j\}_j\) have mutually orthogonal supports, satisfying condition \((i)\) in \lref{lem:PrincErrorALem1}. 
	Using $\{w_j\}_j$ and $\{\sigma_j\}_j$ in \eref{eq:wjsigmaj}, the state $\rho$ in \eref{eq:rhospec1} can be rewritten as:
	\begin{align}
		\rho=
		\lambda \ketbra{\psi}{\psi} 
		+ \underbrace{\sum_{i=2}^{\ell_1} \lambda_i \ketbra{\psi_i}{\psi_i}}_{w_1 \sigma_1}
		+ \underbrace{\sum_{i=\ell_1+1}^{\ell_2} \lambda_i \ketbra{\psi_i}{\psi_i}}_{w_2 \sigma_2}
		+ \cdots
		+ \underbrace{\sum_{i=\ell_{m-1}+1}^{d} \lambda_i \ketbra{\psi_i}{\psi_i}}_{w_m \sigma_m}
		= \lambda\ketbra{\psi}{\psi} + \sum_{j=1}^m w_j \sigma_j  . 
	\end{align}
	This matches the desired form in \eref{eq:PrincErrorArho}.

    To complete the proof, it remains to show that there exist integers $\ell_1, \ell_2,\dots,\ell_{m-1}$ such that condition \((ii)\) is satisfied. 
    To this end, we give a constructive grouping of the eigenvalues, described in Algorithm~\ref{alg:choose_ell}.  
	This algorithm ensures each \(w_j\) (for \(j < m\)) is at least \(\mu/3\) by iteratively adding eigenvalues until the sum exceeds \(\mu/3\). 
	In addition, since  \(\lambda_i \leq \lambda_2 \leq \lambda - \mu\) for \(i \geq 2\), the total weight of any group \(w_j\) cannot exceed \(\lambda - \mu/3\). 
	Algorithm~\ref{alg:choose_ell} terminates when $\ell_j=d$ or the remaining eigenvalues satisfy
	$\sum_{i=\ell_{j}+1}^{d} \lambda_i <\mu/3$. 
	These remaining eigenvalues are grouped into \(w_m\), which must satisfy \(w_m \in (0, \lambda - \mu/3]\). This ensures that condition \((ii)\) holds.
	
	Finally, from the trace relation:
	\begin{align}
		1= \tr\rho=\lambda+ \sum_{j=1}^m w_j \geq \lambda+ (m-1) \min_{j\in\{1,2,\dots,m-1\}}w_j
		\geq \lambda + (m-1) \cdot \frac{\mu}{3}, 
	\end{align}
	we derive that 
	\begin{align}
		m \leq \frac{1 - \lambda}{\mu/3} + 1 
		=\frac{3 - 3\lambda+\mu}{\mu} 
		\leq \frac{3}{\mu}, 
	\end{align}
	where the last inequality follows because $\mu\leq \lambda<3\lambda$. This proves condition \((iii)\).
	
	By construction, all conditions \((i)\)–\((iii)\) are satisfied, completing the proof of \lref{lem:PrincErrorALem1}.
\end{proof}

\begin{lemma}\label{lem:PrincErrorA}
Suppose $0<\eta<1$, integer $t\geq \lceil \log_{1-\mu/3}\frac{\mu\eta}{6}\rceil$, and $O$ is an observable with $\|O\|\leq 1$. Then 
\begin{align}\label{eq:PrincErrorA}
\left| \frac{\tr{O\rho^t}}{\tr{\rho^t}} - \bra{\psi}O\ket{\psi} \right|\leq \eta. 
\end{align}
\end{lemma}

\begin{proof}[Proof of \lref{lem:PrincErrorA}]
\lref{lem:PrincErrorALem1} implies that 
\begin{align}\label{eq:PrincErrorArho^t}
\rho^t = \lambda^t \ketbra{\psi}{\psi} + \sum_{j=1}^m w_j^t \sigma_j^t, 
\quad
\Tr(O\rho^t) = \lambda^t \bra{\psi}O\ket{\psi} + \sum_{j=1}^m w_j^t \Tr(O\sigma_j^t),
\quad 
\lambda^t \leq \Tr(\rho^t)\leq \lambda^t + \sum_{j=1}^m w_j^t . 
\end{align}
It follows that 
\begin{align}\label{eq:ErrorAterm0}
\left|\frac{\tr{O\rho^t}}{\tr{\rho^t}}-\bra{\psi}O\ket{\psi}\right|
&\leq 
\max\Bigg\{ \,
\underbrace{\left|\frac{\tr{O\rho^t}}{\lambda^t}-\bra{\psi}O\ket{\psi}\right|}_{(*)}, \;
\underbrace{\left|\frac{\tr{O\rho^t}}{\lambda^t + \sum_{j} w_j^t}-\bra{\psi}O\ket{\psi}\right|}_{(**)}
\, \Bigg\}. 
\end{align}

Next, we bound the terms $(*)$ and $(**)$, respectively. First, by \eref{eq:PrincErrorArho^t} we have
\begin{align}\label{eq:ErrorAterm1}
(*)
&=
\left|\frac{\sum_{j=1}^m w_j^t \,\tr{O \sigma_j^t}}{\lambda^t}\right|
\leq
\sum_{j=1}^m \left( \frac{w_j}{\lambda}\right)^t \left|\tr{O \sigma_j^t}\right|
\leq
\sum_{j=1}^m \left( \frac{w_j}{\lambda}\right)^t \|O\|
\nonumber\\
&\leq
m \left(\frac{\max_j w_j}{\lambda}\right)^t 
\stackrel{(a)}{\leq}  \frac{3}{\mu} \left( \frac{\lambda-\mu/3}{\lambda}\right)^t
\leq  \frac{3}{\mu} \left( 1- \frac{\mu}{3} \right)^t
\stackrel{(b)}{\leq} \frac{3}{\mu} \cdot \frac{\mu\eta}{6}
< \eta,  
\end{align}
where $(a)$ follows from \lref{lem:PrincErrorALem1} and $(b)$ follows from the assumption $t\geq \lceil \log_{1-\mu/3}\frac{\mu\eta}{6}\rceil$. 
Second, by \eref{eq:PrincErrorArho^t} we have
\begin{align}\label{eq:ErrorAterm2}
(**)
&=
\left| \frac{\lambda^t \bra{\psi}O\ket{\psi} + \sum_{j=1}^m w_j^t \, \tr{O \sigma_j^t}}{\lambda^t + \sum_{j=1}^m w_j^t} - \bra{\psi}O\ket{\psi} \right|
=
\left| \frac{\sum_{j=1}^m w_j^t \, \tr{O \sigma_j^t}
-\sum_{j=1}^m w_j^t \,\bra{\psi}O\ket{\psi}}{\lambda^t + \sum_{j=1}^m w_j^t} 
\right|
\nonumber\\
&\leq
\frac{\sum_{j=1}^m w_j^t \,\left| \tr{O \sigma_j^t}-\bra{\psi}O\ket{\psi}\right| }{\lambda^t + \sum_{j=1}^m w_j^t}
\leq
\frac{\sum_{j=1}^m w_j^t \,\left| \tr{O \sigma_j^t}-\bra{\psi}O\ket{\psi}\right| }{\lambda^t}
\nonumber\\
&\stackrel{(a)}{\leq}
2\sum_{j=1}^m \left( \frac{w_j}{\lambda}\right) ^t 
\leq 2m \left(\frac{\max_j w_j}{\lambda}\right) ^t 
\stackrel{(b)}{\leq} \frac{6}{\mu} \left( \frac{\lambda-\mu/3}{\lambda}\right)^t
\leq \frac{6}{\mu} \left( 1- \frac{\mu}{3} \right)^t
\stackrel{(c)}{\leq} \frac{6}{\mu} \cdot \frac{\mu\eta}{6}
= \eta, 
\end{align}
where $(a)$ holds because $\left|\Tr(O \sigma_j^t)-\bra{\psi}O\ket{\psi}\right|
\leq|\!\Tr(O \sigma_j^t)|+ |\bra{\psi}O\ket{\psi}|\leq2$, 
$(b)$ follows from \lref{lem:PrincErrorALem1},
and $(c)$ follows from the assumption $t\geq \lceil \log_{1-\mu/3}\frac{\mu\eta}{6}\rceil$. 
Equations~\eqref{eq:ErrorAterm0}, \eqref{eq:ErrorAterm1}, and \eqref{eq:ErrorAterm2} together confirm \lref{lem:PrincErrorA}. 
\end{proof}

\begin{lemma}\label{lem:PrincErrorB}
Suppose $0<\eta<1$, integer $t\geq 1$, $O$ is an observable with $\|O\|\leq 1$, $\hat{o}_t,\hat{p}_t\in\R$, 
$|\hat{o}_t-\tr{O\rho^t}|\leq \mu^t\eta/3$, and $|\hat{p}_t-\tr{\rho^t}|\leq \mu^t\eta/3$. Then 
\begin{align}
\left| \frac{\hat{o}_t}{\hat{p}_t} - \frac{\tr{O\rho^t}}{\tr{\rho^t}} \right| \leq \eta. 
\end{align}
\end{lemma}

\begin{proof}[Proof of \lref{lem:PrincErrorB}]
Let $\tilde\eta:=\mu^t\eta/3$. Then we have 
\begin{align}
\left| \frac{\tr{O\rho^t}\pm \tilde\eta }{\tr{\rho^t}-\tilde\eta} - \frac{\tr{O\rho^t}}{\tr{\rho^t}} \right|
&
\stackrel{(a)}{\leq} 
\left| \frac{\tr{O\rho^t}\pm\tilde\eta }{\tr{\rho^t}-\tilde\eta}-\frac{\tr{O\rho^t} }{\tr{\rho^t}-\tilde\eta} \right|
+ \left|\frac{\tr{O\rho^t} }{\tr{\rho^t}-\tilde\eta}-\frac{\tr{O\rho^t}}{\tr{\rho^t}} \right|
\nonumber\\
&=
\left| \frac{\tilde\eta }{\tr{\rho^t}-\tilde\eta}\right|
+ \left|\frac{ \tr{O\rho^t}}{\tr{\rho^t}} \cdot \frac{\tilde\eta }{\tr{\rho^t}-\tilde\eta}\right|
\nonumber\\
&= \left( \left|\frac{ \tr{O\rho^t}}{\tr{\rho^t}}\right|+1 \right)  \cdot 
\left| \frac{\tilde\eta }{\tr{\rho^t}-\tilde\eta}\right| 
\nonumber\\
&\stackrel{(b)}{\leq}  \frac{2\tilde\eta }{\left|\tr{\rho^t}-\tilde\eta\right|}
\stackrel{(c)}{\leq}  \frac{2\mu^t\eta/3 }{\mu^t-\mu^t\eta/3}
=    \frac{2\eta}{3-\eta}
\leq \eta,  
\end{align}
where $(a)$ follows from the triangle inequality $|x+y|\leq |x|+|y|$;
$(b)$ holds because $\left|\tr{O\rho^t}/\tr{\rho^t}\right|\leq 1$; 
and $(c)$ holds because $\tr{\rho^t}\geq \lambda^t\geq \mu^t$. 
Using a similar argument, we can derive that
\begin{align}
\left| \frac{\tr{O\rho^t}\pm \tilde\eta }{\tr{\rho^t}+\tilde\eta} - \frac{\tr{O\rho^t}}{\tr{\rho^t}} \right|
&
\leq \left| \frac{\tr{O\rho^t}\pm\tilde\eta }{\tr{\rho^t}+\tilde\eta}-\frac{\tr{O\rho^t} }{\tr{\rho^t}+\tilde\eta} \right|
+ \left|\frac{\tr{O\rho^t} }{\tr{\rho^t}+\tilde\eta}-\frac{\tr{O\rho^t}}{\tr{\rho^t}} \right|
\nonumber\\
&=
\left| \frac{\tilde\eta }{\tr{\rho^t}+\tilde\eta}\right|
+ \left|\frac{ \tr{O\rho^t}}{\tr{\rho^t}}\cdot \frac{\tilde\eta }{\tr{\rho^t}+\tilde\eta}\right|
\nonumber\\
&= \left( \left|\frac{ \tr{O\rho^t}}{\tr{\rho^t}}\right|+1 \right)  \cdot 
\frac{\tilde\eta }{\tr{\rho^t}+\tilde\eta}
\leq \frac{2\tilde\eta }{\tr{\rho^t}} 
\leq \frac{2\mu^t\eta/3 }{\mu^t}
\leq \eta. 
\end{align}
Therefore, 
\begin{align}
\left| \frac{\hat{o}_t}{\hat{p}_t} - \frac{\tr{O\rho^t}}{\tr{\rho^t}} \right|
\leq \max\left\lbrace 
\left| \frac{\tr{O\rho^t}\pm \tilde\eta }{\tr{\rho^t}-\tilde\eta} - \frac{\tr{O\rho^t}}{\tr{\rho^t}} \right|,
\left| \frac{\tr{O\rho^t}\pm \tilde\eta }{\tr{\rho^t}+\tilde\eta} - \frac{\tr{O\rho^t}}{\tr{\rho^t}} \right|  \right\rbrace 
\leq \eta, 
\end{align}
which completes the proof. 
\end{proof}

\section{Estimation of PT moments with  the PTME protocol}\label{sec:PTME}

The goal of this section is to prove \tref{thm:PTResult} in the main text, which characterizes the performance of our PTME protocol. 
In Secs.~\ref{sec:ExpPTmoment} and~\ref{sec:VarPTmoment}, we 
derive the expectation value and variance of our estimator 
\begin{align}\label{eq:hatLambdaPTSM}
\hat{\Lambda}_k^{U}
	:= \binom{N_M}{k}^{-1} \frac{d_{\!A}^{\,k}}{k!\, d_{\!A_1} }
	\sum_{1\leq i_1 <\dots < i_k \leq N_M} \left( \hat{r}_{i_1}\cdots \hat{r}_{i_k} \right)  \textbf{1}\big\{\hat{b}_{i_1}=\dots=\hat{b}_{i_k}\big\}, 
\end{align}
respectively, where
$\mathbf{b}_{U} = \big\{\hat{b}_1, \dots, \hat{b}_{N_M}\big\}$ and $\mathbf{r}_{U} = \{\hat{r}_1, \dots, \hat{r}_{N_M}\}$ denote the measurement outcomes of the PTME protocol corresponding to the random unitary $U$ (see Algorithm~\ref{alg:PTME} in the End Matter). 
Combining these results, we then present the proof of \tref{thm:PTResult} in Sec.~\ref{sec:ProofthmPTResult}. 
Throughout this section we use the relations $d=d_A d_B$, $d_A=d_{A_1} d_{A_2}$, and $d_B=d_{A_2}$.

\subsection{Expectation of the estimator \texorpdfstring{$\hat{\Lambda}_k^{U}$}{}}\label{sec:ExpPTmoment}

The statistical fluctuation of our estimator $\hat{\Lambda}_k^{U}$ comes from: 
(i) the randomness in the selection of the random unitary $U$ on $\caH_A$, and 
(ii) the randomness in measurement outcomes $\mathbf{b}_{U} = \big\{\hat{b}_1, \dots, \hat{b}_{N_M}\big\}$ 
and $\mathbf{r}_{U} = \{\hat{r}_1, \dots, \hat{r}_{N_M}\}$. 
Taking both sources of randomness into account, the expectation value of \(\hat{\Lambda}_k^U\) is characterized by the following lemma:

\begin{lemma}\label{lam:biasLambdaU}
Suppose $k>1$ is a constant integer independent of the system dimension, and $U$ is a random unitary  sampled from a $k$-design ensemble on $\caH_A$. Then 
$\E \left[ \hat{\Lambda}_k^{U} \right]= \zeta_k^{\PT} + \bigo{d_{\!A}^{\,-1}}$, where $\zeta_k^{\PT}$ is defined in \eref{eq:zetakPT} of the End Matter. 
\end{lemma}

\begin{proof}[Proof of \lref{lam:biasLambdaU}]
Let $\rho_{U}=(U\otimes I_B)\rho(U^{\dag}\otimes I_B)$. By definition of \(\hat{\Lambda}_k^U\), we have
\begin{align}\label{eq:ELambdaU1}
\E \left[ \hat{\Lambda}_k^{U} \right]
&=\underset{U }{\E} \, \underset{\,\mathbf{b}_{U},\mathbf{r}_{U}}{\E} 
\left( \hat{\Lambda}_k^{U} \Big| U \right)
= \frac{d_{\!A}^{\,k}}{k! \,d_{\!A_1}} \, \underset{U }{\E} \, \underset{\,\mathbf{b}_{U},\mathbf{r}_{U}}{\E} 
\left\{ \left. \left( \hat{r}_{i_1}\cdots \hat{r}_{i_k} \right)  \textbf{1}\big\{\hat{b}_{i_1}=\dots=\hat{b}_{i_k}\big\} \right| U \right\} 
\nonumber \\
&=  \frac{d_{\!A}^{\,k}}{k! \,d_{\!A_1}} \, \E_{U} \left\lbrace \, \sum_{\ell=0}^k \sum_{b=0}^{d_{A_1}-1} (-1)^{\ell} 
\Pr\left(\#\left[ \text{$-1$ among $\left\lbrace \hat{r}_{i_1},\dots, \hat{r}_{i_k}\right\rbrace $}\right] =\ell,\, \text{and } \hat{b}_{i_1}=\dots=\hat{b}_{i_k}=b \,\Big| U \right)  \right\rbrace 
\nonumber \\
&=  \frac{d_{\!A}^{\,k}}{k! \,d_{\!A_1}} \, \E_{U} \left\lbrace \,\sum_{\ell=0}^k \sum_{b=0}^{d_{A_1}-1} \binom{k}{\ell} (-1)^{\ell} 
\Pr\left(\hat{r}=-1, \hat{b}=b \,\Big| U\right)^\ell \Pr\left(\hat{r}=+1, \hat{b}=b\,\Big| U\right)^{k-\ell} \right\rbrace 
\nonumber \\
&=  \frac{d_{\!A}^{\,k}}{k! \,d_{\!A_1}} \, \E_{U} \left\lbrace \,\sum_{b=0}^{d_{A_1}-1} 
\left[ \Pr\left(\hat{r}=+1, \hat{b}=b\,\Big| U\right)-\Pr\left(\hat{r}=-1, \hat{b}=b \,\Big| U\right) \right]^k \right\rbrace 
\nonumber \\
&=  \frac{d_{\!A}^{\,k}}{k! \,d_{\!A_1}} \, \sum_{b=0}^{d_{A_1}-1} 
\E_{U} \left\lbrace  \Tr\left[\rho_{U} \left( \ketbra{b}{b}_{A_1} \otimes \mathbb{S}_{A_2,B} \right)  \right]^k \right\rbrace ,
\end{align}
where in the last equality we use 
\begin{align}
\Pr\left(\hat{r}=\pm 1, \hat{b}=b\,\Big| U\right) = 
\Tr\left[\rho_{U} \left( \ketbra{b}{b}_{A_1} \!\otimes \frac{I_{A_2,B}\pm\mathbb{S}_{A_2,B}}{2} \right)  \right].  
\end{align}

Note that 
\begin{align}\label{eq:ELambdaU2}
\E_{U} \left\lbrace  \Tr\left[\rho_{U} \left( \ketbra{b}{b}_{A_1} \!\otimes \mathbb{S}_{A_2,B} \right)  \right]^k \right\rbrace
&= 
\E_{U} \left\lbrace \Tr\left[\rho(U^{\dag}\otimes I_B) \left(\ketbra{b}{b}_{A_1} \!\otimes \mathbb{S}_{A_2,B} \right) (U\otimes I_B) \right]^k \right\rbrace
\nonumber\\
&= 
\Tr\left\{\rho^{\otimes k} \, \E_{U} \left[ (U^{\dag}\otimes I_B)^{\otimes k} \left(\ketbra{b}{b}_{A_1} \!\otimes \mathbb{S}_{A_2,B} \right)^{\otimes k} (U\otimes I_B)^{\otimes k} \right]\right\}
\nonumber\\
&= 
\Tr\left\{\rho^{\otimes k} \, \left(\Phi_{A}\otimes \mathcal{I}_{B} \right) 
\left[ \left(\ketbra{b}{b}_{A_1} \!\otimes \mathbb{S}_{A_2,B} \right)^{\otimes k}\right]  \right\}
\nonumber\\
&= \Tr \left( \rho^{\otimes k}\tilde H_{b,k} \right) , 
\end{align}
where $\Phi_{A}(\cdot)$ is the $k$-ford twirling channel (see \eref{eq:designDef}) acting on subsystem $A$, and $\tilde H_{b,k}$ in the last line is defined as 
\begin{align}
\tilde H_{b,k}
&:=\left(\Phi_{A}\otimes \mathcal{I}_{B} \right) \left[ \left(\ketbra{b}{b}_{A_1} \!\otimes \mathbb{S}_{A_2,B} \right)^{\otimes k}\right], 
\end{align}
which satisfies
\begin{align}\label{eq:ELambdaU3}
\tilde H_{b,k}
&\,\stackrel{(a)}{=} \sum_{\tau,\pi\in \mathrm{S}_k} \Wg_{\pi,\tau}^{(k,d_A)} \, 
(R_{\tau})_A \otimes 
\Tr_A \left\lbrace \left(\ketbra{b}{b}_{A_1} \!\otimes \mathbb{S}_{A_2,B} \right)^{\otimes k} \left[  (R_{\pi}^{\top})_A \otimes I_B^{\otimes k} \right] \right\rbrace 
\nonumber \\
&\,\stackrel{(b)}{=} \sum_{\pi, \tau \in \mathrm{S}_k} \Wg_{\pi,\tau}^{(k,d_A)} \, 
(R_\tau)_{A} \otimes (R_\pi^\top)_B
\nonumber \\
&\,\stackrel{(c)}{=} d_{\!A}^{\,-k} \left[\, \sum_{\pi\in \mathrm{S}_k} (R_\pi)_A \otimes (R_\pi^\top)_B 
+ \bigo{d_{\!A}^{\,-1}} \sum_{\pi, \tau \in \mathrm{S}_k} 
(R_\tau)_{A} \otimes (R_\pi^\top)_B \right] .
\end{align}
Here, $(a)$ follows from \eref{eq:tirling}, $(c)$ follows from \eref{eq:Weingarten}, and $(b)$ holds because 
\begin{align}
\Tr_A \left\lbrace \left(\ketbra{b}{b}_{A_1} \!\otimes \mathbb{S}_{A_2,B} \right)^{\otimes k} \left[  (R_{\pi}^{\top})_A \otimes I_B^{\otimes k} \right] \right\rbrace 
&= \Tr \left[\left( \ketbra{b}{b}_{A_1}\right) ^{\otimes k}\left( R_{\pi}^{\top}\right)_{\!A_1}\right] 
\Tr_{A_2} \left\lbrace \left( \mathbb{S}_{A_2,B}\right) ^{\otimes k} \left[(R_{\pi}^{\top})_{A_2} \otimes I_B^{\otimes k} \right] \right\rbrace 
= (R_\pi^\top)_B .   
\end{align}

By combining Eqs.~\eqref{eq:ELambdaU1}, \eqref{eq:ELambdaU2}, and \eqref{eq:ELambdaU3}, 
the expectation of $\hat{\Lambda}_k^{U}$ can be expressed as 
\begin{align}
\E \left[ \hat{\Lambda}_k^{U} \right] 
&=  \frac{1}{k!} \sum_{\pi\in \mathrm{S}_k} \Tr \left[ \rho^{\otimes k} (R_\pi)_A \otimes (R_\pi^\top)_B \right] 
+ \bigo{d_{\!A}^{-1}} \sum_{\pi, \tau \in \mathrm{S}_k} \Tr \left[ \rho^{\otimes k} 
(R_\tau)_{A} \otimes (R_\pi^\top)_B \right] 
\nonumber \\
&= \frac{1}{k!} \sum_{\pi\in \mathrm{S}_k} \Tr \left[ \left( \rho^{\top_{\!B}}\right)^{\otimes k} (R_\pi)_A \otimes (R_\pi)_B \right] 
+ \bigo{d_{\!A}^{\,-1}}
= \zeta_k^{\PT} + \bigo{d_{\!A}^{\,-1}} .
\end{align}
This completes the proof of \lref{lam:biasLambdaU}. 
\end{proof}

\subsection{Variance of the estimator \texorpdfstring{$\hat{\Lambda}_k^{U}$}{}}\label{sec:VarPTmoment}

\begin{lemma}\label{lem:PTVariance}
Suppose $k>1$ is a constant integer independent of the system dimension, and $U$ is a random unitary  sampled from a $2k$-design ensemble on $\caH_A$. Then 
\begin{align}
\Var\left[\hat{\Lambda}_k^{U} \right] 
= \bigo{\max\left\lbrace \frac{1}{d_{A_1}} ,\frac{d^{k}}{d_{A_1} N_{\!M}^{k}}\right\rbrace  } ,
\end{align}
where $d=d_Ad_B$. 
\end{lemma}

\begin{proof}[Proof of \lref{lem:PTVariance}]
The variance of $\hat{\Lambda}_k^{U}$ reads
\begin{align}\label{eq:PTVariance1} 
\Var\left[ \hat{\Lambda}_k^{U} \right] 
&= \E\left[ \left( \hat{\Lambda}_k^{U} \right) ^2 \right] - \E\left[ \hat{\Lambda}_k^{U}  \right]^2 
= \underset{U }{\E} \, \underset{\,\mathbf{b}_{U},\mathbf{r}_{U}}{\E}  \left[ \left( \hat{\Lambda}_k^{U}  \right) ^2 \Big| U \right]
- \E\left[ \hat{\Lambda}_k^{U}  \right]^2 
\nonumber\\
&= \binom{N_M}{k}^{\!-2} \!\left( \frac{d_{\!A}^{\,k}}{k!\, d_{\!A_1} } \right)^{\!2}\!\!
\sum_{\substack{1\leq i_1<\dots<i_k\leq N_M \\ 1\leq j_1<\dots<j_k\leq N_M}} \,
f_2\left( \mathbf{b}_{U},\mathbf{r}_{U},\left\lbrace i_1,\dots,i_k,j_1,\dots,j_k\right\rbrace \right)
- \,\E\left[ \hat{\Lambda}_k^{U}  \right]^2,
\end{align}
where the function $f_2$ is defined as 
\begin{align}\label{eq:definef2}
f_2\left( \mathbf{b}_{U},\mathbf{r}_{U},\left\lbrace i_1,\dots,i_k,j_1,\dots,j_k\right\rbrace \right)
:=
\underset{U }{\E} \, \underset{\,\mathbf{b}_{U},\mathbf{r}_{U}}{\E}  
\Big[
\left( \hat{r}_{i_1}\cdots \hat{r}_{i_k} \hat{r}_{j_1}\cdots \hat{r}_{j_k} \right) \mathbf{1}\big\{\hat{b}_{i_1}=\dots=\hat{b}_{i_k}\big\}
\mathbf{1}\big\{\hat{b}_{j_1}=\dots=\hat{b}_{j_k}\big\} \Big| U \Big]. 
\end{align}

Let ${\rm Co}[(i_1,\dots,i_k);(j_1,\dots,j_k)]$ denote the number of common elements shared by $(i_1,\dots,i_k)$ and $(j_1,\dots,j_k)$. Then \eref{eq:PTVariance1}  implies that 
\begin{align}\label{eq:PTVariance13} 
&\Var\left[ \hat{\Lambda}_k^{U} \right] +\E\left[ \hat{\Lambda}_k^{U}  \right]^2
\nonumber\\
&= \binom{N_M}{k}^{\!-2} \!\left( \frac{d_{\!A}^{\,k}}{k!\, d_{\!A_1} } \right)^{\!2} 
\sum_{l=0}^{k} 
\sum_{\substack{1\leq i_1<\dots<i_k\leq N_M \\ 1\leq j_1<\dots<j_k\leq N_M\\{\rm Co}[(i_1,\dots,i_k);(j_1,\dots,j_k)]=k-l}}
f_2\left( \mathbf{b}_{U},\mathbf{r}_{U},\left\lbrace i_1,\dots,i_k,j_1,\dots,j_k\right\rbrace \right)
\nonumber\\
&= \binom{N_M}{k}^{\!-2} \!\left( \frac{d_{\!A}^{\,k}}{k!\, d_{\!A_1} } \right)^{\!2} 
\sum_{l=0}^{k} 
\binom{N_M}{k+l} \binom{k+l}{k} \binom{k}{l} \,
f_2\left( \mathbf{b}_{U},\mathbf{r}_{U},\left\lbrace i_1,\dots,i_k,j_1,\dots,j_k\right\rbrace \right)\big|_{{\rm Co}[(i_1,\dots,i_k);(j_1,\dots,j_k)]=k-l}
\nonumber\\
&\stackrel{(a)}{=}\binom{N_M}{k}^{\!-2} \!\left( \frac{d_{\!A}^{\,k}}{k!\, d_{\!A_1} } \right)^{\!2} 
\left\lbrace \,  
\sum_{l=0}^{k-1} 
\binom{N_M}{k+l} \binom{k+l}{k} \binom{k}{l} \, 
\bigo{ \frac{d_{A_1} d_B^{\,k-l}}{d_A^{\,k+l} } } 
+   
\binom{N_M}{2k} \binom{2k}{k}
\left[\bigo{ \frac{d_{A_1}}{d_{\!A}^{\,2k}} }+ \frac{d_{A_1}^{\,2}}{d_{\!A}^{\,2k}} \left( k!\,\zeta_k^{\PT}\right)^2
\right]\right\rbrace 
\nonumber\\
&= 
\sum_{l=0}^{k-1}\binom{N_M}{k}^{-2}  
\binom{N_M}{k+l} \bigo{d^{k-l} d_{\!A_1}^{\,-1}}
+ 
\binom{N_M}{k}^{-2}  \binom{N_M}{2k}\binom{2k}{k} 
\left[\left( \zeta_k^{\PT}\right)^2
+ \bigo{d_{\!A_1}^{\,-1}}
\right],
\end{align}
where $(a)$ follows from \lref{lem:PTVarianceCases} below. 
Note that 
\begin{align} 
	\binom{N_M}{k}^{-2}  \binom{N_M}{2k}\binom{2k}{k}
	&= \left( \frac{N_M-k}{N_M}\right)  \left( \frac{N_M-k-1}{N_M-1}\right)  \cdots \left(\frac{N_M-2k+1}{N_M-k+1}\right) 
	\leq 1.
\end{align}
By combining this inequality with \eref{eq:PTVariance13}, we obtain 
\begin{align}\label{eq:PTVariance14} 
\Var\left[ \hat{\Lambda}_k^{U} \right] 
&\leq 
\sum_{l=0}^{k-1}\binom{N_M}{k}^{-2}  
\binom{N_M}{k+l} \bigo{d^{k-l} d_{\!A_1}^{\,-1}}
+ 
\left[\left( \zeta_k^{\PT}\right)^2
+ \bigo{d_{\!A_1}^{\,-1}}
\right] - \E\left[ \hat{\Lambda}_k^{U}  \right]^2
\nonumber\\
&\stackrel{(b)}{=}\sum_{l=0}^{k-1} \bigo{\frac{ d^{k-l}}{d_{A_1} N_M^{k-l}}}  
+ \left[\left( \zeta_k^{\PT}\right)^2+ \bigo{d_{\!A_1}^{\,-1}}\right] 
- \left[ \zeta_k^{\PT} + \bigo{d_{\!A}^{\,-1}} \right]^2
= \bigo{\max\left\lbrace \frac{1}{d_{A_1}} ,\frac{d^{k}}{d_{A_1} N_{\!M}^{k}}\right\rbrace  } , 
\end{align}
where we use \lref{lam:biasLambdaU} in $(b)$. 
This completes the proof of \lref{lem:PTVariance}. 
\end{proof}

\begin{lemma}\label{lem:PTVarianceCases}
Suppose $k>1$ is a constant integer independent of the system dimension, $U$ is a random unitary  sampled from a $2k$-design ensemble on $\caH_A$, integers $1\leq i_1< i_2<\dots<i_k\leq N_M$ and $1\leq j_1<j_2<\dots<j_k\leq N_M$. 
Let
$\mathbf{b}_{U} = \big\{\hat{b}_1, \dots, \hat{b}_{N_M}\big\}$ and $\mathbf{r}_{U} = \{\hat{r}_1, \dots, \hat{r}_{N_M}\}$ denote the measurement outcomes of the PTME protocol for unitary $U$ (see Algorithm~\ref{alg:PTME}). 
Then the function $f_2$ defined in \eref{eq:definef2} satisfies 
\begin{align}\label{eq:PTVarianceCase1}
f_2\left( \mathbf{b}_{U},\mathbf{r}_{U},\left\lbrace i_1,\dots,i_k,j_1,\dots,j_k\right\rbrace \right)
=
\bigo{ \frac{d_{A_1} d_B^{\,k-l}}{d_A^{\,k+l} } }
\end{align}
when ${\rm Co}[(i_1,\dots,i_k);(j_1,\dots,j_k)]=k-l$ with $l\in\{0,1,\dots,k-1\}$, and 
\begin{align}\label{eq:PTVarianceCase2}
f_2\left( \mathbf{b}_{U},\mathbf{r}_{U},\left\lbrace i_1,\dots,i_k,j_1,\dots,j_k\right\rbrace \right)
=
\bigo{ \frac{d_{A_1}}{d_{\!A}^{\,2k}} }+ \frac{d_{A_1}^{\,2}}{d_{\!A}^{\,2k}} \left( k!\,\zeta_k^{\PT}\right)^2
\end{align}
when ${\rm Co}[(i_1,\dots,i_k);(j_1,\dots,j_k)]=0$. 
\end{lemma}

\subsubsection{Proof of \eref{eq:PTVarianceCase1} in \lref{lem:PTVarianceCases}}

Assume that ${\rm Co}[(i_1,\dots,i_k);(j_1,\dots,j_k)]=k-l$ with $l\in\{0,1,\dots,k-1\}$. In this case, we have 
\begin{align}\label{eq:PTVariance2} 
f_2\left( \mathbf{b}_{U},\mathbf{r}_{U},\left\lbrace i_1,\dots,i_k,j_1,\dots,j_k\right\rbrace \right)
=\underset{U }{\E} \, \underset{\,\mathbf{b}_{U},\mathbf{r}_{U}}{\E}   
\Big[ \left( \hat{r}_{i_1}\cdots \hat{r}_{i_k} \hat{r}_{j_1}\cdots \hat{r}_{j_k} \right) 
\textbf{1}\big\{\hat{b}_{i_1}=\dots=\hat{b}_{i_k}=\hat{b}_{j_1}=\dots=\hat{b}_{j_k}\big\} \Big| U \Big].
\end{align}
Without loss of generality, we assume that $(i_1,\dots,i_k)$ and $(j_1,\dots,j_k)$ coincide on the last $k-l$ pairs. 
Then 
\begin{align}\label{eq:PTVariance3} 
&f_2\left( \mathbf{b}_{U},\mathbf{r}_{U},\left\lbrace i_1,\dots,i_k,j_1,\dots,j_k\right\rbrace \right)
\nonumber\\
&= \E_{U} \left\lbrace \sum_{\ell=0}^{2l} \sum_{b=0}^{d_{A_1}-1} (-1)^{\ell} 
\Pr\left(\#\left[ \text{$-1$ in $\left\lbrace \hat{r}_{i_1},\dots, \hat{r}_{i_l},\hat{r}_{j_1},\dots,\hat{r}_{j_l}\right\rbrace $}\right] =\ell,\, \text{and } \hat{b}_{i_1}=\dots=\hat{b}_{i_k}=\hat{b}_{j_1}=\dots=\hat{b}_{j_{l}}=b \,\Big| U \right)  \right\rbrace 
\nonumber \\
&= \E_{U} \left\lbrace \sum_{\ell=0}^{2l} \sum_{b=0}^{d_{A_1}-1} \binom{2l}{\ell} (-1)^{\ell} 
\Pr\left(\hat{r}=-1, \hat{b}=b \,\Big| U\right)^\ell 
\Pr\left(\hat{r}=+1, \hat{b}=b \,\Big| U\right)^{2l-\ell} 
\Pr\left(\hat{b}=b \,\Big| U\right)^{k-l} \right\rbrace 
\nonumber \\
&= \E_{U} \left\lbrace \sum_{b=0}^{d_{A_1}-1} \Pr\left(\hat{b}=b \,\Big| U\right)^{k-l}
\sum_{\ell=0}^{2l} \binom{2l}{\ell} (-1)^{\ell} 
\Pr\left(\hat{r}=-1, \hat{b}=b \,\Big| U\right)^\ell 
\Pr\left(\hat{r}=+1, \hat{b}=b \,\Big| U\right)^{2l-\ell} 
\right\rbrace 
\nonumber \\
&= \E_{U} \left\lbrace \sum_{b=0}^{d_{A_1}-1} \Pr\left(\hat{b}=b \,\Big| U\right)^{k-l}
\left[ \Pr\left(\hat{r}=+1, \hat{b}=b\,\Big| U\right)-\Pr\left(\hat{r}=-1, \hat{b}=b \,\Big| U\right) \right]^{2l}
\right\rbrace 
\nonumber \\
&= \sum_{b=0}^{d_{A_1}-1} \,\underbrace{ \E_{U} \left\lbrace 
\Tr\left[\rho_{U} \left( \ketbra{b}{b}_{A_1}\!\otimes I_{\!A_2,B} \right)  \right]^{k-l}
\Tr\left[\rho_{U} \left( \ketbra{b}{b}_{A_1} \otimes \mathbb{S}_{A_2,B} \right)  \right]^{2l}
\right\rbrace}_{(*1)}, 
\end{align}
where in the last equality we use 
\begin{align}\label{eq:Pr+-}
\Pr\left(\hat{b}=b\,\Big| U\right) =\Tr\left[\rho_{U} \left( \ketbra{b}{b}_{A_1}\!\otimes I_{\!A_2,B} \right)  \right],
\quad
\Pr\left(\hat{r}=\pm 1, \hat{b}=b\,\Big| U\right) = 
\Tr\left[\rho_{U} \left( \ketbra{b}{b}_{A_1} \!\otimes \frac{I_{A_2,B}\pm\mathbb{S}_{A_2,B}}{2} \right)  \right].  
\end{align}

The term $(*1)$ in \eref{eq:PTVariance3} can be further expressed as
\begin{align}\label{eq:PTVariance4} 
(*1)
&= 
\E_{U} \left\lbrace 
\Tr\left[\rho(U^{\dag}\otimes I_B) \left(\ketbra{b}{b}_{A_1} \!\otimes I_{A_2,B} \right) (U\otimes I_B) \right]^{k-l}
\Tr\left[\rho(U^{\dag}\otimes I_B) \left(\ketbra{b}{b}_{A_1} \!\otimes \mathbb{S}_{A_2,B} \right) (U\otimes I_B) \right]^{2l} 
\right\rbrace
\nonumber\\
&= 
\Tr \!\bigg\{ \rho^{\otimes (k+l)} \, \E_{U} \left\{(U^{\dag}\otimes I_B)^{\otimes (k+l)} 
\left[\left(\ketbra{b}{b}_{A_1} \!\otimes I_{A_2,B} \right)^{\otimes (k-l)}\otimes
\left(\ketbra{b}{b}_{A_1} \!\otimes \mathbb{S}_{A_2,B} \right)^{\otimes 2l} \right]
(U\otimes I_B)^{\otimes (k+l)} \right\} \bigg\}
\nonumber\\
&= 
\Tr\left\{\rho^{\otimes (k+l)} \, \left(\Phi_{A}\otimes \mathcal{I}_{B} \right) 
\left[\left(\ketbra{b}{b}_{A_1} \!\otimes I_{A_2,B} \right)^{\otimes (k-l)}\otimes
\left(\ketbra{b}{b}_{A_1} \!\otimes \mathbb{S}_{A_2,B} \right)^{\otimes 2l} \right]  \right\}
\nonumber\\
&= \Tr \left( \rho^{\otimes (k+l)}\mathcal{Q}_{b,k,l}\right) , 
\end{align}
Here, $\mathcal{Q}_{b,k,l}$ in the last line is defined as 
\begin{align}\label{eq:PTVariance5} 
\mathcal{Q}_{b,k,l}
&:=
\left(\Phi_{A}\otimes \mathcal{I}_{B} \right) 
\left[\left(\ketbra{b}{b}_{A_1} \!\otimes I_{A_2,B} \right)^{\otimes (k-l)}\otimes
\left(\ketbra{b}{b}_{A_1} \!\otimes \mathbb{S}_{A_2,B} \right)^{\otimes 2l} \right]
\nonumber\\
&\,\stackrel{(a)}{=} \sum_{\tau,\pi\in \mathrm{S}_{k+l}} \Wg_{\pi,\tau}^{(k+l,d_A)} \, 
(R_{\tau})_A \otimes 
\underbrace{\Tr_A \left\lbrace  \left[\left(\ketbra{b}{b}_{A_1} \!\otimes I_{A_2,B} \right)^{\otimes (k-l)}\otimes
\left(\ketbra{b}{b}_{A_1} \!\otimes \mathbb{S}_{A_2,B} \right)^{\otimes 2l} \right] 
\left[ (R_{\pi}^{\top})_A \otimes I_B^{\otimes (k+l)} \right] \right\rbrace }_{(*2)}
\nonumber\\
&\,\stackrel{(b)}{=} \sum_{\tau,\pi\in \mathrm{S}_{k+l}} \Wg_{\pi,\tau}^{(k+l,d_A)} \bigo{d_B^{\,k-l}} \, 
(R_{\tau})_A \otimes \left(R_{\pi'} \right)_B
\nonumber\\
&\,\stackrel{(c)}{=} \sum_{\tau,\pi\in \mathrm{S}_{k+l}} \bigo{d_A^{\,-k-l} d_B^{\,k-l}} \, 
(R_{\tau})_A \otimes \left(R_{\pi'} \right)_B, 
\end{align}
where $\pi'\in \mathrm{S}_{k+l}$ is a permutation element related to $\pi$. 
In \eref{eq:PTVariance5}, equality $(a)$ follows from \eref{eq:tirling}, $(c)$ follows from \eref{eq:Weingarten}, and $(b)$ holds because 
\begin{align}
(*2)
&= \Tr \left[\left( \ketbra{b}{b}_{A_1}\right)^{\otimes (k+l)}\left( R_{\pi}^{\top}\right)_{\!A_1}\right] 
\Tr_{A_2} \left\lbrace \left[ ( I_{A_2,B})^{\otimes (k-l)}\otimes ( \mathbb{S}_{A_2,B})^{\otimes 2l}\right]  \left[(R_{\pi}^{\top})_{A_2} \otimes I_B^{\otimes (k+l)} \right] \right\rbrace 
\nonumber\\
&= \Tr_{A_2} \left\lbrace \left[ ( I_{A_2,B})^{\otimes (k-l)}\otimes ( \mathbb{S}_{A_2,B})^{\otimes 2l}\right]  \left[(R_{\pi}^{\top})_{A_2} \otimes I_B^{\otimes (k+l)} \right] \right\rbrace 
\nonumber\\
&= I_B^{\otimes (k-l)} \otimes \Tr_{1,2,\dots,k-l}\left[ \left(R_\pi^\top \right)_B\right] 
= \bigo{d_B^{\,k-l}} \left(R_{\pi'} \right)_B
\end{align}
for some $\pi'\in \mathrm{S}_{k+l}$.

By combining Eqs.~\eqref{eq:PTVariance3}, \eqref{eq:PTVariance4}, and \eqref{eq:PTVariance5}, we have
\begin{align}\label{eq:PTVariance6} 
&f_2\left( \mathbf{b}_{U},\mathbf{r}_{U},\left\lbrace i_1,\dots,i_k,j_1,\dots,j_k\right\rbrace \right)
= \sum_{b=0}^{d_{A_1}-1} \,\Tr \left( \rho^{\otimes (k+l)}\mathcal{Q}_{b,k,l}\right)
\nonumber\\
&\quad = \sum_{b=0}^{d_{A_1}-1} \sum_{\tau,\pi\in \mathrm{S}_{k+l}} \bigo{d_A^{\,-k-l} d_B^{\,k-l}} \, 
\Tr  \left\lbrace \rho^{\otimes (k+l)} \left[(R_{\tau})_A \otimes \left(R_{\pi'} \right)_B\right] \right\rbrace 
= \bigo{ d_{A_1} d_A^{\,-k-l} d_B^{\,k-l} } , 
\end{align}
which confirms \eref{eq:PTVarianceCase1} in \lref{lem:PTVarianceCases}.

\subsubsection{Proof of \eref{eq:PTVarianceCase2} in \lref{lem:PTVarianceCases}}

Assume that ${\rm Co}[(i_1,\dots,i_k);(j_1,\dots,j_k)]=0$. In this case, we have
\begin{align}\label{eq:PTVariance7} 
&f_2\left( \mathbf{b}_{U},\mathbf{r}_{U},\left\lbrace i_1,\dots,i_k,j_1,\dots,j_k\right\rbrace \right)
\nonumber\\
&=\E_{U} \left\lbrace \left[ \underset{\,\mathbf{b}_{U},\mathbf{r}_{U}}{\E}  \left( \hat{r}_{i_1}\hat{r}_{i_2}\cdots \hat{r}_{i_k}  \textbf{1}\big\{\hat{b}_{i_1}=\dots=\hat{b}_{i_k}\big\} \Big| U \right) \right]^2 \right\rbrace 
\nonumber \\
&= \E_{U} \left\lbrace \left[ \, 
\sum_{\ell=0}^{k} \sum_{b=0}^{d_{A_1}-1} (-1)^{\ell} 
\Pr\left(
\#\left[ \text{$-1$ in $\left\lbrace \hat{r}_{i_1},\dots, \hat{r}_{i_k}\right\rbrace $}\right] =\ell,\,\text{and }
\hat{b}_{i_1}=\dots=\hat{b}_{i_k}=b \,\Big| U \right) \right]^2  \right\rbrace 
\nonumber \\
&= \E_{U} \left\lbrace \left[ \,
\sum_{b=0}^{d_{A_1}-1} \sum_{\ell=0}^{k}  (-1)^{\ell}  
\binom{k}{\ell} 
\Pr\left(\hat{r}=-1, \hat{b}=b \,\Big| U\right)^{\ell}
\Pr\left(\hat{r}=+1, \hat{b}=b \,\Big| U\right)^{k-\ell} 
\right]^2 \right\}
\nonumber \\
&= \E_{U} \left\lbrace \left( \,
\sum_{b=0}^{d_{A_1}-1} 
\left[ \Pr\left(\hat{r}=+1, \hat{b}=b \,\Big| U\right)-
\Pr\left(\hat{r}=-1, \hat{b}=b \,\Big| U\right) \right]^k 
\right)^{\!\!2} \right\}
\nonumber \\
&\stackrel{(a)}{=} \E_{U} \left\lbrace \left( \, \sum_{b=0}^{d_{A_1}-1}
\Tr\left[\rho_{U} \left( \ketbra{b}{b}_{A_1} \!\otimes \mathbb{S}_{A_2,B} \right) \right]^{k}
\right)^{\!\!2} \right\rbrace
\nonumber \\
&= \sum_{b=0}^{d_{A_1}-1}\;
\underbrace{ \E_{U} \left\lbrace 
\Tr\left[\rho_{U} \left( \ketbra{b}{b}_{A_1} \!\otimes \mathbb{S}_{A_2,B} \right) \right]^{2k}
\right\rbrace }_{(*3)}
+ \sum_{b\ne b'}\;
\underbrace{ \E_{U} \left\lbrace 
\Tr\left[\rho_{U} \left( \ketbra{b}{b}_{A_1} \!\otimes \mathbb{S}_{A_2,B} \right)  \right]^{k}
\Tr\left[\rho_{U} \left( \ketbra{b'}{b'}_{A_1} \!\otimes \mathbb{S}_{A_2,B} \right) \right]^{k}
\right\rbrace}_{(*4)},
\end{align}
where $(a)$ follows from \eref{eq:Pr+-}. 
By using an argument similar to that in Eqs.~\eqref{eq:ELambdaU2} and \eqref{eq:ELambdaU3}, we can derive that
\begin{align}\label{eq:PTVariance8} 
(*3)&= d_{\!A}^{\,-2k}
\Tr\left\lbrace \rho^{\otimes 2k} \left[\, \sum_{\pi\in \mathrm{S}_{2k}} (R_\pi)_A \otimes (R_\pi^\top)_B 
+ \bigo{d_{\!A}^{\,-1}} \sum_{\pi, \tau \in \mathrm{S}_{2k}} 
(R_\tau)_{A} \otimes (R_\pi^\top)_B \right] \right\rbrace 
= \bigo{d_{\!A}^{\,-2k}}. 
\end{align}

For the term $(*4)$ in \eref{eq:PTVariance7}, we have 
\begin{align}\label{eq:PTVariance9} 
(*4)&= 
\E_{U} \left\lbrace 
\Tr\left[\rho(U^{\dag}\otimes I_B) \left(\ketbra{b}{b}_{A_1}   \!\otimes \mathbb{S}_{A_2,B} \right) (U\otimes I_B) \right]^{k}
\Tr\left[\rho(U^{\dag}\otimes I_B) \left(\ketbra{b'}{b'}_{A_1} \!\otimes \mathbb{S}_{A_2,B} \right) (U\otimes I_B) \right]^{k} 
\right\rbrace
\nonumber\\
&= 
\Tr \!\bigg\{ \rho^{\otimes 2k} \, \E_{U} \left\{(U^{\dag}\otimes I_B)^{\otimes 2k} 
\left[\left(\ketbra{b}{b}_{A_1}^{\otimes k}\otimes\ketbra{b'}{b'}_{A_1}^{\otimes k} \right) \!\otimes \left( \mathbb{S}_{A_2,B} \right)^{\otimes 2k} \right]
(U\otimes I_B)^{\otimes 2k} \right\} \bigg\}
\nonumber\\
&= 
\Tr\left\{\rho^{\otimes 2k} \, \left(\Phi_{A}\otimes \mathcal{I}_{B} \right) 
\left[\left(\ketbra{b}{b}_{A_1}^{\otimes k}\otimes\ketbra{b'}{b'}_{A_1}^{\otimes k} \right) \!\otimes \left( \mathbb{S}_{A_2,B} \right)^{\otimes 2k} \right] \right\}
\nonumber\\
&= \Tr \left( \rho^{\otimes 2k}\tilde{\mathcal{Q}}_{b,b'\!,k}\right) , 
\end{align}
where $\tilde{\mathcal{Q}}_{b,b'\!,k}$ in the last line is defined as 
\begin{align}\label{eq:PTVariance10} 
\tilde{\mathcal{Q}}_{b,b'\!,k} &:= 
\left(\Phi_{A}\otimes \mathcal{I}_{B} \right) 
\left[\left(\ketbra{b}{b}_{A_1}^{\otimes k}\otimes\ketbra{b'}{b'}_{A_1}^{\otimes k} \right) \!\otimes \left( \mathbb{S}_{A_2,B} \right)^{\otimes 2k} \right]
\nonumber\\
&\,\stackrel{(a)}{=} \sum_{\tau,\pi\in \mathrm{S}_{2k}} \Wg_{\pi,\tau}^{(2k,d_A)} \, 
(R_{\tau})_A \otimes 
\underbrace{\Tr_A \left\lbrace \left[\left(\ketbra{b}{b}_{A_1}^{\otimes k}\otimes\ketbra{b'}{b'}_{A_1}^{\otimes k} \right) \!\otimes \left( \mathbb{S}_{A_2,B} \right)^{\otimes 2k} \right]
\left[  (R_{\pi}^{\top})_A \otimes I_B^{\otimes 2k} \right] \right\rbrace }_{(*5)}
\nonumber \\
&\,\stackrel{(b)}{=} \sum_{\tau \in \mathrm{S}_{2k}} \sum_{\pi_1,\pi_2 \in \mathrm{S}_{k}} \Wg_{(\pi_1,\pi_2),\tau}^{(2k,d_A)} \, 
(R_\tau)_{A} \otimes \left[ R_{(\pi_1,\pi_2)}^\top\right] _B
\nonumber \\
&\,\stackrel{(c)}{=} d_{\!A}^{\,-2k} \left\lbrace \,\sum_{\pi_1,\pi_2\in \mathrm{S}_{k}} 
\left[R_{(\pi_1,\pi_2)}\right]_A \otimes \left[ R_{(\pi_1,\pi_2)}^\top\right]_B
+ \bigo{d_{\!A}^{\,-1}} \sum_{\tau \in \mathrm{S}_{2k}} \sum_{\pi_1,\pi_2 \in \mathrm{S}_{k}}  
(R_\tau)_{A} \otimes \left[ R_{(\pi_1,\pi_2)}^\top\right]_B \right\rbrace  .
\end{align}
Here, $(a)$ follows from \eref{eq:tirling}, $(c)$ follows from \eref{eq:Weingarten}, and $(b)$ holds because 
\begin{align}
(*5)
&= \Tr \left[\left(\ketbra{b}{b}^{\otimes k}\otimes\ketbra{b'}{b'}^{\otimes k} \right) R_{\pi}^{\top}\right] 
\Tr_{A_2} \left\lbrace \left( \mathbb{S}_{A_2,B} \right)^{\otimes 2k}  \left[(R_{\pi}^{\top})_{A_2} \otimes I_B^{\otimes 2k} \right] \right\rbrace 
\nonumber\\
&= \Tr \left[\left(\ketbra{b}{b}^{\otimes k}\otimes\ketbra{b'}{b'}^{\otimes k} \right) R_{\pi}^{\top}\right]  (R_{\pi}^{\top})_{B}
\nonumber\\
&=
\begin{cases}
 (R_{\pi}^{\top})_{B} & \text{when } \pi=(\pi_1,\pi_2) \text{ for some } \pi_1,\pi_2\in\mathrm{S}_{k}, \\
0   &  \text{otherwise}.   
\end{cases}
\end{align}
Equations \eqref{eq:PTVariance9} and \eqref{eq:PTVariance10} imply that
\begin{align}\label{eq:PTVariance11} 
(*4)&= 
d_{\!A}^{\,-2k} \left\lbrace \,\sum_{\pi_1,\pi_2\in \mathrm{S}_{k}} 
\Tr \left\{ \rho^{\otimes 2k} \left[R_{(\pi_1,\pi_2)}\right]_A \otimes \left[ R_{(\pi_1,\pi_2)}^\top\right]_B \right\} 
+ \bigo{d_{\!A}^{\,-1}} \sum_{\tau \in \mathrm{S}_{2k}} \sum_{\pi_1,\pi_2 \in \mathrm{S}_{k}}  
\Tr \left\{ \rho^{\otimes 2k} (R_\tau)_{A} \otimes \left[ R_{(\pi_1,\pi_2)}^\top\right]_B \right\} \right\rbrace 
\nonumber\\
&=d_{\!A}^{\,-2k} \sum_{\pi_1,\pi_2\in \mathrm{S}_{k}} 
\Tr \left[  \rho^{\otimes k} (R_{\pi_1})_A \otimes (R_{\pi_1}^\top)_B \right] 
\Tr \left[  \rho^{\otimes k} (R_{\pi_2})_A \otimes (R_{\pi_2}^\top)_B \right] 
+ \bigo{d_{\!A}^{\,-2k-1}} 
\nonumber\\
&=d_{\!A}^{\,-2k} \left\lbrace \,\sum_{\pi\in \mathrm{S}_{k}} 
\Tr \left[  \left( \rho^{\top_{\!B}}\right) ^{\otimes k} (R_{\pi})_A \otimes (R_{\pi})_B \right] \right\rbrace^2 
+ \bigo{d_{\!A}^{\,-2k-1}} 
\nonumber\\
&= \frac{k!^2}{d_{\!A}^{\,2k}} \left( \zeta_k^{\PT}\right)^2+ \bigo{d_{\!A}^{\,-2k-1}}. 
\end{align}

By combining Eqs.~\eqref{eq:PTVariance7}, \eqref{eq:PTVariance8}, and \eqref{eq:PTVariance11}, we have 
\begin{align}\label{eq:PTVariance12} 
f_2\left( \mathbf{b}_{U},\mathbf{r}_{U},\left\lbrace i_1,\dots,i_k,j_1,\dots,j_k\right\rbrace \right)
&= \sum_{b=0}^{d_{A_1}-1}\bigo{d_{\!A}^{\,-2k}} + \sum_{b\ne b'}\left[\frac{k!^2}{d_{\!A}^{\,2k}} \left( \zeta_k^{\PT}\right) ^2+ \bigo{d_{\!A}^{\,-2k-1}} \right] 
\nonumber\\
&=\bigo{d_{A_1} d_{\!A}^{\,-2k}}+ d_{\!A_1}^{\,2} d_{\!A}^{\,-2k} \left( k!\,\zeta_k^{\PT}\right)^2,
\end{align}
which confirms \eref{eq:PTVarianceCase2} in \lref{lem:PTVarianceCases}.

\subsection{Proof of \tref{thm:PTResult} in the main text}\label{sec:ProofthmPTResult}

\begin{theorem}[Restatement of \tref{thm:PTResult}]
\label{thm:PTResultSM}
Suppose $\rho\in\caD(\caH_A\otimes \caH_B)$ and 
$t\geq 2$ is a constant integer independent of the system dimensions $d_A$ and $d_B$.
For any $0<\epsilon<1$, the PTME protocol (given in  Algorithm~\ref{alg:PTME} of the End Matter) can return $\epsilon$-additive error estimates for all 
$p_2^{\PT},\dots,p_t^{\PT}\!$ with high probability, provided that: $(i)$ the random unitaries are sampled from a $2t$-design ensemble, and
\begin{align} 
(ii)\ d_{\!A}\epsilon\geq c_1,
\qquad 
(iii)\ N_U=\max\left\lbrace 1, c_2 \,\Xi_2 \right\rbrace, 
\qquad 
(iv)\ N_M\geq c_3\,d_{\!A} d_B \min\left\lbrace \Xi_2^{1/t}, 1 \right\rbrace
\end{align}
for some constants $c_1,c_2,c_3>0$ independent of $d_A,d_B,\epsilon$. Here, $\Xi_2$ is a shorthand for $\Xi_2(d_A,d_B,\epsilon)=d_B/(d_{A}\epsilon^2)$. 
\end{theorem}

\begin{proof}[Proof of \tref{thm:PTResult}]
Let $0<\delta<1$ be a constant independent of $d_A,d_B$ and $\epsilon$. 
According to the union bound and \lref{lem:PTMestTrans} below, in order to estimate all 
$p_2^{\PT},p_3^{\PT},\dots,p_t^{\PT}$ within $\epsilon$ additive error and failure probability at most $\delta$, it 
suffices to ensure that the following condition holds for all $k=2,\dots,t$:  
\begin{equation}\label{eq:EachUBcepsPT}
\Pr\left\{ \left| \hat{\zeta}_k^{\PT} - \zeta_k^{\PT}\right| > c\epsilon \right\} \leq \frac{\delta}{t},
\end{equation}
where $c>0$ is some constant factor independent of $d_A,d_B,\epsilon$, and $\hat\zeta_k^{\PT} = \frac{1}{N_U} \sum_{s=1}^{N_U} \hat{\Lambda}_k^{U_s}$ is 
our estimator for $\zeta_k^{\PT}$ (see Algorithm~\ref{alg:PTME}). 
Note that \eref{eq:EachUBcepsPT} can be achieved if the following two relations hold: 
\begin{align}\label{eq:CondiUBcepsPT}
(a)\ \, \left|\E\left[ \hat\zeta_k^{\PT}\right] -\zeta_k^{\PT}\right| < \frac{c\epsilon}{2} 
\quad\text{and}\quad
(b)\ \Pr\left\{ \left|\hat\zeta_k^{\PT}-\E\left[ \hat\zeta_k^{\PT}\right] \right| > \frac{c\epsilon}{2}  \right\} \leq \frac{\delta}{t}. 
\end{align} 

By \lref{lam:biasLambdaU}, $\hat\zeta_k^{\PT}$ is 
an estimator for $\zeta_k^{\PT}$ with bias $\mathcal O(d_{\!A}^{\,-1})$. 
Thus, the relation $(a)$ in \eref{eq:CondiUBcepsPT} holds provided that $d_{\!A}\epsilon\geq c_1$ for some constants $c_1>0$.

In addition, Chebyshev's inequality implies that 
\begin{align}
\Pr\left\{ \left|\hat\zeta_k^{\PT}-\E\left[ \hat\zeta_k^{\PT}\right] \right| > \frac{c\epsilon}{2}  \right\}
\leq \frac{4}{c^2\epsilon^2} \Var\left[ \hat\zeta_k^{\PT}\right] 
= \frac{4}{c^2\epsilon^2 N_U} \Var\!\Big[ \hat{\Lambda}_k^{U} \Big] 
= \bigo{\frac{1}{\epsilon^2 N_U} \cdot \max\left\lbrace \frac{1}{d_{A_1}} ,\frac{d^{k}}{d_{A_1} N_{\!M}^{k}}\right\rbrace},  
\end{align}
where $d=d_A d_B$, and the last equality follows from \lref{lem:PTVariance}.
Hence, to make the relation $(b)$ in \eref{eq:CondiUBcepsPT} hold, it is sufficient to take 
\begin{align} 
\frac{1}{\epsilon^2 N_U d_{A_1}} \leq 1
\quad\text{and}\quad  
\frac{d^{k}}{\epsilon^2 N_U d_{A_1} N_{\!M}^{k}} \leq 1 \quad \forall k=2,\dots,t,
\end{align}
where we omit constant coefficients, including $\delta$ and $t$. 
It is straightforward to verify that these two conditions are satisfied by choosing $N_U$ and $N_M$ as in \tref{thm:PTResultSM}. 
This completes the proof. 
\end{proof}

The following lemma can be derived by a similar induction argument used in the proof of \lref{lem:EstimatorTrans}. 

\begin{lemma}\label{lem:PTMestTrans}
Suppose $t\geq 2$ is a constant integer. For any $0<\eta<1$, if the estimates $\big\{\hat{\zeta}_k^{\PT}\big\}_{k=2}^t$ satisfy
\begin{align}
\left| \hat{\zeta}_k^{\PT} - \zeta_k^{\PT}\right|  \leq c\eta \quad \forall k = 2, \dots, t, 
\end{align} 
then the quantities $\big\{\hat{p}_k^{\PT}\big\}_{k=2}^t$ computed in Step~11 of Algorithm~\ref{alg:PTME} satisfy
$\big|\hat{p}_k^{\PT} - p_k^{\PT}\big| = \mathcal{O}(\eta)$ for all $k=2,\dots,t$.
\end{lemma}

\section{Classical computational cost in postprocessing}\label{sec:computational_cost}
In this section, we analyze the classical computational cost of our CBNE and PTME protocols. 

\subsection{Computational cost of CBNE for state moment estimation} 
Here we analyze the classical computational complexity of the CBNE protocol (in Algorithm~\ref{alg:CBNEmoments}) for estimating $p_t$, where $t\geq 2$ is a constant integer. 
The postprocessing stage involves computing the collision estimator
\begin{align}\label{eq:collision_moment_SM2}
\hat M_k^U=
\frac{\kappa_k}{d}\binom{N_M}{k}^{-1}
\sum_{1\leq i_1<\cdots<i_k\leq N_M}
\mathbf{1}\{\hat{b}_{i_1}=\cdots=\hat{b}_{i_k}\}. 
\end{align}
for $k=2,3,\dots,t$ across $N_U$ distinct random unitaries $U$, 
where $\mathbf{b}_U=\{\hat{b}_1,\dots,\hat{b}_{N_M}\}$ denotes the list of measurement outcomes with $\hat{b}_i\in\{0,\dots,d-1\}$. 

To evaluate this quantity efficiently, it is convenient to group the
indices $1\le i\le N_M$ according to the value of $\hat{b}_i$.
For $j\in\{0,\ldots,d-1\}$, define
$\theta_j = \bigl|\{i\,|\,\hat{b}_i=j\}\bigr|$ 
to be the number of occurrences of the value $j$ among the $N_M$ outcomes.
A $k$-tuple $(i_1,\ldots,i_k)$ contributes to the sum in \eref{eq:collision_moment_SM2} if and only if all its associated outcomes correspond to the same value $j$.
Therefore, the summation can be rewritten as
\begin{equation}
\sum_{1\le i_1<\cdots<i_k\le N_M}
\mathbf{1}\{\hat{b}_{i_1}=\cdots=\hat{b}_{i_k}\}
=
\sum_{j\,|\,\theta_j>0}\binom{\theta_j}{k}.
\end{equation}
To compute this expression, one first scans the list
$\{\hat{b}_i\}_{i=1}^{N_M}$ and records the frequencies $\theta_j$ of the observed values, which requires $\bigo{N_M}$ time.
Note that the number of nonzero $\theta_j$ is at most $\bigo{N_M}$.
The sum $\sum_j \binom{\theta_j}{k}$ can then be computed by iterating over these values.
Because $k$ is a constant, each binomial coefficient $\binom{\theta_j}{k}$ can be evaluated in $\bigo{1}$ time, yielding a total cost of $\bigo{N_M}$.
Finally, since the prefactor $\frac{\kappa_k}{d}\binom{N_M}{k}^{-1}$ in \eref{eq:collision_moment_SM2} can also be computed in $\bigo{1}$ time, the overall time cost for evaluating $\hat M_k^U$ is $\bigo{N_M}$.

Recall that $\hat M_k^U$ should be computed for $k=2,\dots,t$ across $N_U$ distinct $U$.
Therefore, the overall classical computational complexity for estimating $p_t$ via CBNE is
\begin{equation}
T_1=N_U\cdot\sum_{k=2}^{t}\bigo{N_M}
=\bigo{N_UN_M}
=\bigo{\max \left\{\frac{d^{1-1/t}}{\epsilon^{2/t}},\,\frac{1}{\epsilon^2}\right\}}, 
\end{equation}
where the second equality holds because $t$ is a constant, and the last equality follows from \eref{eq:SampleCompCBNE} in the main text. 
Note that this computational complexity has the same order as the sample complexity $N_U N_M$, and is therefore the best scaling one can expect.
We remark again that the information of the random unitaries $U$ is not used in the postprocessing.

\subsection{Computational cost of CBNE for estimating \texorpdfstring{$\Tr(O\rho^t)$}{}}

Next, we analyze the classical computational complexity of the CBNE protocol (in Algorithm~\ref{alg:CBNEmain} of the End Matter) for estimating $\Tr(O\rho^t)$, where $t\ge1$ is a constant integer.
The postprocessing stage involves computing both $\hat M_k^U$ in \eref{eq:collision_moment_SM2} and 
\begin{align}\label{eq:GammakUSM2}
\hat{\Gamma}_k^U(O_0)=
\frac{\kappa_{k+1}}{d}\binom{N_M}{k}^{-1}
\sum_{1\le i_1<\cdots<i_k\le N_M}
\mathbf{1}\{\hat{b}_{i_1}=\cdots=\hat{b}_{i_k}\}\,
\bra{r_{i_1}} U O_0 U^{\dag} \ket{r_{i_1}},
\end{align}
for $k=2,3,\dots,t$ across $N_U$ distinct random unitaries $U$, where  $O_0$ denotes the traceless part of the observable $O$.

To compute $\hat{\Gamma}_k^U(O_0)$ efficiently, as analyzed in the previous subsection, it is convenient to group the indices according to the value of $\hat{b}_i$.
Let $\theta_j=\bigl|\{i\,|\,\hat{b}_i=j\}\bigr|$ 
denote the number of occurrences of outcome $j\in\{0,\ldots,d-1\}$ among the list $\mathbf{b}_U=\{\hat{b}_1,\dots,\hat{b}_{N_M}\}$.
Then the summation in \eref{eq:GammakUSM2} can be rewritten as
\begin{equation}\label{eq:GammakUSM3}
\sum_{1\le i_1<\cdots<i_k\le N_M}
\mathbf{1}\{\hat{b}_{i_1}=\cdots=\hat{b}_{i_k}\}\,
\bra{r_{i_1}}U O_0 U^\dag\ket{r_{i_1}}
=
\sum_{j|\theta_j>0}
\binom{\theta_j}{k}
\bra{j}U O_0 U^\dag\ket{j}.
\end{equation}
The number of nonzero 
$\theta_j$ is at most $\bigo{N_M}$, which can be obtained by scanning the list
$\{\hat{b}_i\}_{i=1}^{N_M}$ with $\bigo{N_M}$ time.

For each $j$ with $\theta_j>0$, we additionally need to compute the quantity 
$\bra{j}U O_0 U^\dag\ket{j}$.
Assuming that the $d\times d$ matrix representation of $U$ is accessible. Let $e_j$ denote the $d$-dimensional computational-basis vector corresponding to $\ket{j}$.
Then
\begin{equation}
\bra{j}U O_0 U^\dag\ket{j}
=v_j^\dagger O_0 v_j,\qquad
v_j=U^\dagger e_j.
\end{equation}
Since $v_j$ is simply the $j$-th column of $U^\dagger$, obtaining this vector requires $\bigo{d}$ time.
Then the matrix-vector product $O_0v_j$ requires $\bigo{d^2}$ operations for a general dense matrix $O_0$, and the final inner product $v_j^\dagger (O_0 v_j)$ costs $\bigo{d}$ time.
Therefore, evaluating $\bra{j}U O_0 U^\dag\ket{j}$ requires $\bigo{d^2}$ time in the worst case.

If the observable $O_0$ has additional structure, the computational cost can be further reduced.
For example, if $O_0$ is sparse and has ${\rm nnz}(O_0)$ nonzero entries, then the  matrix-vector multiplication $v_j^\dagger O_0 v_j$  can be performed in $\bigo{{\rm nnz}(O_0)}$ time, which leads to an overall $\bigo{{\rm nnz}(O_0)+d}$ cost for computing $\bra{j}U O_0 U^\dag\ket{j}$. 
Alternatively, if $O_0$ has rank $r$, one can write $O_0=AB^\dagger$ with $A,B\in\mathbb{C}^{d\times r}$.
In this case, we have 
\begin{equation}
\bra{j}U O_0 U^\dag\ket{j}
=v_j^\dagger O_0 v_j = (A^\dagger v_j)^\dagger (B^\dagger v_j),
\qquad 
v_j=U^\dagger e_j, 
\end{equation}
which can be evaluated in $\bigo{rd}$ time.
In summary, the overall complexity for evaluating $\bra{j}U O_0 U^\dag\ket{j}$ is
\begin{equation}\label{eq:tildeTO}
\tilde T(O_0)= \bigo{ \min\!\big\{ d+{\rm nnz}(O_0), d\cdot {\rm rank}(O_0) \big\} }
\end{equation}

Because $k$ is a constant, each binomial coefficient $\binom{\theta_j}{k}$ in \eref{eq:GammakUSM3}  can be evaluated in $\bigo{1}$ time.
Thus, computing the sum
$\sum_{j:\theta_j>0}\binom{\theta_j}{k}\bra{j}U O_0 U^\dag\ket{j}$
costs $\bigo{N_M \tilde T(O_0)}$ time in total.
This leads to an $\bigo{N_M \tilde T(O_0)}$ overall time cost for evaluating $\hat{\Gamma}_k^U(O_0)$ in \eref{eq:GammakUSM2}.
Recall that we need to compute both $\hat M_k^U$ and $\hat{\Gamma}_k^U(O_0)$ for all $k=2,\dots,t$ and for $N_U$ distinct $U$.
Therefore, the overall classical computational complexity for estimating $\Tr(O\rho^t)$ via CBNE is
\begin{align}
T_2
&=T_1+N_U\cdot\sum_{k=2}^{t}\bigo{N_M \tilde T(O_0)}
=\bigo{N_UN_M \tilde T(O_0)}
\nonumber\\
&=\bigo{\max \left\{\frac{d^{1-1/t}\mathfrak{B}^{1/t}}{\epsilon^{2/t}},\,\frac{\mathfrak{B}}{\epsilon^2}\right\}
\cdot\min\!\Big\{ d+{\rm nnz}(O_0), d\cdot {\rm rank}(O_0) \Big\}}, 
\qquad 
\mathfrak{B}= \max \left\lbrace \Tr(O_0^2),1\right\rbrace, 
\end{align}
where the second equality holds because $t$ is a constant, and the third equality follows from \eref{eq:SampleCompCBNE} in the main text.

\subsection{Computational cost of PTME for PT moment estimation}

Here we analyze the classical computational complexity of the PTME protocol (in Algorithm~\ref{alg:PTME} in the End Matter) for estimating $p_t^{\PT}$, where $t\geq 2$ is a constant integer. 
The postprocessing stage involves computing the estimator
\begin{align}\label{eq:hatLambdaPTSM2}
\hat{\Lambda}_k^{U}
	:= \binom{N_M}{k}^{-1} \frac{d_{\!A}^{\,k}}{k!\, d_{\!A_1} }
	\sum_{1\leq i_1 <\dots < i_k \leq N_M} \left( \hat{r}_{i_1}\cdots \hat{r}_{i_k} \right) \times  \textbf{1}\big\{\hat{b}_{i_1}=\dots=\hat{b}_{i_k}\big\}, 
\end{align}
for $k=2,3,\dots,t$ across $N_U$ distinct random unitaries $U$, 
 where
$\mathbf{b}_{U} = \big\{\hat{b}_1, \dots, \hat{b}_{N_M}\big\}$ and $\mathbf{r}_{U} = \{\hat{r}_1, \dots, \hat{r}_{N_M}\}$ denote the measurement outcomes with $\hat r_i\in\{\pm1\}$ and $\hat b_i\in\{0,\dots,d_{A_1}-1\}$.

To evaluate this quantity efficiently, it is again convenient to group the
indices $1\le i\le N_M$ according to the value of $\hat b_i$.
For each $j\in\{0,\ldots,d_{A_1}-1\}$, define
$\theta_j = \bigl|\{i\,|\,\hat b_i=j\}\bigr|$ 
to be the number of occurrences of the outcome $j$ among the $N_M$ outcomes, and let $\mathcal S_j=\{i\,|\,\hat b_i=j\}$ 
denote the corresponding index set.
A $k$-tuple $(i_1,\ldots,i_k)$ contributes to the summation in
\eref{eq:hatLambdaPTSM2} only when all its associated outcomes  correspond to the same value
$j$.
Therefore, the summation can be rewritten as
\begin{equation}\label{eq:sumjPTMESM}
\sum_{1\le i_1<\cdots<i_k\le N_M}
(\hat r_{i_1}\cdots \hat r_{i_k})\times 
\mathbf 1\{\hat b_{i_1}=\cdots=\hat b_{i_k}\}
=
\sum_{j\,|\,\theta_j>0}
\sum_{\substack{i_1<\cdots<i_k\\ i_1,\dots,i_k\in \mathcal S_j}}
\hat r_{i_1}\cdots \hat r_{i_k}.
\end{equation}

To compute this expression, one first scans the lists
$\{\hat b_i\}_{i=1}^{N_M}$ and $\{\hat r_i\}_{i=1}^{N_M}$ and records the
values $\theta_j$ together with the corresponding values $\hat r_i$.
This step requires $\bigo{N_M}$ time.
Note that among the list $\{\hat b_i\}_{i=1}^{N_M}$ there can be at most $N_M$ distinct
values of $\hat b_i$, so the number of nonzero $\theta_j$ is at most $\bigo{N_M}$.
For each $j$ with $\theta_j>0$, we need to compute
\begin{equation}
\sum_{\substack{i_1<\cdots<i_k\\ i_1,\dots,i_k\in \mathcal S_j}}
\hat r_{i_1}\cdots \hat r_{i_k}.
\end{equation}
Because $\hat r_i=\pm1$ and $k$ is a constant, this quantity can be
computed in $\bigo{\theta_j}$ time.
Summing over all $j$ as in \eref{eq:sumjPTMESM} thus costs 
$\sum_j \bigo{\theta_j}=\bigo{N_M}$ time. 
Since the prefactor
$\binom{N_M}{k}^{-1}\frac{d_A^{\,k}}{k!\,d_{A_1}}$
in \eref{eq:hatLambdaPTSM2} can also be evaluated in $\bigo{1}$ time,
the overall time cost for computing $\hat{\Lambda}_k^U$ is $\bigo{N_M}$.

Recall that $\hat{\Lambda}_k^U$ needs to be computed for
$k=2,3,\dots,t$ across $N_U$ distinct $U$.
Thus, the total classical computational complexity for estimating
$p_t^{\PT}$ via PTME is
\begin{equation}
T_3=
N_U\cdot\sum_{k=2}^{t}\bigo{N_M}
=\bigo{N_U N_M}
=\bigo{ \max\left\lbrace d_{\!A} d_B \left( \frac{d_B}{d_{A}\epsilon^2} \right)^{1/t}, \frac{d_B^2}{\epsilon^2}  \right\rbrace },
\end{equation}
where the last equality follows from \tref{thm:PTResult} in the main text. 
Note that this computational complexity has the same order as the sample complexity $N_U N_M$, and is therefore the best scaling one can expect.
We remark again that the information of the random unitaries $U$ is not used in the postprocessing.

\section{Comparison with previous works}
\label{sec:comparison}

Numerous methods have been proposed in the literature for estimating nonlinear state properties such as $p_t$, $\Tr(O\rho^t)$, and $p_t^{\PT}$. These methods typically rely on collective operations performed over multiple copies of $\rho$~\cite{Quek2024multivariatetrace,liu2025estimating,Shin2025resourceefficient,shi2025near,subacsi2019entanglement,zhou2024hybrid,liu2024auxiliary,chen2025simultaneous,shi2025near,ye2025exponential,PhysRevLett.88.217901,bruni2004measurimg,zhang2025measuring,PhysRevLett.121.150503,PhysRevLett.94.040502,PhysRevA.99.062309,PhysRevLett.129.260501}, employing diverse techniques such as generalized swap tests~\cite{zhou2024hybrid,PhysRevLett.88.217901,bruni2004measurimg,PhysRevLett.94.040502,PhysRevA.99.062309}, block encoding~\cite{liu2025estimating,zhang2025measuring}, and deliberately constructed circuits~\cite{Quek2024multivariatetrace,liu2024auxiliary}, among others.
Beyond the sample-access model, several studies have further explored scenarios where a purification of the unknown state $\rho$ is available~\cite{liu2024exponential,zhang2025measuring,chen2025simultaneous}.

However, the primary focus of this work is estimation via \emph{single-copy} measurements. Thus, we refrain from further elaboration on these multi-copy approaches. Instead, in the remainder of this section, we survey existing works on nonlinear property estimation using single-copy operations, and compare them with our results.
The number of measurement settings required by these protocols are summarized in Table~\ref{tab:compare} of the End Matter.

\subsection{Sample-optimal estimation of \texorpdfstring{$p_2$}{}, \texorpdfstring{$\Tr(\sigma\rho)$}{}, and \texorpdfstring{$\Tr(O\rho^2)$}{}}

The estimation of the quantum state purity $p_2=\Tr(\rho^2)$ and the inner product $\Tr(\sigma\rho)$ between two unknown states with single-copy measurements has been studied in Refs.~\cite{PhysRevA.99.052323,brydges2019probing,PhysRevLett.124.010504,anshu2022distributed}. 
Their protocols rely on RMs, and have the following resource requirements for achieving an additive error $\epsilon$ \cite{anshu2022distributed}: 
\begin{align} \label{eq:NuNmt=2}
N_U = \bigo{\max\left\lbrace 1, \dfrac{1 }{d\epsilon^2} \right\rbrace}, 
\qquad 
N_M=  \bigo{\min\left\lbrace \frac{\sqrt{d}}{\epsilon}, d \right\rbrace}.
\end{align}
Here $d$ is the system dimension, $N_U$ is the number of global random unitaries used, and $N_M$ is the repetition number per unitary. 
This performance matches that of our CBNE protocol stated in \tref{thm:MomOResult} of the main text with $t=2$. 
The resulting sample complexity $N_U N_M=\mathcal O(\sqrt{d})$ achieves the optimal scaling with respect to $d$ among protocols based on single-copy operations \cite{anshu2022distributed,chen2022exponential,gong2024sample}.
Notably, the protocol in Ref.~\cite{anshu2022distributed} utilizes a collision estimator and can be viewed as a special case of our CBNE protocol with $t=2$ and $O=\openone$.

By generalizing the purity estimation protocol of Refs.~\cite{PhysRevA.99.052323,brydges2019probing}, a recent work \cite{du2025optimal} proposed an efficient protocol for estimating 
$\Tr(O\rho^2)$, termed the \emph{observable-driven RM} (ORM) protocol. 
When $O$ is a dichotomic observable (i.e., a Hermitian operator with eigenvalues $\pm1$), the resource costs of ORM coincide with those in \eref{eq:NuNmt=2}.  
For a general observable $O$ with $\|O\|\leq 1$, the number of measurement settings required by ORM becomes
\begin{align} \label{eq:NuNmt=2Liu}
N_U = \bigo{\max\left\lbrace 1, \dfrac{1 }{d\epsilon^2} \right\rbrace \log(\epsilon^{-1})}, 
\end{align}
and single-setting estimation is no longer achievable. 

ORM attains the optimal sample complexity with respect to $d$ when $\|O\|_1=\Theta(d)$ \cite{du2025optimal}. 
However, its experimental procedure depends on the target observable $O$ and therefore cannot be applied to simultaneously estimate many different observables, unlike CSE and CBNE. 
To address this issue, Ref.~\cite{du2025optimal} further proposed a \emph{braiding RM} (BRM) protocol, whose experimental stage is independent of the observable $O$. 
BRM requires $N_U=\bigo{1/\epsilon^2}$ random unitaries and $N_M=\bigo{\sqrt{d\Tr(O^2)}}$ measurements per unitary to estimate $\Tr(O\rho^2)$. 
Compared with this approach, our CBNE protocol achieves improved performance in both the setting number $N_U$ and sample cost $N_UN_M$, and  enables the estimation of more general higher-order functions $\Tr(O\rho^t)$.

\subsection{Estimation of \texorpdfstring{$p_t$}{} with the vEB protocol}

In Ref.~\cite{PhysRevLett.108.110503}, van Enk and Beenakker (vEB) proposed the first protocol for estimating state moments $p_t=\Tr(\rho^t)$ using single-copy measurements. 
The central idea of their protocol is as follows. 
Suppose one performs RMs on the state $\rho$ using a fixed random unitary $U$ sufficiently many times. 
From the measurement outcomes, one can obtain an accurate estimate of the outcome probabilities
$\Pr_\rho(b|U)=\bra{b}U\rho U^\dagger\ket{b}$ for $b\in\{0,\dots,d-1\}$, where $d$ is the system dimension. 
The protocol then averages these probabilities over many random unitaries to estimate the quantities
\begin{align}
\mathcal R_k(b):=\mathbb E_U[{\Pr}_\rho(b|U)^k], \qquad k=2,3\dots,t .
\end{align}
Using properties of Haar-random unitaries, one obtains
\begin{align}\label{eq:zeta_def_prlSM}
\kappa_k \, \mathcal R_k(b)
=
\kappa_k \, \underset{U\sim \mathcal E}{\E} \left[\<b|U\rho U^{\dag}|b\>^k\right]
= \kappa_k \Tr \left[ \rho^{\otimes k}  \Phi_{\mathcal E_{k}} \big( \ketbra{b}{b}^{\otimes k} \big) \right]
= \frac{1}{k!}\sum_{\pi\in \mathrm{S}_k}\Tr(\rho^{\otimes k}R_\pi) :=\zeta_k,
\end{align}
As discussed earlier, each $\zeta_k$ can be expressed as a polynomial of the moments $p_j$ with $j\le k$.
Thus, by estimating $\zeta_k$ for $k=1,2,\dots,t$ and inverting these polynomial relations, one can sequentially recover the desired moments $p_j$ \cite{PhysRevLett.108.110503}.

Reference~\cite{PhysRevLett.108.110503} did not analyze the sample complexity of the vEB protocol. 
In fact, a direct plug-in implementation of the vEB idea can be costly.
Indeed, since $\kappa_t=\binom{d+t-1}{t}=\mathcal O(d^t)$, Eq.~\eqref{eq:zeta_def_prlSM} implies that estimating $\zeta_t$ with constant additive error requires estimating $\mathcal R_t(b)$, and thus the probabilities $\Pr_\rho(b|U)$, to accuracy $\mu=\mathcal O(d^{-t})$. 
Achieving such precision requires at least on the order of $1/\mu^2=\mathcal O(d^{2t})$ measurement repetitions for each chosen random unitary, leading to a large measurement overhead.

Although our CBNE protocol also estimates state moments through the quantities $\zeta_k$, it is conceptually different from the vEB protocol. 
Instead of first estimating individual outcome probabilities $\Pr_\rho(b|U)$ and then computing their powers, our protocol directly estimates the power sums of these probabilities through collision statistics. 
This avoids the need for extremely precise probability estimation and leads to a substantially lower sample complexity. 
Furthermore, our protocol extends beyond moment estimation to more general nonlinear functions, including $\Tr(O\rho^t)$, principal component properties, and PT moments $p_t^{\PT}$, thereby enabling broader applications.

\subsection{Estimation of \texorpdfstring{$p_3^{\PT}$}{} with the ZZL protocol}
In Ref.~\cite{singlezhou}, Zhou, Zeng, and Liu (ZZL) proposed an efficient protocol for estimating the third-order PT moment 
$p_3^{\PT}=\Tr\left[(\rho_{AB}^{\top_{\! B}})^3\right]$ via RMs, with $\rho_{AB}$ being a quantum state on $\caH_A\otimes \caH_B$. 
The required number of random unitaries (measurement settings) and the total sample complexity are 
\begin{align}
N_U=\bigo{\frac{1}{\epsilon^2}}
\quad \text{and} \quad
N_UN_M=\bigo{ \frac{(d_Ad_B)^{2/3}}{\epsilon^2} }, 
\end{align}
respectively \cite{singlezhou}. 
Compared with this protocol, our PTME protocol offers advantages in several aspects. 
First, PTME  requires a significantly smaller setting number when $d_A\gg d_B$ (see \tref{thm:PTResult} in the main text). 
Second, the sample complexity for estimating $p_3^{\PT}$ with PTME reads 
\begin{align}
N_U N_M=  \mathcal O \left[\max\left\lbrace \frac{d_B^2}{\epsilon^2}, 
\left(\frac{d_A d_B^2}{\epsilon}\right)^{2/3} \right\rbrace \right], 
\end{align}
which is smaller than that of the ZZL protocol in the regime $d_B\leq \max\{\sqrt{d_A},1/\epsilon^2\}$. 
Third, PTME enables the estimation of higher-order PT moments, whereas the ZZL protocol is restricted to $p_3^{\PT}$. 
Fourth, PTME only needs local random unitaries on $\caH_A$ during the experimental stage, while the ZZL protocol requires global random unitaries on $\caH_A\otimes\caH_B$ \cite{singlezhou}, which can be experimentally challenging in some scenarios.

To address the last issue mentioned above, ZZL also proposed a variant protocol for estimating $p_3^{\PT}$ \cite{singlezhou}. 
This protocol applies local random unitaries on both $\caH_A$ and $\caH_B$, followed by Bell-basis measurements on each qubit pair shared between systems $A$ and $B$. 
However, this variant protocol only applies to quantum states with $d_A=d_B$, and its statistical performance remains unclear.

\subsection{Nonlinear property estimation with CSE}

Classical shadow estimation (CSE) is a recently developed framework for efficiently predicting properties of an unknown quantum state $\rho$ using RMs \cite{HKPshadow20}. 
In each experimental round of CSE, one first applies a unitary $U$ randomly drawn from a suitable ensemble $\mathcal E$ to the $d$-dimensional quantum state $\rho$, and then measures the rotated state $U\rho U^\dagger$ in the computational basis. 
The measurement outcome $\hat{b}$ together with the applied unitary $U$ constitutes the raw data of the experiment. 
From this data one can construct a classical estimator $\hat{\rho}$, referred to as a classical snapshot, which reproduces the original state $\rho$ in expectation: $\mathbb{E}[\hat{\rho}] = \rho$. 
By collecting $N$ independent snapshots $\{\hat{\rho}_1,\dots,\hat{\rho}_N\}$, one obtains a classical representation of the state that enables the estimation of properties of $\rho$ via classical postprocessing. 
For example, when $\mathcal E$ is a unitary 3-design (e.g., the Clifford group), the snapshot $\hat\rho$ has the form
\begin{align}\label{eq:hatrhoCSE}
\hat\rho
= (d+1) U^{\dag}|\hat b\>\<\hat b| U -\openone , 
\end{align} 
and the number of samples for estimating a linear property $\Tr(O\rho)$ within error $\epsilon$ reads $N=\bigo{\Tr(O_0^2)/\epsilon^2}$ \cite{HKPshadow20}.

A key feature of CSE, also shared by CBNE, is that the experimental stage is independent of the target property one wants to estimate. 
This independence enables simultaneous prediction of numerous distinct properties \cite{elben2023randomized,HKPshadow20}. 
Consequently, classical shadows provide an efficient approach for characterizing quantum systems and have found numerous applications in quantum learning and quantum many-body physics.

Beyond linear functions, 
CSE can naturally be used to predict nonlinear functions of quantum states 
\cite{elben2020mixedstate,PhysRevLett.127.260501,HKPshadow20,elben2023randomized,PRXQuantum.4.010303,PhysRevX.14.031035}. 
For instance, the moment $p_t=\Tr(\rho^t)$ can be estimated using $t$ distinct snapshots 
$\hat\rho_{i_1},\dots,\hat\rho_{i_t}$ obtained in different experimental rounds. 
Due to their independence, $\Tr(\hat\rho_{i_1}\hat\rho_{i_2}\cdots\hat\rho_{i_t})$ provides an unbiased estimate of $p_t$. 
To reduce the estimation error, one may use $N$ independent snapshots and average over all possible index tuples \cite{pelecanos2026beating,elben2020mixedstate}:
\begin{align}\label{eq:CSEptEstimator}
\widehat{p_t} = \frac{1}{N(N-1)\dots(N-t+1)} \sum_{\text{distinct} \ i_1,i_2,\dots,i_t \in\{1,\dots,N\}} 
\Tr(\hat\rho_{i_1}\hat\rho_{i_2}\cdots\hat\rho_{i_t}). 
\end{align}
Similarly, unbiased estimators for more general quantities such as $\Tr(O\rho^t)$ and the PT moments $p_t^{\PT}$ can be constructed in the same spirit \cite{PRXQuantum.4.010303,elben2020mixedstate}:
\begin{align}
\widehat{\Tr(O\rho^t)} 
&= \frac{1}{N(N-1)\dots(N-t+1)} \sum_{\text{distinct} \ i_1,i_2,\dots,i_t \in\{1,\dots,N\}}  \Tr(O \hat\rho_{i_1} \hat\rho_{i_2}\cdots \hat\rho_{i_t}),
\label{eq:CSEptOEstimator}
\\
\widehat{p_t^{\PT}} 
&= \frac{1}{N(N-1)\dots(N-t+1)} \sum_{\text{distinct} \ i_1,i_2,\dots,i_t \in\{1,\dots,N\}} 
\Tr\left[ \left(\hat\rho_{i_1}\right)^{\top_{\! B}} \left(\hat\rho_{i_2}\right)^{\top_{\! B}} \cdots \left(\hat\rho_{i_t}\right)^{\top_{\! B}}\right].  
\label{eq:CSEPTMEstimator}
\end{align}
However, for general $t\ge2$, the statistical performance of these estimators is not yet fully understood. 
A detailed analysis is currently available only for the estimator $\widehat{p_t}$ in Eq.~\eqref{eq:CSEptEstimator}, as recently studied in Ref.~\cite{pelecanos2026beating}. 
Below we briefly review their result and derive the corresponding sample complexity under the additive error criterion.

For constant integer $t$, Ref.~\cite[Theorem~5.2]{pelecanos2026beating} shows that $\widehat{p_t}$ has variance
\begin{align}
\Var \left[\widehat{p_t} \right]
= \mathcal O \left[\frac{1}{d} \sum_{j=0}^{t-1} \left( \frac{d}{N}\right)^{t-j} \Tr(\rho^{2j})\right].
\end{align}
By Chebyshev's inequality, 
$\Pr\left\{ \left|\widehat{p_t}-p_t\right| > \epsilon \right\}
\leq \Var\left[\widehat{p_t} \right]/\epsilon^2$. 
Hence, to achieve the target accuracy $\epsilon$ with high probability, one needs to choose
\begin{align}\label{eq:Nshadowpt}
N=  \bigo{\max\left\lbrace \frac{d}{\epsilon^{2/t}}, \frac{1}{\epsilon^2} \right\rbrace}
\end{align}
for a general quantum state $\rho$. 
Here $N$ characterizes both the number of measurement settings and the total number of samples consumed. 
To our knowledge, this represents the best known scaling for estimating higher-order moments of a general quantum state using single-copy measurements prior to our work.
In comparison, the sample complexity of our CBNE protocol for estimating $p_t$ is (see \eref{eq:SampleCompCBNE} in the main text): 
\begin{align}
N_U N_M
=
\bigo{\max\left\lbrace
\frac{d^{1-1/t}}{\epsilon^{2/t}},
\frac{1}{\epsilon^2}
\right\rbrace }.
\end{align}
Thus improves upon the scaling in \eref{eq:Nshadowpt} by a factor of $d^{1/t}$, implying that CBNE has a strictly better sample efficiency than CSE on nonlinear state moment estimation. 

Moreover, CBNE uses substantially fewer measurement settings than direct CSE-based nonlinear estimators. 
For moment estimation, the number of settings required by CBNE is only
$N_U =  \max\left\lbrace 1, c/(d\epsilon^2)\right\rbrace$
for some constant \(c\), which is much smaller than the setting number in Eq.~\eqref{eq:Nshadowpt}. 
This reduction is practically important: on many quantum platforms, repeated measurements under a fixed setting are relatively inexpensive, whereas changing the measurement basis can be time-consuming.
A similar distinction appears for estimating \(\Tr(O\rho^t)\) and \(p_t^{\PT}\), where the number of required measurement settings also scales with the total sample cost $N$, which is exponentially large with the system size.   
By contrast, CBNE and PTME have a much smaller setting complexity; see Table~\ref{tab:compare}.

Another limitation of direct CSE-based nonlinear estimation is its postprocessing cost. 
Evaluating the estimators in Eqs.~\eqref{eq:CSEptEstimator}, \eqref{eq:CSEptOEstimator}, and \eqref{eq:CSEPTMEstimator} requires constructing \(N\) snapshots \(\hat\rho\) and combining them over distinct \(t\)-tuples of indices. 
A straightforward implementation therefore involves \(\bigo{N^t}\) tuple products and the corresponding matrix multiplications, which can become prohibitive when \(N\) is large.
By contrast, CBNE and PTME process the data directly through collision statistics. 
As shown in Sec.~\ref{sec:computational_cost}, the postprocessing cost for estimating \(p_t\) and \(p_t^{\PT}\) scales only linearly in \(N\), 
and the applied random unitaries \(U\) need not be used in postprocessing, unlike CSE. 
For estimating \(\Tr(O\rho^t)\), the CBNE postprocessing cost is \(\bigo{N\,\tilde T(O_0)}\), where \(\tilde T(O_0)\) is defined in Eq.~\eqref{eq:tildeTO}. 
Thus, for the nonlinear quantities considered here, the collision-based estimators of our protocols substantially reduce the postprocessing overhead compared with 
direct CSE snapshot-product estimators.

\subsection{Single-setting CSE via ancillary systems}

Two recent works~\cite{PhysRevX.13.011049,PhysRevLett.131.160601} propose a variant of the CSE approach that enables property estimation using a single measurement setting. 
In each experimental round of their protocol, the $n$-qubit state $\rho$ of interest is first entangled with a sufficiently large ancillary system. Next, instead of applying randomly sampled unitaries, a fixed global unitary $U$ or a fixed Hamiltonian time evolution is applied to the joint system consisting of $\rho$ and the ancillas, followed by a measurement in the computational basis. 
From the resulting measurement outcomes, one can construct classical shadow snapshots $\hat{\rho}$ that remain unbiased estimators of $\rho$. 
These snapshots can then be used in classical postprocessing to estimate various functions of $\rho$. 
For linear properties, Ref.~\cite{PhysRevLett.131.160601} shows that this variant protocol achieves a sample efficiency comparable to the standard CSE protocol, while its performance for nonlinear properties remains unclear.

The underlying principle of this approach is that, when the ancillary system is sufficiently large, the overall experimental procedure (adding ancillas, unitary evolution, and computational-basis measurement) effectively implements a \emph{tomographically complete} POVM on the original state $\rho$. 
Consequently, the full information of $\rho$ can in principle be reconstructed from the measurement data.

The notable feature of this variant CSE approach, namely, property estimation from a single measurement setting, is similar to our CBNE and PTME protocols. 
However, this advantage comes at the cost of requiring a large ancillary system, whose qubit number scales as (see Ref.~\cite[Theorem~1]{PhysRevLett.131.160601}): 
\begin{align}
n_a = \bigo{n+\log(\epsilon^{-1})}. 
\end{align}
Such a requirement can be prohibitive for near-term quantum devices when the system size $n$ is large.
In contrast, our CBNE and PTME protocols require far fewer ancillary qubits. 
As shown in the main text, in common scenarios our protocols achieve a single measurement setting without introducing any ancillary qubits. 
Even in the worst case, the number of ancillary qubits needed is bounded by
\begin{align}
n_a \le 2\log_2(\epsilon^{-1}) + \bigo{1},
\end{align}
which is completely independent of the system size $n$ of the original system. 
This represents a substantial reduction in ancillary-system overhead compared with Refs.~\cite{PhysRevX.13.011049,PhysRevLett.131.160601}, particularly for large-scale quantum systems.

\section{Estimation of other nonlinear functions with the CBNE framework}\label{sec:OtherNonlinear}

In this section, we discuss how our CBNE framework can be extended beyond $\Tr(O\rho^t)$, $\Tr(\rho^t)$, and $\Tr\left[(\rho_{AB}^{\top_{\! B}})^t\right]$. 
The main purpose here is to illustrate that the central idea of CBNE, namely, extracting information from collision statistics under single-copy randomized measurements, applies to a broader class of nonlinear quantities as well. 
We present two representative examples, 
showing how our framework can be adapted to estimate the expectation values of two-body observables and general multivariate functions.
 
\subsection{Estimation of two-body observables}
We first consider the estimation of $\Tr(O \rho^{\otimes 2})$, where $\rho$ is an unknown quantum state on $\caH$ and $O$ is a two-body observable on $\caH^{\otimes 2}$. 

Suppose that copies of $\rho$ are prepared sequentially. 
In each experimental round, we apply a Haar-random unitary $U$ to $\rho$, followed by a computational-basis measurement. 
Repeating this procedure $N_M$ times with the same unitary $U$ produces a list of outcomes
$\mathbf{b}_U=\{\hat b_1,\hat b_2,\dots,\hat b_{N_M}\}$. 
From these data, we define the following estimator:
\begin{align}\label{eq:defDelta}
\hat{\Delta}^{U}:=    
\frac{ 24\kappa_{4}}{d}  \binom{N_M}{2}^{-1}
\sum_{ 1\leq i_1 < i_2 \leq N_M} 
\textbf{1}\left\{\hat{b}_{i_1}=\hat{b}_{i_2}\right\} \cdot \Tr \left[ \left(U^{\dag}|\hat{b}_{i_1}\>\<\hat{b}_{i_1}|U\right)^{\otimes 2}O\right].    
\end{align}
This estimator is again constructed from collision events, in the same spirit as the basic CBNE protocol.
As proved below, the expectation value of $\hat{\Delta}^{U}$ is given by
\begin{align}\label{eq:expDelta}
\E \left[ \hat{\Delta}^{U} \right] 
&= 2 \Tr(O \rho^{\otimes 2}) + 2 \Tr(\mathbb{S}O \rho^{\otimes 2})+ [1+\Tr(\rho^2)] [\Tr(O) + \Tr(\mathbb{S}O)] + 2 \Tr(O_1\rho)+ 2 \Tr(O_2\rho)
\nonumber \\
&\quad\ 
+ 2 \Tr[(\mathbb{S}O)_1\rho]
+ 2 \Tr[(\mathbb{S}O)_2\rho] 
+ 2 \Tr(O_1\rho^2)+ 2 \Tr(O_2\rho^2) + 2 \Tr[(\mathbb{S}O)_1\rho^2] 
+ 2 \Tr[(\mathbb{S}O)_2\rho^2] 
\end{align}
where $\mathbb{S}$ denotes the swap operator on $\caH^{\otimes 2}$; 
$O_1:=\Tr_2(O)$ and $O_2:=\Tr_1(O)$
denote the reduced operators on the first and second subsystem, respectively; the same notation is used for $(\mathbb{S}O)_1$ and $(\mathbb{S}O)_2$.

Equation \eqref{eq:expDelta} shows that $\E \left[ \hat{\Delta}^{U} \right]$ is a linear combination of the target quantity $\Tr(O\rho^{\otimes 2})$, its swapped counterpart $\Tr(\mathbb{S}O\rho^{\otimes 2})$, and several lower-order quantities involving $\rho$ and $\rho^2$. 
Importantly, all of the latter terms, namely
\begin{align}
\Tr(\rho^2),\ 
\Tr(O_1\rho),\ \Tr(O_2\rho),\ 
\Tr(O_1\rho^2),\ \Tr(O_2\rho^2),\ 
\Tr((\mathbb{S}O)_1\rho),\ \Tr((\mathbb{S}O)_2\rho),\ 
\Tr((\mathbb{S}O)_1\rho^2),\ \Tr((\mathbb{S}O)_2\rho^2),
\end{align}
can be estimated using the standard CBNE protocol (see Theorem~\ref{thm:MomOResult} in the main text). 
Therefore, combining $\hat{\Delta}^{U}$ with these CBNE estimators yields an estimation of
\begin{align}\label{eq:sym_target}
\frac{1}{2}\Big[\Tr(O\rho^{\otimes 2})+\Tr(\mathbb{S}O\rho^{\otimes 2})\Big].
\end{align}

In particular, if $\rho$ is a pure state, then $\rho^{\otimes 2}\mathbb{S}= \rho^{\otimes 2}$, and \eqref{eq:sym_target} reduces exactly to the desired quantity $\Tr(O\rho^{\otimes 2})$. 
Moreover, if $O$ is an observable supported on the symmetric subspace of $\caH^{\otimes 2}$, then
$\mathbb{S}O=O$, 
and \eqref{eq:sym_target} also reduces to $\Tr(O\rho^{\otimes 2})$. Many important observables satisfy this requirement, including 
the projector $\Pi_{\rm sym}^{(2)}$ onto the symmetric subspace, arbitrary spectral observables of the form
$O=\sum_j \lambda_j \ketbra{\phi_j}{\phi_j}$, 
where each $\ket{\phi_j}$ lies in the symmetric subspace.

In conclusion, the nonlinear function $\Tr(O\rho^{\otimes 2})$ can also be estimated within the CBNE framework if $\rho$ is known to be pure or $O$ is supported on the symmetric subspace.  
In these cases, we expect the same concentration mechanism as in the standard CBNE protocol to yield single-setting estimation in the large-dimensional or few-ancilla regimes; a rigorous analysis is left for future work.
Moreover, since the experimental procedure is independent of the observable $O$, the same measurement data can be reused to estimate many different two-body observables simultaneously. 
\begin{proof}[Proof of \eref{eq:expDelta}]
The expectation value of $\hat{\Delta}^{U}$ (with respect to the randomness in unitary selection and measurement outcomes) reads: 
\begin{align}
\E \left[ \hat{\Delta}^{U}\right] 
&=\underset{U }{\E} \, \underset{\,\mathbf{b}_{U}}{\E} \left[ \hat{\Delta}^{U} \Big| U \right]
 =\underset{U }{\E} \left\{ \frac{24\kappa_{4}}{d} \sum_{b=0}^{d-1} 
            \Pr\left(\hat{b}_{i_1}=\hat{b}_{i_2}=b \Big|U \right) \cdot \Tr \left[ \left(U^{\dag}|b\>\<b|U\right)^{\otimes 2}O\right]\right\} 
\nonumber \\
&=\frac{24\kappa_{4}}{d}\, \sum_{b=0}^{d-1} \, \underset{U }{\E} \left\{ \<b|U\rho U^{\dag}|b\>^2 \cdot \Tr \left[ \left(U^{\dag}|b\>\<b|U\right)^{\otimes 2}O\right] \right\} 
=
\frac{24\kappa_{4}}{d} \,\sum_{b=0}^{d-1} \,\underset{U }{\E} \left\{ \Tr \left[ \left( \rho^{\otimes 2}\otimes O \right) \left(U^{\dag}|b\>\<b|U\right)^{\otimes 4}\right] \right\}
\nonumber \\
&=
\frac{24\kappa_{4}}{d} \,\sum_{b=0}^{d-1} \Tr \left[ \left( \rho^{\otimes 2}\otimes O \right)  
\Phi_{\rm H} \big( \ketbra{b}{b}^{\otimes {4}} \big) \right]
\stackrel{(a)}{=}
\sum_{\pi\in \mathrm{S}_{4}}  \Tr\left[ \left(\rho^{\otimes 2} \otimes O \right)R_\pi\right]
\nonumber \\
&= 2 \Tr(O \rho^{\otimes 2}) + 2 \Tr(\mathbb{S}O \rho^{\otimes 2})+ [1+\Tr(\rho^2)] [\Tr(O) + \Tr(\mathbb{S}O)] + 2 \Tr(O_1\rho)+ 2 \Tr(O_2\rho)
\nonumber \\
&\quad\ 
+ 2 \Tr[(\mathbb{S}O)_1\rho]
+ 2 \Tr[(\mathbb{S}O)_2\rho] 
+ 2 \Tr(O_1\rho^2)+ 2 \Tr(O_2\rho^2) + 2 \Tr[(\mathbb{S}O)_1\rho^2] 
+ 2 \Tr[(\mathbb{S}O)_2\rho^2] , 
\end{align}
where $(a)$ follows from the Haar-twirling identity in \eref{eq:PhiHaar}.
\end{proof}

\subsection{Estimation of multivariate functions}

We next discuss the estimation of multivariate nonlinear functions of the form
\begin{align}\label{eq:multivariate_target}
\Tr(O_1\rho_1^{t_1} O_2\rho_2^{t_2}\cdots O_m\rho_m^{t_m}),
\end{align}
where $t_1,t_2,\dots,t_m$ are positive integers, $\rho_1,\dots,\rho_m$ are possibly distinct quantum states on $\caH$, and $O_1,\dots,O_m$ are observables on  $\caH$. 
Such quantities arise in tasks such as bounding the quantum Fisher information \cite{PRXQuantum.5.030338,PhysRevLett.127.260501,PhysRevLett.134.110802}, cross-platform verification \cite{anshu2022distributed,PhysRevLett.124.010504}, multivariate trace estimation \cite{Quek2024multivariatetrace}, and many other quantum information-processing applications.

To estimate \eqref{eq:multivariate_target}, it suffices to construct  an operator-valued estimator of $\rho_j^{t_j}$ for each $j=1,\dots,m$. 
This can again be achieved within the CBNE framework. 
Suppose copies of a fixed state $\rho$ are prepared sequentially. 
In each experimental round, we apply a Haar-random unitary $U$ to $\rho$, followed by a computational-basis measurement. 
Repeating this procedure $N_M$ times with the same $U$ yields a list of measurement outcomes
$\mathbf{b}_U=\{\hat b_1,\dots,\hat b_{N_M}\}$.
From these data, we define the estimator
\begin{align}\label{eq:defUpsilon}
\hat{\Upsilon}_k^{U} (\rho) :=    
\frac{ \kappa_{k+1}}{d}  \binom{N_M}{k}^{-1}
\sum_{ 1\leq i_1 < \dots < i_k \leq N_M} 
\textbf{1}\left\{\hat{b}_{i_1}=\dots=\hat{b}_{i_k}\right\} \cdot U^{\dag}|\hat{b}_{i_1}\>\<\hat{b}_{i_1}|U .  
\end{align}
As proved below, its expectation value is
\begin{align}\label{eq:expUpsilon}
\E \left[ \hat{\Upsilon}_k^{U} (\rho) \right] 
=\gamma_k(\rho):=
\frac{1 }{(k+1)! } \sum_{\pi\in \mathrm{S}_{k+1}}  \Tr_{1,\dots,k}\left[ \left(\rho^{\otimes k} \otimes \openone \right)R_\pi\right], 
\end{align}
where $\Tr_{1,\dots,k}$ denotes the partial trace over the first $k$ systems.

Note that the operator $\gamma_k(\rho)$ can be expanded as a polynomial in the moments $p_i=\Tr(\rho^i)$
and the state powers $\rho^j$, with $2\le i\le k$ and $0\le j\le k$. 
For the first few values of $k$, one finds
\begin{align}\label{eq:gammaExpan}
\gamma_1(\rho) &= \frac{\openone}{2}+\frac{\rho}{2}, 
\nonumber \\   
\gamma_2(\rho) &= \frac{(1+p_2)\openone}{6}+\frac{\rho}{3}+\frac{\rho^2}{3},
\nonumber \\ 
\gamma_3(\rho) &= \frac{(1+3p_2+2p_3) \openone}{24} + \frac{(1+p_2) \rho}{8} + \frac{\rho^2}{4} +\frac{\rho^3}{4}, 
\quad   \dots   
\end{align}
Therefore, combining the estimators $\hat\Upsilon_i^U(\rho)$ with the CBNE estimators for the moments $p_2,\dots,p_k$, one can recursively reconstruct estimators of $\rho,\rho^2,\dots,\rho^k$. 
For instance,
\begin{align}
\rho &= 2\gamma_1(\rho)-\openone, 
\nonumber \\
\rho^2 &= 3\gamma_2(\rho)-2\gamma_1(\rho)+\frac{1-p_2}{2}\openone,
\nonumber \\
\rho^3
&=
4\gamma_3(\rho)
-3\gamma_2(\rho)
+(1-p_2)\gamma_1(\rho)
+\frac{3p_2-1-2p_3}{6}\openone, 
\end{align}
and similarly for higher powers. 
We denote by $\widehat{\rho^k}$ an estimator of $\rho^k$ constructed in this way.

Applying this construction independently to each state $\rho_j$, we obtain estimators
$\widehat{\rho_1^{t_1}},\ \widehat{\rho_2^{t_2}},\ \dots,\ \widehat{\rho_m^{t_m}}$.
Then
\begin{align}
\Tr \left(
O_1\widehat{\rho_1^{t_1}}
O_2\widehat{\rho_2^{t_2}}
\cdots
O_m\widehat{\rho_m^{t_m}}
\right)
\end{align}
can be used to approximate the target quantity \eqref{eq:multivariate_target}, provided that the random data used to construct the estimators $\widehat{\rho_j^{t_j}}$ are independent across different $j$.

We remark that this independence requirement means that the random unitaries used for different states $\rho_j$ should be chosen independently. 
Therefore, unlike the single-state quantities considered before, one should not expect a single measurement setting to suffice in general when $m\ge 2$. 
Nevertheless, the total number of required settings may still remain very small: we expect that $m$ independent settings, one for each state $\rho_j$, should be sufficient in the large-dimensional regime or with the aid of a few ancillary qubits.
We leave a rigorous performance analysis for future work. 
Finally, we emphasize again that, since the experimental protocol for each $\rho_j$ is independent of $O_1,\dots,O_m$, the same measurement data can be reused to estimate many different multivariate functions simultaneously.

\begin{proof}[Proof of \eref{eq:expUpsilon}]
The expectation of $\hat{\Upsilon}_k^{U} (\rho)$ (with respect to the randomness in unitary selection and measurement outcomes) reads: 
\begin{align}
\E \left[ \hat{\Upsilon}_k^{U} (\rho) \right] 
&=\underset{U }{\E} \, \underset{\,\mathbf{b}_{U}}{\E} \left[ \hat{\Upsilon}_k^{U}(\rho)  \Big| U \right]
 =\underset{U }{\E} \left[ \frac{\kappa_{k+1}}{d} \sum_{b=0}^{d-1} 
            \Pr\left(\hat{b}_{i_1}=\dots=\hat{b}_{i_k}=b \,|U \right) \cdot U^{\dag}|b\>\<b|U\right] 
\nonumber \\
&=\frac{\kappa_{k+1}}{d}\, \sum_{b=0}^{d-1} \, \underset{U }{\E} \left[ \<b|U\rho U^{\dag}|b\>^k \cdot U^{\dag}|b\>\<b|U \right] 
\stackrel{(a)}{=} 
\frac{\kappa_{k+1}}{d} \sum_{b=0}^{d-1} \Tr_{1,\dots,k} \left[ \left( \rho^{\otimes k}\otimes \openone \right)  
\Phi_{\rm H} \big( \ketbra{b}{b}^{\otimes {(k+1)}} \big) \right]
\nonumber \\
&\,\stackrel{(b)}{=}
\frac{1 }{(k+1)! } \sum_{\pi\in \mathrm{S}_{k+1}}  \Tr_{1,\dots,k}\left[ \left(\rho^{\otimes k} \otimes \openone \right)R_\pi\right]
=\gamma_k(\rho) , 
\end{align}
where $(a)$ holds because $\<b|U\rho U^{\dag}|b\>^k \cdot U^{\dag}|b\>\<b|U =\Tr_{1,\dots,k} \left[ \left( \rho^{\otimes k}\otimes \openone \right) 
\big( U^{\dag}\ketbra{b}{b} U\big)^{\otimes (k+1)}\right]$, and $(b)$ follows from \eref{eq:PhiHaar}.
\end{proof}

\section{More numerical results}\label{sec:moreNumerical}

In this section, we provide additional numerical results supporting the performance of CBNE and PTME. 
Throughout this paper, the random unitaries used in the simulations are generated by finite-depth two-local brickwork circuits rather than exact global unitary designs. 
Specifically, a brickwork circuit consists of alternating layers of nearest-neighbor two-qubit gates: odd layers act on qubit pairs \((1,2),(3,4),\ldots\), while even layers act on pairs \((2,3),(4,5),\ldots\). 
Each two-qubit gate is independently sampled from the Haar measure. 
Such local random circuits are known to form approximate unitary designs
\cite{Brandao2016LocalDesign,Haferkamp2022Interleaved}. 
Thus, our numerical results also serve as a practical check that the predicted performance of CBNE and PTME persists when exact unitary designs are replaced by experimentally more feasible approximate designs, consistent with the theoretical guarantees in the End Matter.

\subsection{Sample-complexity scaling with system size}

In this subsection, we provide additional numerical evidence for the sample-complexity scalings predicted by Theorems~\ref{thm:MomOResult} and~\ref{thm:PTResult} in the main text.
We focus on third-order quantities, since \(t=3\) already captures the nonlinear scaling behavior of our estimators.

We first consider the estimation of the state moment \(p_3=\Tr(\rho^3)\) using the CBNE protocol.
According to \tref{thm:MomOResult}, the total number of measurements required to estimate \(p_3\) up to a constant additive error \(\epsilon\) scales as
\begin{equation}
N_{\mathrm{tot}}^{\mathrm{CBNE}}
= N_U N_M
= \bigo{d^{2/3}}
= \bigo{2^{\frac{2}{3}n}},
\end{equation}
where \(n\) is the number of qubits in \(\rho\) and \(d=2^n\) is the system dimension.

To demonstrate this prediction, we use the same noisy TFIM ground states as in the main text,
\(\rho=(1-p)|\mathrm{gs}\rangle\langle\mathrm{gs}|+p\mathbb{I}/d\),
with \(J=h=1\) and \(p=0.2\).
For each system size \(n\), we fix the number of random unitaries to \(N_U=200\), and then determine the smallest number of measurements per unitary \(N_M\) such that the statistical error falls below the threshold \(0.1\).
Here \(N_U\) is chosen to be larger than one only to suppress statistical fluctuations and make the scaling with \(n\) more clearly visible in the numerics.
This choice is made only for numerical stability and visualization; it does not contradict our theoretical result that, when the system dimension is sufficiently large, single-setting estimation with $N_U=1$ is already sufficient. 
Our numerical results, shown in Fig.~\ref{fig:Samplevsn}(a), indicate that the fitted scaling of the total sample complexity $N_U N_M$ is consistent with the theoretical prediction
\(N_{\mathrm{tot}}^{\mathrm{CBNE}}\propto 2^{2n/3}\).

\begin{figure}[b]
\centering
\includegraphics[width=0.7\linewidth]{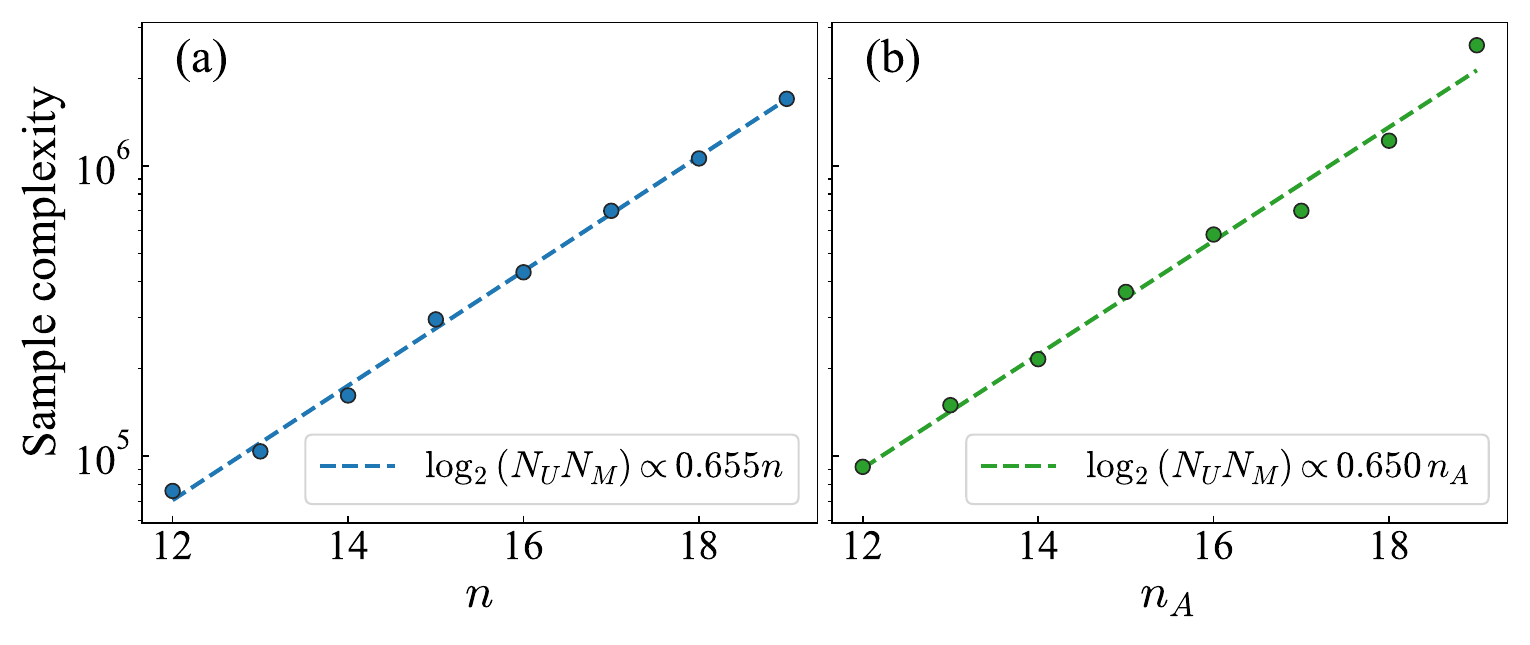}
\caption{Scaling of the minimal number of measurements $N_UN_M$ required to estimate (a) the state moment $\Tr(\rho^3)$ and (b) the PT moment $\Tr\!\big[(\rho_{AB}^{\top_{\! B}})^3\big]$ up to a fixed statistical error of 0.1, for fixed $N_U=200$. 
Here $\rho=(1-p)|\mathrm{gs}\>\<\mathrm{gs}|+pI/d$, where $|\mathrm{gs}\>$ is the ground state of $n$-qubit TFIM with $J=h=1.0$ and depolarizing probability $p=0.2$. 
For each data point, we choose the minimal $N_M$ such that the statistical error $\le0.1$ averaged over $50$ experiments. 
In panel (b), the total system has \(n=n_A+1\) qubits; subsystem \(A\) consists of the first \(n_A\) qubits, and subsystem \(B\) is the last qubit.}
\label{fig:Samplevsn}
\end{figure}

We next examine the sample complexity of our PTME protocol. 
For estimating the PT moment \(p_3^{\PT}=\Tr\!\big[(\rho_{AB}^{\top_{\! B}})^3\big]\), \tref{thm:PTResult} in the main text gives the required total number of measurements
\begin{equation}
N_{\mathrm{tot}}^{\mathrm{PTME}}
=
N_U N_M
=
\mathcal O\left[
\max\left\{
\frac{d_B^2}{\epsilon^2},
\left(\frac{d_A d_B^2}{\epsilon}\right)^{2/3}
\right\}
\right].
\label{eq:PTME_sample_complexity_num}
\end{equation}
When the subsystem \(B\) and the target accuracy \(\epsilon\) are fixed, this becomes
\begin{equation}\label{eq:PT3order}
N_{\mathrm{tot}}^{\mathrm{PTME}}
=
\bigo{d_A^{2/3}}=\bigo{2^{\frac23 n_A}}, 
\end{equation}
where $n_A$ and $n_B$ denote the qubit number of subsystems $A$ and $B$, respectively. 
Figure~\ref{fig:Samplevsn}(b) verifies this theoretical prediction numerically, in which we fix \(N_U=200\), $n_B=1$, and choose, for each \(n_A\), the smallest number \(N_M\) required to reach the same target statistical error 0.1. 
The fitted exponent is close to the expected value \(2/3\), supporting the sample-complexity scaling in \eref{eq:PT3order}. 

Together, Fig.~\ref{fig:Samplevsn}(a,b) confirms that both CBNE and PTME exhibit the predicted sample-complexity scaling that is sublinear in the system dimension for estimating nonlinear moment quantities, which is consistent with our theoretical results.

\subsection{Exponential error suppression in CBNE and PTME}

In this subsection, we provide additional numerical evidence for the dimension-dependent error suppression predicted by our theory. 
Recall that, for CBNE, Theorem~\ref{thm:MomOResult} shows that the number of measurement settings required to estimate \(\Tr(O\rho^t)\) is
$N_U=\max\left\{1,\,c\mathfrak B/(d\epsilon^2) \right\}$,
where $c$ is a constant factor, \(d\) is the system dimension, and \(\mathfrak B=\max\{\Tr(O_0^2),1\}\). 
Thus, increasing the dimension \(d\) reduces the required setting number and eventually drives the protocol into the single-setting regime. 
Equivalently, when \(N_U=1\) is fixed, the statistical estimation error is expected to decrease as the effective dimension increases. 
This is the mechanism behind the ancilla-assisted version of CBNE: appending \(n_a\) ancillary qubits increases the effective dimension from \(d\) to \(2^{n_a}d\), while preserving the target value \(\Tr(O\rho^t)\).

\begin{figure}[b]
\centering
\includegraphics[width=0.9\textwidth]{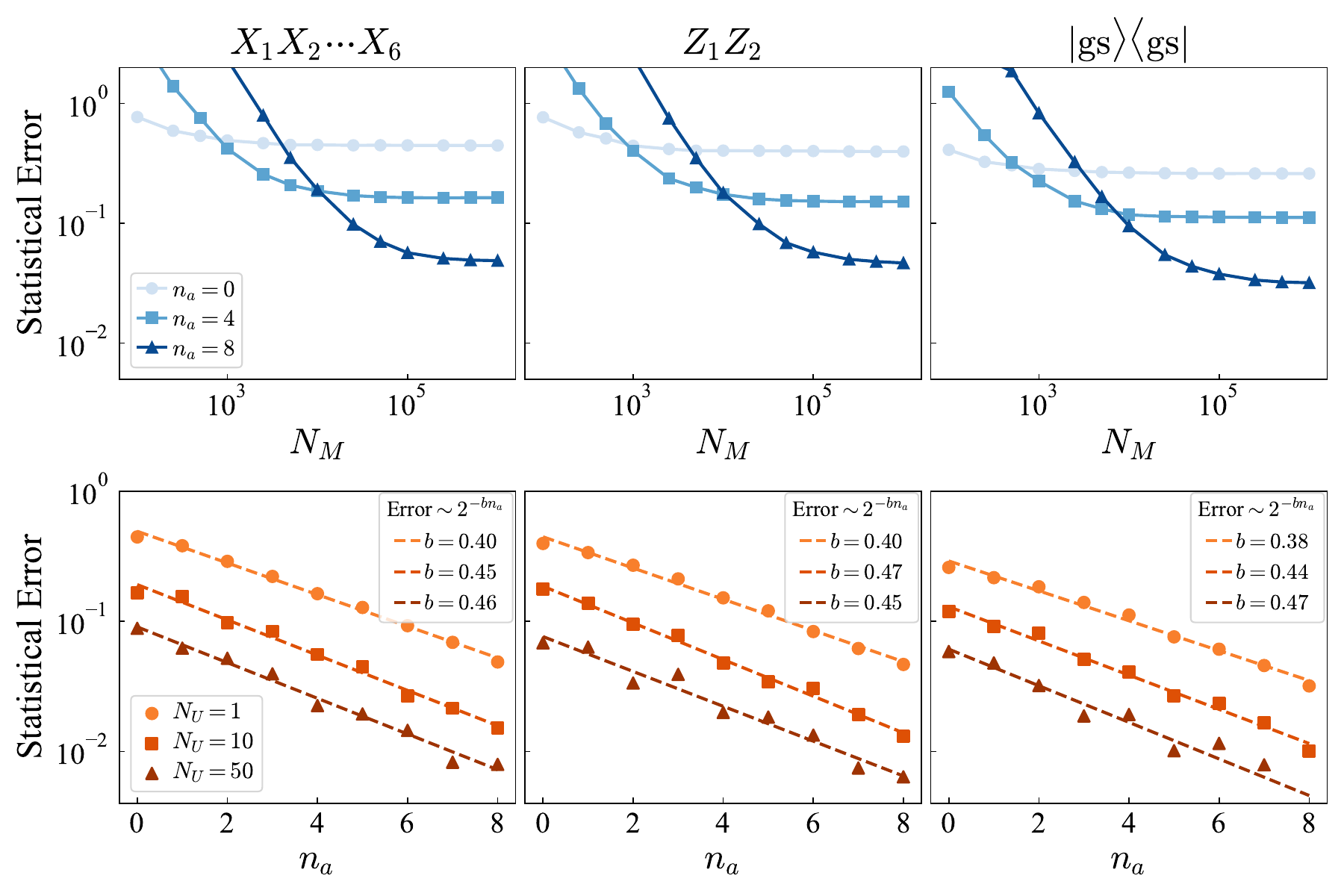}
\caption{Average statistical error in estimating \(\Tr(O\rho^3)\) via CBNE.
Here \(\rho\) is the Gibbs state of the 6-qubit 1D TFIM with \(J=h=1.0\) and \(\beta=2.0\).
The observables are Pauli operators \(O=X_1X_2\cdots X_6\), \(O=Z_1Z_2\), and the ground state of the TFIM $|\mathrm{gs}\>\<\mathrm{gs}|$. 
The upper plots show the error scaling with \(N_M\) for \(N_U=1\) and different numbers of ancillary qubits \(n_a\).
The lower plots show the error scaling with \(n_a\) for fixed \(N_M=10^6\) and different values of \(N_U\).
}
\label{fig:ancilla}
\end{figure}

\begin{figure}
\centering
\includegraphics[width=0.7\textwidth]{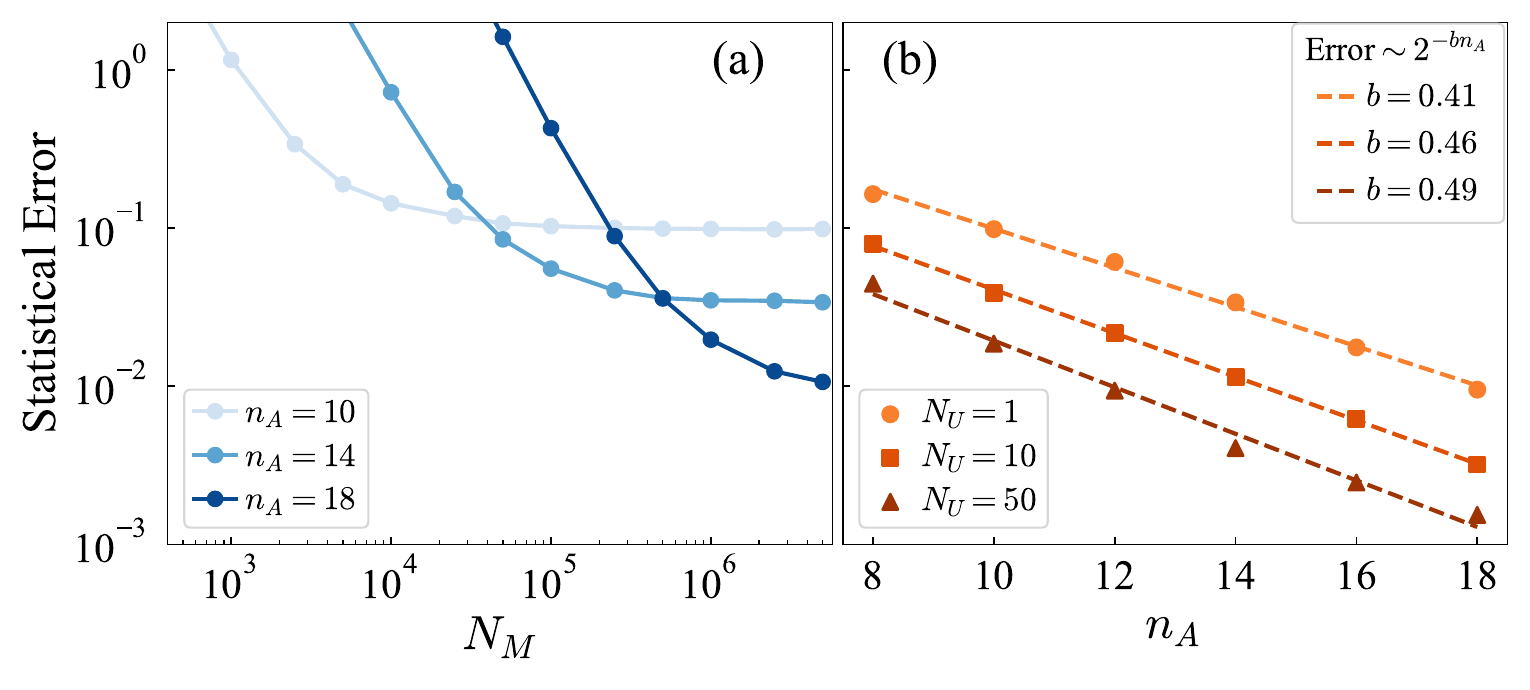}
\caption{Average statistical error in estimating \(p_4^{\PT}\) via PTME.   Here \(\rho=(1-p)|\mathrm{gs}\rangle\langle\mathrm{gs}|+pI/d\), with \(p=0.3\), and \(|\mathrm{gs}\rangle\) is the ground state of the 1D TFIM with \(J=h=1.0\).
The size of subsystem \(B\) is fixed at \(n_B=2\).
(a) Error scaling versus \(N_M\) for \(N_U=1\) and different subsystem sizes \(n_A\).
(b) Error scaling versus \(n_A\) for fixed \(N_M=10^7\) and different \(N_U\). }
\label{fig:SE_nA}
\end{figure}

We first test this prediction for CBNE. 
The target state $\rho$ is chosen as a Gibbs state of the six-qubit 1D TFIM, and we estimate nonlinear quantities \(\Tr(O\rho^3)\) for three observables:
\[
X_1X_2X_3X_4X_5X_6,\qquad Z_1Z_2I_3I_4I_5 I_6,\qquad |\mathrm{gs}\rangle\langle \mathrm{gs}|,
\]
where $X_i$ and $Z_i$ denote the Pauli-$X$ and Pauli-$Z$ operators acting on the $i$-th qubit, 
and \(|\mathrm{gs}\rangle\) denotes the TFIM ground state. 
For each value of \(n_a\), we append \(n_a\) ancillary qubits initialized in the fixed pure state \(\ket{0^{n_a}}\bra{0^{n_a}}\) and run CBNE with a single measurement setting, \(N_U=1\). 
As shown in Fig.~\ref{fig:ancilla}, the statistical error decreases exponentially with \(n_a\). 
This agrees with the theoretical prediction above, confirming that enlarging the effective Hilbert space through ancillary qubits enables high-precision estimation with a single measurement setting. 
Moreover, the estimates for all three observables in Fig.~\ref{fig:ancilla} are obtained by postprocessing the same measurement data, demonstrating that CBNE can simultaneously estimate multiple distinct observables.

We next consider PTME, for which
\tref{thm:PTResult} in the main text  predicts a similar dimension-dependent suppression. 
For estimating PT moments, the number of measurement settings scales as $N_U=\max\left\{1,\,c\,d_B/(d_A\epsilon^2) \right\}$, 
where \(d_A\) and \(d_B\) are the dimensions of subsystems $A$ and $B$, respectively. 
Hence, for fixed \(d_B\), increasing \(d_A\) or adding ancillary qubits to subsystem $A$ reduces the setting complexity and suppresses the single-setting statistical error.
We verify this behavior numerically using depolarized TFIM ground states. 
We fix subsystem \(B\) to contain \(n_B=2\) qubits and vary the size \(n_A\) of subsystem \(A\). 
As shown in Fig.~\ref{fig:SE_nA}, the statistical error for estimating the fourth order PT moment $p_4^{\PT}$ decreases exponentially as \(n_A\) increases, confirming the error suppression predicted by our theory.

\subsection{PTME for entanglement detection}
In this subsection, we demonstrate how PTME can be used for entanglement detection through PT-moment witnesses.

Entanglement is one of the central resources in quantum information science~\cite{Horodecki2009Entanglement}. It enables the advantage of quantum communication over classical strategies. Therefore, entanglement detection is a fundamental problem in quantum resource theory. Among the various criteria developed, the Positive Partial Transposition (PPT) criterion stands out as one of the most important entanglement detection criteria~\cite{Peres1996PPT}. Suppose an $n$-qubit state $\rho$ on Hilbert space $\caH = \caH_A \otimes \caH_B$. The PPT criterion states that for all separable states, $\rho^{\top_B}$ remains positive, while for entangled states, $\rho^{\top_B}$ may not be. 
Thus, the positivity of $\rho^{\top_B}$ can be used as a criterion for detecting entanglement. It has been shown that, for all pure states and low-dimensional mixed states ($2 \times 2$ and $2 \times 3$), PPT provides a necessary and sufficient condition for separability~\cite{Horodecki1996PPT}.

Although the PPT criterion provides a clear and straightforward method for entanglement detection, determining whether $\rho^{\top_B}$ is positive requires precise estimation of all its eigenvalues, which is challenging in practical scenarios. 
As a result, a series of entanglement criteria based on partial transpose (PT) moments have emerged. 
While these criteria are generally weaker than the PPT criterion, they offer a simpler implementation and can be extracted through techniques such as randomized measurements~\cite{elben2020mixedstate,singlezhou,PhysRevX.14.031035,PhysRevLett.121.150503,PhysRevLett.94.040502,PhysRevA.99.062309}.

A class of PT-moment criteria is based on the Stieltjes moment problem~\cite{Yu2021Entanglement}. 
For \(k\ge0\), define the \((k+1)\times(k+1)\) Hankel matrix
\begin{equation}
\big[B_k(\rho^{\top_B})\big]_{i,j}
=
p^{\PT}_{i+j+1},
\qquad i,j=0,\ldots,k .
\end{equation}
If \(\rho\) is separable, then the moment sequence \(\{p_t^{\PT}\}_{t=1}^m\) must satisfy
$B_{\lfloor (m-1)/2 \rfloor}(\rho^{\top_B})\geq0$. 
A weaker but simpler necessary condition is
\begin{equation}
\left|B_{\lfloor (m-1)/2 \rfloor}(\rho^{\top_B})\right| := \det\left[B_{\lfloor (m-1)/2 \rfloor}\left(\rho^{\top_B}\right)\right] \geq 0 .
\end{equation}
Since \(B_k\) is the leading principal submatrix of \(B_{k+1}\), the condition \(B_{k+1}\geq0\) gives a stronger entanglement criterion than \(B_k\geq0\). 
For \(k=1\), using \(p_1^{\PT}=1\), the positivity condition \(B_1(\rho^{\top_B})\geq0\) implies
$p_3^{\PT}-(p_2^{\PT})^2\ge0$.
Therefore, a violation,
$(p_2^{\PT})^2-p_3^{\PT}>0$,
certifies entanglement; this is the well-known \(p_3^{\PT}\)-PPT criterion~\cite{elben2020mixedstate}.

We next present a numerical experiment in which our PTME protocol is used to estimate PT moments and construct the Hankel-matrix witnesses for entanglement detection. 
The states considered here are Gibbs states of the TFIM. 
At low temperature, i.e., large \(\beta\), the Gibbs state approaches the TFIM ground state and is typically entangled, whereas at high temperature, i.e., small \(\beta\), it approaches the maximally mixed state and becomes separable. 
As shown in Fig.~\ref{fig:Gibbs_Entdet}, the Hankel criterion based on \(B_1(\rho^{\top_B})\) fails to detect entanglement for certain intermediate values of \(\beta\). 
In contrast, the higher-order criterion based on \(B_2(\rho^{\top_B})\), which uses PT moments up to order five, can detect entanglement in this regime. 
This illustrates the advantage of accessing higher-order PT moments. 
At the same time, higher-order criteria are more sensitive to statistical noise, since they involve higher-order moments and higher-degree polynomial functions of the estimated moments. 
PTME is useful in this setting because the same measurement data can be reused to estimate all required PT moments, allowing one to combine lower-order and higher-order criteria to improve the overall detection performance.

\begin{figure}[t]
\centering
\includegraphics[width=0.5\textwidth]{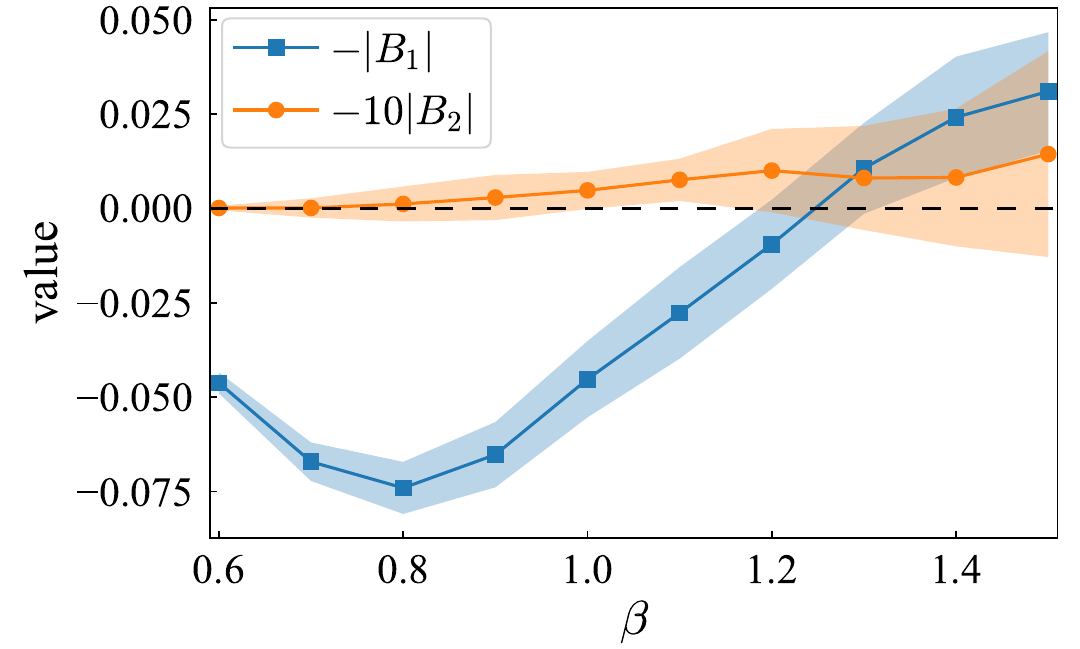}
\caption{Values of the entanglement-detection criteria constructed from \(B_k(\rho^{\top_B})\) as functions of the inverse temperature \(\beta\), where positive values detect entanglement.
In the plot, the matrix \(B_k(\rho^{\top_B})\) is abbreviated as \(B_k\), and $|B_k|:=\det(B_k)$.
Here \(\rho\) is a 12-qubit Gibbs state of the TFIM with \(J=1\) and \(h=2.5\).
Subsystem \(A\) consists of the first 11 qubits, and subsystem \(B\) is the remaining one qubit.
Each point is obtained by averaging the PTME estimates over 1000 independent trials with \(N_U=100\) and \(N_M=10^6\).
The shaded regions is plotted according to the standard deviation.}
\label{fig:Gibbs_Entdet}
\end{figure}

Another class of PT-moment criteria is given by the \(D_k\) witness hierarchy~\cite{Neven2021Entanglement}. 
Let \(e_k(x_1,\ldots,x_d)\) denote the \(k\)-th elementary symmetric polynomial,
\begin{equation}
e_k(x_1,\ldots,x_d)
:=
\sum_{1\le j_1<\cdots<j_k\le d}
x_{j_1}\cdots x_{j_k},
\qquad k=1,\ldots,d,
\qquad e_0(x_1,\ldots,x_d)=1. 
\end{equation}
Let \(\lambda_1,\ldots,\lambda_d\) be the eigenvalues of \(\rho^{\top_B}\). 
The criterion states that, if \(\rho\) is separable, then  
$e_k(\lambda_1,\ldots,\lambda_d)\ge0$
for all \(k=1,\ldots,d\). 
Therefore, a negative value of \(e_k(\lambda_1,\ldots,\lambda_d)\) indicates the existence of entanglement. 
Using Newton's identities, each \(e_k(\lambda_1,\ldots,\lambda_d)\) can be expressed as a polynomial in the first \(k\) PT moments \(p_1^{\PT},\ldots,p_k^{\PT}\). 
Define
\[
D_k := -k\,e_k(\lambda_1,\ldots,\lambda_d),
\]
so that \(D_k>0\) detects entanglement \cite{Neven2021Entanglement}. 
The first few polynomials are
\begin{align}
D_1 &= -p^{\PT}_1,\nonumber\\
D_2 &=  p^{\PT}_2-\left(p^{\PT}_1\right)^2 ,\nonumber\\
D_3 &= -p^{\PT}_3 + \frac{3}{2}p^{\PT}_1p^{\PT}_2 - \frac{1}{2} \left(p^{\PT}_1\right)^3,\nonumber\\
D_4 &=p^{\PT}_4-\frac{1}{2}\left[\left(p^{\PT}_1\right)^2-p^{\PT}_2\right]^2
+\frac{1}{3}\left(p^{\PT}_1\right)^4
-\frac{4}{3}p^{\PT}_1p^{\PT}_3,\nonumber\\
D_5 &= -p^{\PT}_5 
+\frac{5}{4}p^{\PT}_1p^{\PT}_4 
+\frac{5}{6}p^{\PT}_2p^{\PT}_3
-\frac{5}{6}\left(p^{\PT}_1\right)^2p^{\PT}_3
-\frac{5}{8}p^{\PT}_1\left(p^{\PT}_2\right)^2 
+\frac{5}{12}\left(p^{\PT}_1\right)^3p^{\PT}_2
-\frac{1}{24}\left(p^{\PT}_1\right)^5 .
\end{align}

\begin{figure}[t]
\centering
\includegraphics[width=0.5\textwidth]{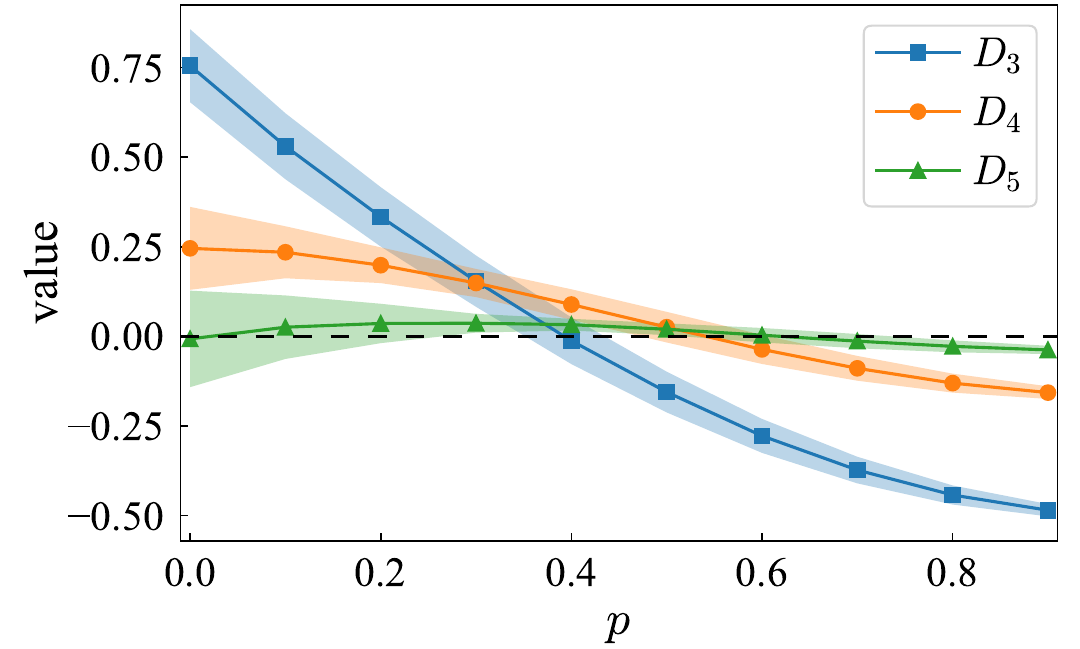}
\caption{Values of the \(D_k\) entanglement witnesses estimated by PTME as functions of the depolarizing probability \(p\), where positive values detect entanglement. 
Here \(\rho=(1-p)|\mathrm{gs}\rangle\langle\mathrm{gs}|+pI/d\), and \(|\mathrm{gs}\rangle\) is the ground state of the 12-qubit Heisenberg model with \(J=1\). 
Subsystem \(A\) consists of the first 11 qubits, and subsystem \(B\) is the remaining qubit. 
Each point is obtained by averaging over 1000 independent PTME trials with \(N_U=1\) and \(N_M=10^6\). 
The shaded regions indicate the standard error.
}
\label{fig:Depol_Entdet}
\end{figure}

Next, we demonstrate the use of PTME for entanglement detection with the \(D_k\) witness hierarchy. 
We consider the depolarized state
$\rho=(1-p)|\mathrm{gs}\rangle\langle\mathrm{gs}|+pI/d$,
where \(|\mathrm{gs}\rangle\) is the ground state of the one-dimensional Heisenberg model with open boundary conditions, whose Hamiltonian is written as
\begin{equation}
H=-J\sum_{i=1}^{n-1}
\left(X_iX_{i+1}+Y_iY_{i+1}+Z_iZ_{i+1}\right).
\end{equation}
Figure~\ref{fig:Depol_Entdet} shows that higher-order  witnesses can detect entanglement in certain regimes where lower-order witnesses fail. 
PTME is useful in this setting because it estimates the required PT moments simultaneously from a single measurement setting, allowing one to combine $D_k$ witnesses of different orders and thereby detect a broader range of entangled states.

\section{Sample lower bound for state moment estimation with a single measurement setting}
\label{cbne:sec:single-basis-lower-bound}

In this section, we prove the sample lower bound for state-moment estimation stated in Theorem~\ref{thm:PrinEigen} of the main text. Specifically, we show that estimating the $t$-th moment $p_t=\Tr(\rho^t)$ from a single measurement basis requires $\Omega(d^{1-1/t})$ copies of $\rho$ for any fixed $t\geq2$ and sufficiently small constant accuracy.

\subsection{Measurement model and proof overview}

We consider protocols that measure $N$ copies of an arbitrary unknown $d$-dimensional state $\rho$ in the same orthonormal basis. The basis may be chosen deterministically or at random, but its choice is independent of $\rho$. The estimator has access to the chosen basis and all measurement outcomes, and may perform arbitrary classical postprocessing and use additional classical randomness. This single-basis measurement model is a restricted subclass of non-adaptive single-copy measurement protocols; more general non-adaptive protocols may, for example, use different measurement settings across different copies.

\begin{theorem}[Restatement of Theorem~\ref{thm:PrinEigen}]
\label{cbne:thm:single-basis-lower-bound}
Fix an integer $t\geq2$ and $0<\epsilon<(2t)^{1-t}$. There exists a threshold $d_{t,\epsilon}\geq t$ such that, for every $d\geq d_{t,\epsilon}$, any protocol in the above measurement model that outputs an estimate $\hat{p}_t$ satisfying 
\begin{equation}\label{cbne:eq:single-basis-accuracy}
\Pr\{|\hat{p}_t-\Tr(\rho^t)|\leq \epsilon\}\ge \frac{2}{3}
\end{equation}
for any $d$-dimensional quantum state $\rho$ must use
$N>12^{-1/t}d^{1-1/t}$ copies. The probability is taken over the random choice of basis, the measurement outcomes, and any classical randomness used by the estimator.
\end{theorem}

The proof proceeds by reducing moment estimation to a binary hypothesis-testing problem. We construct two ensembles of random states whose $t$-th moments concentrate around two values separated by more than $2\epsilon$, while their measurement-data distributions remain close in \emph{total variation} (TV) distance when $N\ll d^{1-1/t}$. On the other hand, concentration of their $t$-th moments implies that any estimator achieving additive error $\epsilon$ with sufficiently high probability would distinguish the two ensembles with high probability. This contradicts the TV bound and yields the claimed $\Omega(d^{1-1/t})$ sample lower bound for moment estimation.

For the discussion below, we write the chosen basis as $\caB=\{|u_i\rangle\}_{i=1}^d$ and denote the measurement outcomes by $Y_1,\ldots,Y_N$. Conditional on $\rho$ and $\caB$, the measurements on the $N$ state copies are independent, with
\begin{equation}\label{cbne:eq:single-basis-outcomes}
 \Pr\{Y_1=y_1,\ldots,Y_N=y_N\,|\,\rho,\caB\}
 =\prod_{\ell=1}^N\langle u_{y_\ell}|\rho|u_{y_\ell}\rangle.
\end{equation}
Denote the two state ensembles by $\mathcal E_{b,N}^{(d)}$, where $b\in\{0,1\}$. The ensembles are allowed to depend on $N$, since a performance guarantee that holds for every fixed state must also hold after averaging over any state ensemble chosen for that sample size. For a fixed basis $\caB$, let $P_{b,N}^{\caB}$ denote the distribution of the $N$ ordered outcomes when a state is drawn once from $\mathcal E_{b,N}^{(d)}$ and then used for all $N$ measurements:
\begin{equation}\label{cbne:eq:single-basis-data-distribution}
 P_{b,N}^{\caB}(y_1,\ldots,y_N)
 =\E_{\sigma\sim\mathcal E_{b,N}^{(d)}}
   \prod_{\ell=1}^N\langle u_{y_\ell}|\sigma|u_{y_\ell}\rangle.
\end{equation}
The expectation lies outside the product because the same sampled state $\sigma$ is used for all copies. Finally, let $P_{b,N}$ denote the joint distribution of $(\caB,Y_1,\ldots,Y_N)$, including any randomness in the basis selection. For probability distributions $P$ and $Q$, we write their TV distance as
\[
\operatorname{TV}(P,Q)=\sup_A |P(A)-Q(A)|,
\]
where the supremum is over measurable events $A$.

\subsection{Eigenvalue lists with matching lower moments}

We first construct two eigenvalue lists whose power sums agree up to order \(t-1\) but differ at order \(t\), providing the moment matching and separation needed to construct the two hard state ensembles.

For $b\in\{0,1\}$ and integers $k\geq1$, define
\begin{equation}\label{cbne:eq:single-basis-spectra}
 \lambda_j^{(b)}=\frac{1+\cos((2j+b)\pi/t)}{t},
 \quad j=0,\ldots,t-1;
 \qquad\qquad p_k^{(b)}=\sum_{j=0}^{t-1}(\lambda_j^{(b)})^k,\quad k\geq1.
\end{equation}

\begin{lemma}[Matching spectral moments]
\label{cbne:lem:single-basis-spectra}
Both vectors $\lambda^{(b)}$ are nonnegative and sum to one. Their power sums satisfy
\begin{align}
 p_k^{(0)}=p_k^{(1)}
 &=t^{1-k}2^{-k}\binom{2k}{k},\qquad 1\leq k<t,
 \label{cbne:eq:single-basis-matching-powers}\\
 p_t^{(b)}&=t^{1-t}2^{-t}\left[\binom{2t}{t}+2(-1)^b\right],
 \qquad
 \Delta_t:=p_t^{(0)}-p_t^{(1)}=2(2t)^{1-t}>0.
 \label{cbne:eq:single-basis-moment-gap}
\end{align}
\end{lemma}
\begin{proof}
Let $z_j=\exp(\rmi(2j+b)\pi/t)$. For any integer $\ell$, the geometric sum gives
\begin{equation}
 \sum_{j=0}^{t-1}z_j^\ell
 =\begin{cases}
 0,&t\nmid\ell,\\
 t(-1)^{b\ell/t},&t\mid\ell.
 \end{cases}
\end{equation}
Since $(1+\cos((2j+b)\pi/t))^k=2^{-k}z_j^{-k}(1+z_j)^{2k}$, we obtain
\begin{equation}
 p_k^{(b)}=t^{-k}2^{-k}\sum_{a=0}^{2k}\binom{2k}{a}
                  \sum_{j=0}^{t-1}z_j^{a-k}.
\end{equation}
For $k<t$, only the exponent $a-k=0$ contributes. For $k=t$, the exponents $-t,0,t$ contribute, giving Eqs.~\eqref{cbne:eq:single-basis-matching-powers} and~\eqref{cbne:eq:single-basis-moment-gap}. Non-negativity follows from \(1+\cos\phi\ge0\) for all real \(\phi\), and the formula for $k=1$ gives normalization.
\end{proof}

The vectors $\lambda^{(b)}$ serve as eigenvalue templates; the random states constructed below need not have exactly these eigenvalues. For $t=3$, the two templates are, up to ordering, $(2/3,1/6,1/6)$ and $(1/2,1/2,0)$. They have the same second moment $1/2$, while their third moments are $11/36$ and $1/4$, respectively.

Let $\operatorname{Exp}(1)$ denote the exponential distribution with mean one. To relate these eigenvalues to measurement probabilities, take independent $E_1,\ldots,E_t\sim\operatorname{Exp}(1)$ and define
\begin{equation}\label{cbne:eq:single-basis-scalar-variable}
 X_b=\sum_{a=1}^t\lambda_{a-1}^{(b)}E_a,
 \qquad \mu_k^{(b)}=\E X_b^k\quad(k\geq0).
\end{equation}
For $z$ sufficiently close to zero, independence of the $E_a$ and the moment-generating function $\E \rme^{sE}=(1-s)^{-1}$ for $E\sim\operatorname{Exp}(1)$ give
\begin{equation}\label{cbne:eq:single-basis-scalar-mgf}
 \E\rme^{zX_b}
 =\prod_{j=0}^{t-1}(1-z\lambda_j^{(b)})^{-1}.
\end{equation}
Using $-\log(1-x)=\sum_{r\geq1}x^r/r$ and the definition $p_r^{(b)}=\sum_j(\lambda_j^{(b)})^r$, we further obtain
\begin{equation}
 \E\rme^{zX_b}
 =\exp\left(\sum_{r\geq1}\frac{p_r^{(b)}}{r}z^r\right)
 =\sum_{k\geq0}\frac{\mu_k^{(b)}}{k!}z^k,
\end{equation}
where the last equality is the Taylor expansion of the moment-generating function. The coefficient of $z^k$ in the exponential series depends only on $p_1^{(b)},\ldots,p_k^{(b)}$. Therefore, the matching relations $p_r^{(0)}=p_r^{(1)}$ for $r<t$ (\lref{cbne:lem:single-basis-spectra}) imply
\begin{equation}\label{cbne:eq:single-basis-scalar-moments}
 \mu_k^{(0)}=\mu_k^{(1)},\qquad 0\leq k<t.
\end{equation}
These are moments of the scalar variables $X_b$, rather than spectral moments of the random states.

\subsection{Weighted state ensembles and measurement probabilities}\label{sec:three_properties}

We now construct the two hard state ensembles $\mathcal E_{0,N}^{(d)}$ and $\mathcal E_{1,N}^{(d)}$. The construction is designed to satisfy two competing requirements: their $t$-th state moments remain separated, while their measurement-outcome distributions become nearly indistinguishable.

Let $G$ be a $d\times t$ random matrix with independent standard complex Gaussian entries, each with density $\pi^{-1}\rme^{-|z|^2}$. This choice is convenient for two reasons. First, different rows of $G$ are independent, which will simplify the calculation of averaged outcome probabilities. Second, the distribution of $G$ is unchanged under left multiplication by any fixed unitary, which will allow us to treat all measurement bases in the same way.
The joint density of $G$ is
\[
g(G)=\pi^{-dt}\rme^{-\Tr(G^\dagger G)},
\]
and we write $\E_G$ for expectation under this original Gaussian distribution.

For each $b\in\{0,1\}$, define
\begin{equation}\label{cbne:eq:single-basis-weighted-construction}
 \Lambda_b=\diag{\lambda_0^{(b)},\ldots,\lambda_{t-1}^{(b)}},\qquad
 S_b=G\Lambda_bG^\dagger/d,\qquad
 Z_b=\Tr(S_b),\qquad
 \sigma_b=S_b/Z_b.
\end{equation}
Set $C_{b,N}=\E_G Z_b^N$. We define the ensemble $\mathcal E_{b,N}^{(d)}$ by sampling $G$ according to the weighted density
\[
\widetilde g_{b,N}(G)
=
\frac{Z_b(G)^N g(G)}{C_{b,N}}
\]
and outputting the state $\sigma_b(G)$. 
The factor $Z_b^N$ is chosen to cancel the factor $Z_b^{-N}$ arising from the normalization of $\sigma_b$ in the probability of an $N$-outcome measurement string. These ensembles are well defined: \lref{cbne:lem:single-basis-spectra} gives $\Lambda_b\geq 0$ with $\Tr\Lambda_b=1$, and the Gaussian distribution gives $Z_b>0$ almost surely. Gaussian entries have moments of every order, so $0<C_{b,N}<\infty$. Thus $\sigma_b$ is a density operator and each ensemble has total probability one. All $N$ outcomes are measured from the same state $\sigma_b(G)$.

Write $\widetilde{\Pr}_{b,N}$ and $\widetilde{\E}_{b,N}$ for the weighted law of $G$ and the corresponding expectation. For every bounded measurable function $f$,
\begin{equation}\label{cbne:eq:single-basis-trace-weight}
 \widetilde{\E}_{b,N}f(G)
 =\frac{\E_G[Z_b^N f(G)]}{C_{b,N}}.
\end{equation}
The same identity holds for any nonnegative measurable $f$. The two ensembles have three properties that will be used later:

\emph{(1) Unitary invariance of the ensembles.}
For any fixed unitary $U$, the Gaussian distribution is invariant under $G\mapsto UG$. The weight is also unchanged because $Z_b(UG)=Z_b(G)$, while $\sigma_b(UG)=U\sigma_b(G)U^\dagger$. Thus each ensemble is unchanged by unitary conjugation, and every fixed measurement basis gives the same averaged outcome distribution. It therefore suffices to analyze the computational basis $\caB_0=\{|i\rangle\}_{i=1}^d$.

\emph{(2) Matching outcome distributions without \(t\)-fold collisions.}
Under the original Gaussian distribution, each entry satisfies $|G_{ia}|^2\sim\operatorname{Exp}(1)$. Define
\begin{equation}\label{cbne:eq:single-basis-row-variables}
 X_i^{(b)}=\sum_{a=1}^t\lambda_{a-1}^{(b)}|G_{ia}|^2,
 \qquad
 Z_b=\frac1d\sum_{i=1}^dX_i^{(b)}.
\end{equation}
The variables $X_i^{(b)}$ are independent copies of $X_b$ (defined in \eref{cbne:eq:single-basis-scalar-variable}) under the original Gaussian distribution. For fixed $G$, the probability of measurement outcome $i$ is
$\langle i|\sigma_b(G)|i\rangle= X_i^{(b)}/(dZ_b)$.
For an ordered string $y=(y_1,\ldots,y_N)$, let $c_i(y)$ denote the number of occurrences of label $i$, so that $\sum_i c_i(y)=N$. Then
\begin{align}\label{cbne:eq:single-basis-weighted-outcomes}
 P_{b,N}^{\caB_0}(y)
 =\frac1{C_{b,N}}\E_G\left[
 Z_b^N\prod_{i=1}^d\left(\frac{X_i^{(b)}}{dZ_b}\right)^{c_i(y)}\right]
 =\frac{\E_G\prod_{i=1}^d(X_i^{(b)})^{c_i(y)}}{d^NC_{b,N}}
 =\frac{\prod_{i=1}^d\mu_{c_i(y)}^{(b)}}{d^NC_{b,N}}.
\end{align}
The second equality uses $\sum_i c_i(y)=N$ to cancel the factor $Z_b^N$, while the third uses row independence under the original Gaussian distribution. No row independence is assumed under the weighted distribution. By \eref{cbne:eq:single-basis-scalar-moments}, whenever every $c_i(y)<t$, the two probabilities $P_{0,N}^{\caB_0}(y)$ and  $P_{1,N}^{\caB_0}(y)$ differ only by the constant factor $C_{1,N}/C_{0,N}$, independent of $y$. This observation will allow us to bound their total variation distance by the probability of a $t$-fold collision.

\emph{(3) Separation of the $t$-th state moments.}
By \lref{cbne:lem:single-basis-spectra},
\[
p_t^{(0)}-p_t^{(1)}=\Delta_t>2\epsilon
\]
in the accuracy range of the theorem. For fixed $t$, \lref{cbne:lem:single-basis-concentration} below shows that, as $d\to\infty$ with $N/d\to0$, the $t$-th moments of states drawn from the two ensembles concentrate around $p_t^{(0)}$ and $p_t^{(1)}$, respectively. This condition holds throughout the relevant regime $N=\mathcal{O}(d^{1-1/t})$.

\subsection{Concentration of the state moment}

The weighted states need not have eigenvalues $\lambda^{(b)}$ exactly. We now show that, for fixed $t$, their $t$-th moments concentrate around $p_t^{(b)}$ as $d\to\infty$ with $N/d\to0$. This condition holds in the sample regime considered in the lower-bound proof, where $N\leq12^{-1/t}d^{1-1/t}$ implies $N/d\leq12^{-1/t}d^{-1/t}\to0$ for fixed $t$. We use the operator and Frobenius norms $\|\cdot\|$ and $\|\cdot\|_2$ defined in Sec.~\ref{sec:Prelimi}.

\begin{lemma}[Concentration under the weighted ensemble]
\label{cbne:lem:single-basis-concentration}
For either $b\in\{0,1\}$, integers $d\geq t\geq2$ and $N\geq0$, and $0<\eta\leq1/2$,
\begin{equation}\label{cbne:eq:single-basis-moment-concentration}
 \widetilde{\Pr}_{b,N}\left\{
 |\Tr(\sigma_b^t)-p_t^{(b)}|>4t\eta\right\}
 \leq\frac1{\eta^2}\left(\frac{(t+1)^2}{d}+\frac{2N^2}{d^2}\right).
\end{equation}
\end{lemma}

\begin{proof}
Fix $b$ and abbreviate $\lambda_a=\lambda_{a-1}^{(b)}$ for $1\leq a\leq t$, $\Lambda=\Lambda_b$, and $Z=Z_b$. Let $T_a=\|G_{\cdot a}\|^2$ be the squared length of column $a$. Under the original Gaussian distribution, the columns are independent and unitarily invariant, with $T_a\sim\Gamma(d,1)$; their unit directions are independent of all lengths~\cite{harrow2013church}. The distribution $\Gamma(s,1)$ has density proportional to $x^{s-1}\rme^{-x}$ for $x>0$, with mean and variance $s$.

Since $dZ=\sum_a\lambda_aT_a$, expanding the weight $(dZ)^N$ produces nonnegative terms proportional to $\prod_a T_a^{k_a}$ with $\sum_a k_a=N$. Multiplying a $\Gamma(d,1)$ density by $x^{k_a}$ and normalizing gives a $\Gamma(d+k_a,1)$ density. Thus the weighted column law is a convex mixture whose components satisfy
\begin{equation}\label{cbne:eq:single-basis-gamma-mixture}
 T_a\mid\mathbf k\sim\Gamma(d+k_a,1),\qquad
 \mathbf k\in\bbz_{\geq0}^t,\quad\sum_a k_a=N.
\end{equation}
Within each component the columns are independent and remain unitarily invariant, because the weight depends only on their lengths.

Let $H=G^\dagger G/d$. For a fixed component, write $\E_{\mathbf k}$ for expectation. Since $H_{aa}=T_a/d$, the Gamma mean and variance give
\begin{equation}
 \E_{\mathbf k}(H_{aa}-1)^2=\frac{d+k_a+k_a^2}{d^2}.
\end{equation}
Unitary invariance gives $\E_{\mathbf k}G_{\cdot a}G_{\cdot a}^\dagger=(d+k_a)\openone_d/d$. Hence column independence implies, for $a\ne c$,
\begin{equation}
 \E_{\mathbf k}|H_{ac}|^2
 =\frac1{d^2}\Tr\left[
 (\E_{\mathbf k}G_{\cdot a}G_{\cdot a}^\dagger)
 (\E_{\mathbf k}G_{\cdot c}G_{\cdot c}^\dagger)\right]
 =\frac{(d+k_a)(d+k_c)}{d^3}.
\end{equation}
Summing these entries and using $\sum_a k_a=N$ yields
\begin{align}\label{cbne:eq:single-basis-gram-second-moment}
 \E_{\mathbf k}\|H-\openone_t\|_2^2
 &=\frac{t^2}{d}+\frac{(2t-1)N}{d^2}
   +\frac{\sum_a k_a^2+d^{-1}\sum_{a\ne c}k_ak_c}{d^2}\nonumber\\
 &\leq\frac{t^2}{d}+\frac{(2t-1)N}{d^2}+\frac{N^2}{d^2}
 \leq\frac{(t+1)^2}{d}+\frac{2N^2}{d^2}.
\end{align}
The last inequality uses $2tN\leq N^2+t^2$ and $t\leq d$. The bound is uniform in $\mathbf k$, so averaging the mixture and applying Markov's inequality gives
\begin{equation}\label{cbne:eq:single-basis-gram-concentration}
 \widetilde{\Pr}_{b,N}\{\|H-\openone_t\|_2>\eta\}
 \leq\frac1{\eta^2}\left(\frac{(t+1)^2}{d}+\frac{2N^2}{d^2}\right).
\end{equation}

To relate this bound to the state moment, suppose $\|H-\openone_t\|_2\leq\eta\leq1/2$. Since $\|H-\openone_t\|\leq\|H-\openone_t\|_2$,
\begin{equation}
 (1-\eta)\Lambda\leq\Lambda^{1/2}H\Lambda^{1/2}
 \leq(1+\eta)\Lambda.
\end{equation}
Taking traces gives $1-\eta\leq Z\leq1+\eta$. The matrices $G\Lambda G^\dagger/d$ and $\Lambda^{1/2}H\Lambda^{1/2}$ have the same nonzero eigenvalues. Let $\alpha_1,\ldots,\alpha_t$ be the nonzero eigenvalues of $\sigma_b$, padded with zeros to $t$ entries. Ordering both $\alpha_a$ and $\lambda_a$ nonincreasingly, eigenvalue monotonicity under the positive semidefinite order gives
\begin{equation}
 \frac{1-\eta}{Z}\lambda_a\leq\alpha_a
 \leq\frac{1+\eta}{Z}\lambda_a,\qquad 1\leq a\leq t.
\end{equation}
Consequently,
\begin{equation}
 \sum_a|\alpha_a-\lambda_a|
 \leq\frac{\eta+|Z-1|}{Z}\sum_a\lambda_a
 \leq\frac{2\eta}{1-\eta}\leq4\eta.
\end{equation}
Since $\alpha_a,\lambda_a\in[0,1]$, the scalar inequality $|x^t-y^t|\leq t|x-y|$ on $[0,1]$ implies
\begin{equation}
 |\Tr(\sigma_b^t)-p_t^{(b)}|
 =\left|\sum_a\alpha_a^t-\sum_a\lambda_a^t\right|
 \leq t\sum_a|\alpha_a-\lambda_a|\leq4t\eta.
\end{equation}
Combining this implication with \eref{cbne:eq:single-basis-gram-concentration} proves the lemma.
\end{proof}

\subsection{Indistinguishability from the collision probability}

We now show that the two hard state ensembles are nearly indistinguishable from the measurement data. The key observation is that, conditioned on the absence of a $t$-fold collision, the two outcome distributions are identical. Consequently, their distinguishability can be controlled by the probability of a $t$-fold collision.
For an ordered outcome string $y=(y_1,\ldots,y_N)$, let $c_i(y)$ denote the number of occurrences of label $i$, and define
\[
\A=\{y:\max_i c_i(y)<t\}.
\]
Thus, $\A$ is the event that no $t$-fold collision occurs. On $\A$, every scalar moment $\mu_{c_i(y)}^{(b)}$ appearing in \eref{cbne:eq:single-basis-weighted-outcomes} has order strictly below $t$. As discussed in Sec.~\ref{sec:three_properties}, the two conditional outcome distributions are identical whenever no $t$-fold collision occurs. Their total variation distance can therefore be bounded by the probability of the complementary event $\A^c$, namely the occurrence of a $t$-fold collision.

\begin{lemma}[Quantitative indistinguishability]
\label{cbne:thm:single-basis-indistinguishability}
For integers $d\geq t\geq2$ and $N\geq0$, let $\caB$ be any fixed orthonormal basis of $\bbc^d$. The weighted ensembles $\mathcal E_{b,N}^{(d)}$ satisfy
\begin{equation}\label{cbne:eq:single-basis-tv-bound}
 \operatorname{TV}(P_{0,N}^{\caB},P_{1,N}^{\caB})
 \leq\frac{N(N-1)\cdots(N-t+1)}{(d+1)\cdots(d+t-1)}
 \leq\frac{N^t}{d^{t-1}}.
\end{equation}
\end{lemma}

\begin{proof}
As discussed in Sec.~\ref{sec:three_properties}, the state ensembles are unitarily invariant. Hence the averaged outcome distribution is the same for every orthonormal basis, and it suffices to prove the bound for the computational basis $\caB_0=\{|i\rangle\}_{i=1}^d$.

For distributions on the finite outcome space $[d]^N$, 
$\operatorname{TV}(P,Q)
=1-\sum_y\min\{P(y),Q(y)\}$.
By \eref{cbne:eq:single-basis-weighted-outcomes}, for every $y\in\A$, we have
$P_{0,N}^{\caB_0}(y)/P_{1,N}^{\caB_0}(y)
=C_{1,N}/C_{0,N}$. Thus, whenever $\A$ is nonempty, the two distributions have the same conditional distribution on $\A$. Writing this common conditional distribution as $q_\A$, we have
\[
P_{b,N}^{\caB_0}(y)
=P_{b,N}^{\caB_0}(\A)\,q_\A(y),
\qquad y\in\A.
\]
Therefore,
\begin{align}
\sum_{y\in\A}
\min\{P_{0,N}^{\caB_0}(y),P_{1,N}^{\caB_0}(y)\}
=
\sum_{y\in\A}
q_\A(y)\,
\min_bP_{b,N}^{\caB_0}(\A)
=
\min_bP_{b,N}^{\caB_0}(\A).
\end{align}
Consequently,
\begin{align}\label{cbne:eq:single-basis-collision-reduction}
 \operatorname{TV}(P_{0,N}^{\caB_0},P_{1,N}^{\caB_0})
 \leq
 1-\min_bP_{b,N}^{\caB_0}(\A)
 =\max_bP_{b,N}^{\caB_0}(\A^c).
\end{align}
If $\A$ is empty, then $P_{b,N}^{\caB_0}(\A^c)=1$, so the same inequality follows trivially from $\operatorname{TV}\leq1$.

It remains to bound the probability of a $t$-fold collision. Condition on a fixed state $\sigma$ and write
$q_i=\langle i|\sigma|i\rangle$.
The $N$ measurement outcomes are then independent with probabilities $\{q_i\}_{i=1}^d$. The event $\A^c$ occurs exactly when there exists a set of $t$ measurement positions whose outcomes are identical. For any specified set of $t$ positions, the probability that all corresponding outcomes coincide is
$\sum_{i=1}^d q_i^t$.
Hence, by a union bound over the $\binom Nt$ choices of $t$ positions,
\begin{equation}\label{cbne:eq:single-basis-collision-union}
 P_{b,N}^{\caB_0}(\A^c)
 \leq
 \binom Nt\,
 \E_{\sigma\sim\mathcal E_{b,N}^{(d)}}
 \sum_{i=1}^d\langle i|\sigma|i\rangle^t.
\end{equation}
For $N<t$, no $t$-fold collision is possible, so the right-hand side is zero and the desired bound follows immediately. We therefore assume $N\geq t$ below.

Let $U$ be a Haar-random unitary independent of $\sigma$. Since $\mathcal E_{b,N}^{(d)}$ is unitarily invariant, $U\sigma U^\dagger$ has the same distribution as $\sigma$. Hence
\begin{align}
\E_{\sigma\sim\mathcal E_{b,N}^{(d)}}
\sum_{i=1}^d\langle i|\sigma|i\rangle^t
&=
\E_{\sigma\sim\mathcal E_{b,N}^{(d)}}\E_U
\sum_{i=1}^d\langle i|U\sigma U^\dagger|i\rangle^t
=
\E_{\sigma\sim\mathcal E_{b,N}^{(d)}}
\sum_{i=1}^d  \Tr\!\Big( \sigma^{\otimes t}\E_U\left[ (U^\dag|i\>\<i|U)^{\otimes t}\right]\Big)
\nonumber\\
&=
\E_{\sigma\sim\mathcal E_{b,N}^{(d)}}
\sum_{i=1}^d 
\Tr\left( \sigma^{\otimes t} \frac{\Pi_{\mathrm{sym}}^{(t)}}{\kappa_t } \right)
\leq 
\frac{d}{\kappa_t }
=
\frac{t!}{(d+1)\cdots(d+t-1)},
\end{align}
where 
$\sym^{(t)}$ is the projector onto the symmetric subspace of $\caH^{\otimes t}$, 
$\kappa_t=\binom{t+d-1}{t}$ is the dimension of this symmetric subspace, the third equality follows from \eref{eq:PhiHaar} in Sec.~\ref{sec:Prelimi}, the inequality holds because $0\leq\Tr( \sigma^{\otimes t} \Pi_{\mathrm{sym}})\leq 1$. 
Combining this with \eref{cbne:eq:single-basis-collision-union} gives
\begin{align}
 P_{b,N}^{\caB_0}(\A^c)
 \leq
 \binom Nt
 \frac{t!}{(d+1)\cdots(d+t-1)}
 =
 \frac{N(N-1)\cdots(N-t+1)}
 {(d+1)\cdots(d+t-1)}.
\end{align}
Together with \eref{cbne:eq:single-basis-collision-reduction}, this proves the first inequality in \eref{cbne:eq:single-basis-tv-bound}. Finally, the second inequality in \eref{cbne:eq:single-basis-tv-bound} holds because
$N(N-1)\cdots(N-t+1)\leq N^t$ and 
$(d+1)\cdots(d+t-1)\geq d^{t-1}$.
This completes the proof. 
\end{proof}

\subsection{Proof of the sample lower bound}

We are now ready to prove \tref{cbne:thm:single-basis-lower-bound}.

\begin{proof}[Proof of \tref{cbne:thm:single-basis-lower-bound}]
We first show that, for fixed $t$ and $\epsilon$, the $t$-th moments of states drawn from the two ensembles $\mathcal E_{0,N}^{(d)}$ and $\mathcal E_{1,N}^{(d)}$ concentrate, with high probability, in two disjoint intervals separated by a positive constant independent of $d$. By \lref{cbne:lem:single-basis-spectra}, $p_t^{(0)}$ and $p_t^{(1)}$ depend only on $t$ and satisfy
$\Delta_t:=p_t^{(0)}-p_t^{(1)}>2\epsilon$.
Let
$\gamma= \Delta_t/2-\epsilon>0$
and
$\eta=\gamma/(8t)$.
Then $0<\eta\leq1/2$. By \lref{cbne:lem:single-basis-concentration}, the event
\begin{equation}
 \G_b=\left\{
 |\Tr(\sigma_b^t)-p_t^{(b)}|
 \leq\frac{\gamma}{2}
 \right\}
\end{equation}
has probability at least $1-\delta_d$ under the weighted ensemble for every
$0\leq N\leq12^{-1/t}d^{1-1/t}$, where
\begin{equation}\label{cbne:eq:single-basis-testing-concentration-error}
 \delta_d=
 \frac{64t^2}{\gamma^2}
 \left(
 \frac{(t+1)^2}{d}
 +2\cdot12^{-2/t}d^{-2/t}
 \right)
 \longrightarrow0
 \qquad(d\to\infty).
\end{equation}
This bound is uniform over the stated sample range for fixed $t$ and $\epsilon$. Choose $d_{t,\epsilon}\geq t$ sufficiently large that
$\delta_d<1/12$ for all $d\geq d_{t,\epsilon}$.

Now fix $d\geq d_{t,\epsilon}$ and suppose, toward a contradiction, that an estimator satisfying \eref{cbne:eq:single-basis-accuracy} uses
$N\leq12^{-1/t}d^{1-1/t}$.
Since the accuracy guarantee holds for every fixed state, averaging it over either weighted ensemble gives
\[
\Pr\left\{
|\hat p_t-\Tr(\sigma_b^t)|\leq\epsilon
\right\}
\geq\frac23.
\]
Its intersection with $\G_b$ therefore has probability at least
$2/3-\delta_d$. On this intersection,
\begin{align}
|\hat p_t-p_t^{(b)}|
&\leq
|\hat p_t-\Tr(\sigma_b^t)|
+
|\Tr(\sigma_b^t)-p_t^{(b)}|
\leq
\epsilon+\frac{\gamma}{2}
<
\frac{\Delta_t}{2},
\end{align}
where the final inequality follows from
$\gamma=\Delta_t/2-\epsilon$.

Consider the test that outputs $b=0$ if
$\hat p_t>(p_t^{(0)}+p_t^{(1)})/2$,
and outputs $b=1$ otherwise. Whenever the above intersection event occurs, this test identifies the correct hypothesis. Hence, under either hypothesis, its success probability is at least $2/3-\delta_d>7/12$. For equally likely hypotheses, its average success probability is therefore also greater than $7/12$.

On the other hand, \lref{cbne:thm:single-basis-indistinguishability} and the assumed sample bound imply that, for every fixed orthonormal basis $\caB$,
\[
\operatorname{TV}
(P_{0,N}^{\caB},P_{1,N}^{\caB})
\leq
\frac{N^t}{d^{t-1}}
\leq
\frac1{12}.
\]
The protocol chooses its measurement basis independently of the unknown state. If $\nu$ denotes the distribution of the chosen basis, then the full data distribution under hypothesis $b$ is
$P_{b,N}(\caB,y)
=\nu(\caB)P_{b,N}^{\caB}(y)$.
Therefore,
\begin{align}\label{cbne:eq:single-basis-small-tv}
\operatorname{TV}(P_{0,N},P_{1,N})
&=
\E_{\caB\sim\nu}
\operatorname{TV}
(P_{0,N}^{\caB},P_{1,N}^{\caB})
\leq\frac1{12}.
\end{align}
For equally likely hypotheses, Le Cam's two-point method~\cite{Yu1997Assouad,chen2022exponential,huang2022quantum} then implies that the success probability of any test, including randomized tests, satisfies
\begin{equation}\label{cbne:eq:single-basis-testing-success}
 p_{\mathrm{succ}}
 \leq
 \frac{1+\operatorname{TV}(P_{0,N},P_{1,N})}{2}
 \leq
 \frac{13}{24}.
\end{equation}
This contradicts the success probability greater than
$7/12>13/24$ achieved by the test constructed from the estimator. Hence
$N>12^{-1/t}d^{1-1/t}$
for all $d\geq d_{t,\epsilon}$. The prefactor depends only on $t$, while the dimension threshold may additionally depend on $\epsilon$ through $\gamma$.
\end{proof}


\begin{thebibliography}{94}
\bibitem [{com()}]{compare_PhysRevLett.108.110503}%
  \BibitemOpen
  \href@noop {} {}\bibinfo {note} {We note that state moment estimation via the relation \eqref{eq:zeta_def_prl} was previously suggested by Ref.~\cite{PhysRevLett.108.110503}. Unlike that approach, which estimates the individual probabilities $\Pr_{\rho}(b|U)$, CBNE directly estimates their power sums through collision statistics, thereby significantly reducing the resource cost; see SM Sec. S7 for details.}\BibitemShut {Stop}%
\bibitem [{eff()}]{eff_dim}%
  \BibitemOpen
  \href@noop {} {}\bibinfo {note} {In this case, the required setting number becomes $N_U=1$ because the effective system dimension satisfies $2^{n_a}d \ge c_1 \mathfrak{B}/\epsilon^2$. This comes at the cost of applying a random unitary to the enlarged system. In addition, the overall sample complexity remains $N_U N_M=\mathcal{O}(\mathfrak{B}/\epsilon^2)$.}\BibitemShut {Stop}%
\end{thebibliography}
\end{document}